\documentclass[10pt]{article}
\usepackage{tikz}

\usetikzlibrary{fit,shapes, arrows, shadows, calc, automata, positioning,backgrounds,decorations.pathmorphing,backgrounds,petri,decorations.markings,decorations.pathreplacing}

\usetikzlibrary{topaths,calc}
\usepackage{makeidx}  

\usepackage{booktabs}

\usepackage{xspace}
\usepackage{paralist}
\usepackage{amssymb}
\usepackage{amsmath}
\usepackage{amsthm}
\usepackage{mathtools}
\usepackage{wasysym}

\usepackage{wrapfig}
\usepackage{stmaryrd}
\usepackage{array}
\usepackage{fullpage}
\usepackage{hyperref}

\theoremstyle{plain}                   
\newtheorem{theorem}{Theorem}
\newtheorem{lemma}[theorem]{Lemma}
\newtheorem{proposition}[theorem]{Proposition}
\newtheorem{corollary}[theorem]{Corollary}
\newtheorem{definition}[theorem]{Definition}

\theoremstyle{definition}             

\newtheorem{example}[theorem]{Example}

\newtheorem{observation}[theorem]{Observation}

\title{Covers of Query Results}

\author{
Ahmet Kara and Dan Olteanu\\ 
Department of Computer Science, University of Oxford\\
\texttt{$\{$ahmet.kara,dan.olteanu$\}$@cs.ox.ac.uk}
}

\newcommand{\indexnat}[1]{[#1]}

\newcommand{\setindexel}[3]{\{#1\}_{#2 \in #3}}

\newcommand{\indexjoin}[3]{\Join_{#2 \in [#3]}#1}
\newcommand{\setindexjoin}[3]{\Join_{#2 \in #3}#1}
\newcommand{\takemax}[2]{\max_{#1}\{#2\}}
\newcommand{\takemin}[2]{\min_{#1}\{#2\}}
\newcommand{\sizeof}[1]{\mid\hspace{-1mm} #1 \hspace{-1mm} \mid}
\newcommand{\minus}[2]{#1\backslash \{#2\}}

\newcommand{\sign}[1]{{\cal S}(#1)}

\newcommand{\att}[1]{\mathit{att}(#1)}
\newcommand{\inst}[1]{\mathbf{#1}}

\newcommand{\calA}{\mathcal{A}}
\newcommand{\calB}{\mathcal{B}}
\newcommand{\calE}{\mathcal{E}}

\newcommand{\calJ}{\mathcal{J}}
\newcommand{\calO}{\mathcal{O}}
\newcommand{\softO}{\widetilde{\mathcal{O}}}
\newcommand{\calP}{\mathcal{P}}

\newcommand{\calS}{\mathcal{S}}
\newcommand{\calT}{\mathcal{T}}

\newcommand{\calV}{\mathcal{V}}

\newcommand{\corejoin}{\mathring{\Bowtie}}

\newcommand{\tupleofs}{\mathit{tuple}}
\newcommand{\relofs}{\mathit{rel}}

\newcommand{\relof}[1]{\mathit{rel}(#1)}
\newcommand{\tupleofind}[2]{\mathit{tuple}_{#1}(#2)}
\newcommand{\relofind}[2]{\mathit{rel}_{#1}(#2)}

\newcommand{\weight}[1]{\mathit{weight}(#1)}

\newcommand{\fhtw}[1]{\textsf{fhtw}(#1)}

\newcommand{\evo}{\textsf{EVO}}
\newcommand{\free}{\textit{free}}
\newcommand{\ext}{\textit{ext}}

\newcommand{\dom}{\textsf{dom}}
\newcommand{\Dom}{\textsf{Dom}}
\newcommand{\faqw}[1]{\textsf{faqw}(#1)}
\newcommand{\vala}{{\textsf a}}

\newcommand{\valb}{{\textsf b}}
\newcommand{\valc}{{\textsf c}}
\newcommand{\zero}{{\bf 0}}
\newcommand{\one}{{\bf 1}}

\newcommand{\TAB}{\makebox[2.5ex][r]{}}%
\newcommand{\STAB}{\makebox[1.5ex][r]{}}%
\newcommand{\LET}{\textbf{let}\xspace}%
\newcommand{\FOREACH}{\textbf{foreach}\xspace}%
\newcommand{\DO}{\textbf{do}\xspace}%
\newcommand{\RETURN}{\textbf{return}\xspace}%
\newcounter{CommentCounter}

\newcommand{\nop}[1]{}
\definecolor{goodgreen}{rgb}{0.1, 0.5, 0.1}

\begin{document}

\maketitle

\begin{abstract}
We introduce succinct lossless representations of query results called covers. They are subsets of the query results that correspond to minimal edge covers in the hypergraphs of these results.

We first study covers whose structures are given by fractional hypertree decompositions of join queries. 
For any decomposition of a query, we give asymptotically tight size bounds for the covers of the query result over that decomposition and show that such covers can be computed  in worst-case optimal time up to a logarithmic factor in the database size. For acyclic join queries, we can compute covers compositionally using query plans with a new operator called cover-join. The tuples in the query result can be enumerated from any of its covers with linearithmic pre-computation time and constant delay.

We then generalize covers from joins to functional aggregate queries that express a host of computational problems such as aggregate-join queries, in-database optimization, matrix chain multiplication, and inference in probabilistic graphical models.
 \end{abstract}

\section{Introduction}
\label{sec:introduction}

This paper introduces succinct lossless representations of query results called covers. Given a database and a join query or, more generally, a functional aggregate query (FAQ)~\cite{FAQ:PODS:2016}, a cover is a subset of the query result that, together with a (fractional hypertree) decomposition of the query~\cite{Gottlob09b}, recovers the query result. Covers enjoy desirable properties. 

First, they can be more succinct than the listing representation of the query result. For a join query $Q$, database $\inst{D}$, and a decomposition $\calT$ of $Q$ with fractional hypertree width $w$~\cite{Marx:2010}, a cover over $\calT$ has size ${\mathcal O}(|\inst{D}|^w)$. In contrast, there are arbitrarily large databases for which the listing representation of the query result has size $\Omega(|\inst{D}|^{\rho^*})$, where $\rho^*$ is the fractional edge cover number of $Q$~\cite{AtseriasGM13}. The gap between the fractional hypertree width and the fractional edge cover number can be as large as the number of relation symbols in $Q$. For an FAQ (and the special case of a join query) $\varphi$, any cover of its result can be computed in time ${\mathcal O}(|\inst{D}|^w\log |\inst{D}|)$, where $w$ is the FAQ-width~\cite{FAQ:PODS:2016} of $\varphi$. FAQs can express aggregates over database joins~\cite{BKOZ:PVLDB:2013}, in-database optimization~\cite{SOC:SIGMOD:2016,NNOS2017}, matrix chain multiplication, and inference in probabilistic graphical models.

Second, the tuples in the query result can be enumerated from one of its covers with linearithmic pre-computation time and constant delay. This is not the case for the representation defined by the pair of database and join query (unless W[1]=FPT)~\cite{Segoufin:SIGREC:2015}. The benefits of covers over the latter representation are less apparent for acyclic queries, for which both representations share the same linear-size bound and desirable enumeration complexity \cite{Durand07}. 
For acyclic joins, the question thus becomes why to succinctly represent a query result by one relation instead of the pair of a set of relations and the query. 
We next highlight three practical benefits. 
Covers readily provide a subset of the query result without the need to compute the join. This improves cache locality for subsequent operations, e.g., aggregates, since we only need to read in tuple by tuple from the cover instead of reading tuples from different relations stored at different locations in memory and then joining them. 
Similarly, covers provide access locality for disk operations since tuples from the cover are stored on the same disk page, whereas tuples from different relations are stored on different pages.  
Furthermore, covers are samples of the query result that disregard the uninformative yet exhaustive pairings brought by Cartesian products. In exploratory data analysis, the explicit listing of Cartesian products is overwhelming to the user since it may be 
very large. An alternative approach that would present the user with many relations and the query, would have to rely on the user to figure out possible tuples in the query result, which is not desirable. A cover, in contrast, is a compact relation that absolves the user from ad-hoc joining of relations and from re-discovering Cartesian products in a large listing of tuples.
Finally, processing following the in-database joins may require a single relation as input, as it is the case for machine learning over joins~\cite{SOC:SIGMOD:2016}. Indeed, instead of learning regression models over the result of a join we can instead learn them over one of its covers.

Third, covers use the standard listing representation. Prior work introduced lossless representations of query results called {\em factorized databases} that achieve the same succinctness as covers, yet they are directed acyclic graphs that represent the query result as circuits whose nodes are data values or the relational operators Cartesian product and union~\cite{FDB:TODS:2015}. The graph representation makes difficult their adoption as a data representation model by mainstream database systems that rely on relational storage (factorized {\em computation} is however used in relational systems~\cite{NNOS2017}). A relational alternative to factorized databases, as metamorphosed in covers, can prove useful in a variety of settings. The intermediate results in query plans can be represented as covers. In distributed query plans, covers can encode succinctly the otherwise expensive intermediate query results that are communicated among servers in each round~\cite{F1:PVLDB:2013} and can be processed as soon as each of their tuples is received.

{\noindent\bf The contributions} of this paper are as follows:
\begin{itemize}
\item Section~\ref{sec:cores_join_results} introduces covers of join query results  and their correspondence to minimal edge covers in the hypergraphs of the query results. We also give tight size bounds for covers and show that the tuples in the query result can be enumerated from any cover with linearithmic pre-computation time and constant delay.

\item Given a database and a join query, covers of its result can be computed in worst-case optimal time (modulo a log factor). Section~\ref{sec:core-plans} focuses on  the compositionality of cover computation for acyclic join queries. We introduce cover-join plans to compute covers in time linearithmic in their sizes and the size of the input database. A cover-join plan is a binary plan that follows the structure of a join tree of the acyclic query. It uses a cover-join operator that computes covers of the join of two relations, which may be input relations or covers for subqueries. Different plans may lead to different sets of covers. There are covers that cannot be obtained using binary plans.

\item Section~\ref{sec:applications} generalizes our notion of covers from joins to functional aggregate queries by representing succinctly both tuples and aggregates in the query result.
\end{itemize}
We consider natural join queries where each relation is used at most once. The appendix extends our results to arbitrary 
equi-join queries and provides further
details, examples and proofs.

{\noindent\bf Related work.} There are three strands of directly related work: cores in databases and graph theory; succinct representations of query results; and normal forms for relational data.

Cores of graphs, queries, and universal solutions to data exchange problems revolve around smaller yet lossless representations that are homomorphically minimal subgraphs~\cite{CoreGraph:1992}, subqueries~\cite{Chandra:1977}, and universal solutions~\cite{Fagin:CoreDataExchange:2005}, respectively. A further application of graph cores is in the context of the Semantic Web, where cores of RDF graphs are used to obtain minimal representations and normal forms of such graphs~\cite{Gutierrez:CoreSW:2004}. 
Our notion of covers is different. Covers rely on query decompositions to achieve succinctness, and they only become lossless {\em in conjunction} with a decomposition. If we ignore the decomposition, the covers become lossy as they are subsets of the result. Whereas in data exchange all universal solutions have the same core (up to isomorphism), the result of a query may have exponentially many incomparable covers.
 While not a defining component of cores in data exchange, generalized hypertree decompositions can help derive improved algorithms for computing the core of a relational instance with labeled nulls under different classes of dependencies~\cite{Gottlob:CoreDataExchange:2005}.

Covers are relational encodings of d-representations, a lossless graph-based factorization of the query result~\cite{FDB:TODS:2015}. The structure of d-representations is given by variable orders called d-trees, which are an alternative syntax for fractional hypertree decompositions. Whereas d-representations are lossless on their own, covers need the decomposition to derive the missing tuples. Decompositions are the data-independent price to pay for achieving the data-dependent succinctness of factorized representations using the listing representation. Both d-representations and covers achieve succinctness by avoiding the materialization of Cartesian products. Whereas the former encode the products symbolically and losslessly, the covers only keep a minimal subset of the product that is enough to reconstruct it entirely.

The goal of database design is to avoid redundancy in the {\em input} database. Existing normal forms achieve this by {\em decomposing} one relation into several relations guided by functional and join dependencies~\cite{Fagin:ETNF:2012}. Covers exploit the join dependencies to avoid redundancy in the query {\em output}. They do not decompose the result back into the (now globally consistent) input database. Like factorized representations, covers are a normal form for relations representing query results. From a cover of a join result over a decomposition, we can obtain a decomposition of the join result in project-join normal form (5NF)~\cite{Fagin:5NF:1979} by taking one projection of the cover onto the attributes of each bag of the decomposition.


\section{Preliminaries}
\label{sec:preliminaries}

{\noindent\bf Databases.} 
We assume an ordered domain of data values. 
A relation schema is a finite set of attributes. 
For an attribute $A$, we denote by $\dom(A)$
its domain.
A database schema is a finite set of relation symbols. 
A tuple $t$ over a relation schema $S$ is a mapping from the attributes in $S$ to values in their respective domains. A relation over a relation schema $S$ is a finite set of tuples over $S$. 
A database $\inst{D}$ over a database schema 
$\calS$ contains for each relation symbol in $\calS$, 
a relation over the same schema.
For a relation (symbol) $R$ and  tuple $t$, we use 
$\sign{R}$ and $\sign{t}$ to refer to their schemas 
and write $R(S)$ to express
that the schema of $R$ is $S$.
The tuples $t_1, \ldots ,t_n$ are joinable if $\pi_{S_{i,j}}t_i = \pi_{S_{i,j}}t_j$ for all $i,j \in \indexnat{n}$ and $S_{i,j}=\sign{t_i}\cap \sign{t_j}$. The size $|R|$ of a relation $R$ is the number of its tuples. The size $|\inst{D}|$ of a database $\inst{D}$ is the sum of the sizes of its relations.

\smallskip

{\noindent\bf Natural Join Queries.}
We consider natural join queries of the form $Q = R_1(S_1) \Join \ldots \Join R_n(S_n)$, where each $R_i$ is a relation symbol over relation schema $S_i$ and refers to a database relation over the same schema. Notation-wise we do not 
distinguish between a relation symbol and the corresponding relation.  
The joins in $Q$ are expressed by sharing attributes across relation schemas. 
The schema $\sign{Q}$ of $Q$ is the set of relation symbols in $Q$: $\sign{Q}=\{R_i\}_{i\in[n]}$. The set $\mathit{att}(Q)$ of attributes of $Q$ is the union of the schemas of its relation symbols: $\mathit{att}(Q) = \bigcup_{i\in[n]}S_i$. The size $|Q|$ of $Q$ is the number of its relation symbols: $|Q|=n$. A database is {\em globally consistent} with respect to a query $Q$ if there are no (dangling) tuples that do not contribute to the result of $Q$~\cite{AHV95}.
Two relations $R_1$ and $R_2$ are called consistent if the database $\{R_1,R_2\}$
is globally consistent with respect to the query $R_1 \Join R_2$.
We assume that relation symbols in $Q$ are non-repeating and each relation symbol corresponds to a distinct relation. Appendix~\ref{sec:covers_equi_join} lifts these restrictions and extends our contributions to arbitrary equi-join queries.

\smallskip

{\noindent\bf Hypergraphs.} Let $H$ be a multi-hypergraph (hypergraph for short) whose edge multiset   $E$ may contain multiple hyperedges (edges for short) with the same node set. A fractional edge cover for $H$ is a function $\gamma$ mapping each edge in $H$ to a positive number such that  $\Sigma_{e \ni v} \gamma(e) \geq 1$ for each node $v$ of $H$, i.e., the sum of the function values for all edges incident  to $v$ is at least $1$. We define the weight of a fractional edge cover $\gamma$ as $\weight{\gamma}=\Sigma_{e \in E} \gamma(e)$. The fractional edge cover number $\rho^*(H)$ of $H$ is the minimum weight of fractional edge covers of $H$. It can be obtained from a fractional edge cover where the edge weights are rational numbers of bit-length polynomial in the size of $H$~\cite{AtseriasGM13}. 

We use hypergraphs for queries and for relations representing their results. 
 The hypergraph $H$ of a query $Q$ consists of one node $A$ for each attribute
$A$ in $Q$ and one edge $\sign{R}$ for each 
relation symbol $R\in\sign{Q}$.
We define $\rho^*(Q) = \rho^*(H)$.

Let $R$ be a relation and $\calP$ a set of 
(possibly overlapping)  subsets of $\sign{R}$ such that 
$\bigcup_{S \in \calP}S = \sign{R}$.  
The hypergraph $H$ of $R$ over ${\cal P}$ consists of one node for each distinct tuple in $\pi_S R$ for each attribute set $S\in{\cal P}$ and one edge for each tuple in $R$. The edge for a tuple $t$ thus consists of all nodes for tuples 
$\pi_S (t)$ with $S\in \calP$. 
We use $\tupleofind{}{v}$ to denote the tuple represented by a node or edge $v$ in $H$. Given a subset $M$ of the edges in $H$, we define $\relofind{}{M}=\setindexel{\tupleofind{}{e}}{e}{M}$ as the relation represented by $M$. The set $M$ is an edge cover of $H$ if each node in $H$ is contained in at least  one edge in $M$. The set $M$ is a minimal edge cover if it is an edge cover and any of its strict subsets is not.

\begin{example}\label{ex:hypergraph} 
Consider the path query $Q=R_1(A,B)\Join R_2(B,C)\Join R_3(C,D)$.
Figure~\ref{fig:hypergraph} 
depicts in the top row a database of the three relations $R_1$, $R_2$ and $R_3$, 
the query result and a subset of it. In the bottom row, the figure depicts the hypergraph of $Q$ (and its decomposition defined below), the hypergraph of its result over the attribute sets
 $\{\{A,B\},\{B,C\},\{C,D\}\}$, and the hypergraph of a subset of the query result over the same attribute sets.
\end{example}

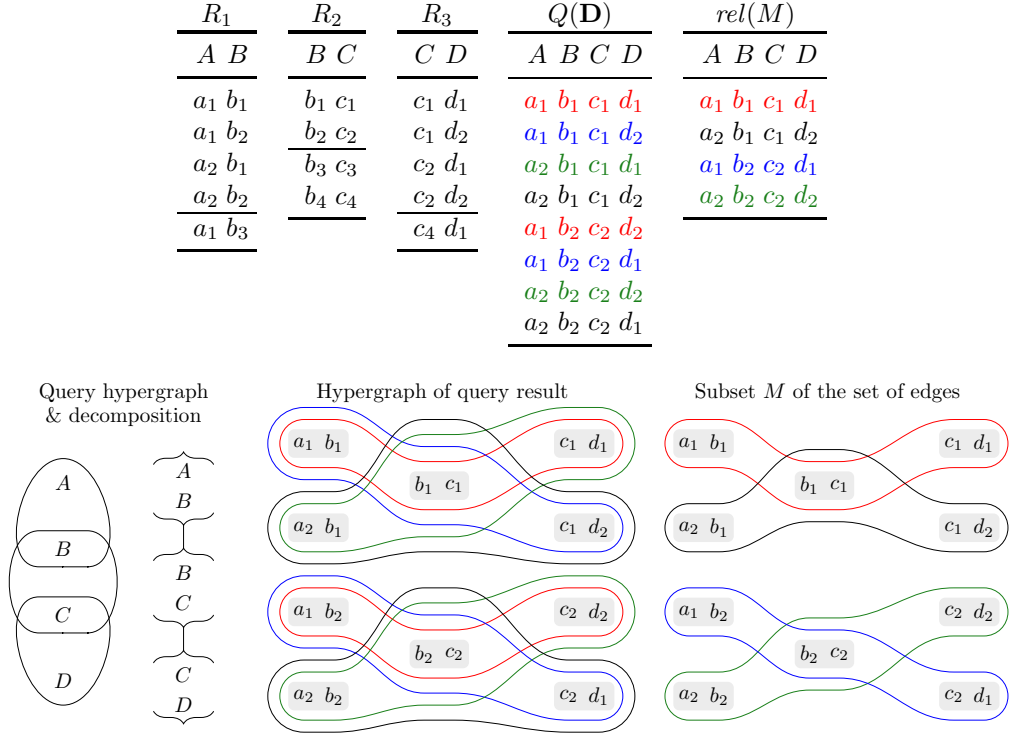
\begin{figure}[h]

\begin{center}
\begin{tabular}{lllll}

\begin{tabular}{c@{\;}c@{\;}}
  \multicolumn{2}{c}{$R_1$} \\
  \toprule
  $A$ & $B$  \\\midrule
  $a_1$ & $b_1$\\
  $a_1$ & $b_2$\\
  $a_2$ & $b_1$\\
  $a_2$ & $b_2$\\\hline
  $a_1$ & $b_3$\\  
  \bottomrule
  \\
  \\
  \\
\end{tabular}

&

\begin{tabular}{c@{\;}c@{\;}}
  \multicolumn{2}{c}{$R_2$} \\
  \toprule
  $B$ & $C$  \\\midrule
  $b_1$ & $c_1$\\
  $b_2$ & $c_2$\\\hline
  $b_3$ & $c_3$\\
  $b_4$ & $c_4$\\    
  \bottomrule
  \\
  \\
  \\
  \\
  \end{tabular}

&

\begin{tabular}{c@{\;}c@{\;}}
  \multicolumn{2}{c}{$R_3$} \\
  \toprule
  $C$ & $D$  \\\midrule
  $c_1$ & $d_1$\\
  $c_1$ & $d_2$\\
  $c_2$ & $d_1$\\
  $c_2$ & $d_2$\\\hline
  $c_4$ & $d_1$\\
  \bottomrule
  \\
  \\
  \\
\end{tabular}

&

\begin{tabular}{c@{\;}c@{\;}c@{\;}c@{\;}}
  \multicolumn{4}{c}{$Q(\inst{D})$} \\
  \toprule
  $A$ & $B$ & $C$ & $D$  \\\midrule 
  {\color{red}$a_1$} & {\color{red}$b_1$}  & {\color{red}$c_1$} & {\color{red}$d_1$}\\
  {\color{blue}$a_1$} & {\color{blue}$b_1$} & {\color{blue}$c_1$} & {\color{blue}$d_2$}\\
  {\color{goodgreen}$a_2$} & {\color{goodgreen}$b_1$} & {\color{goodgreen}$c_1$} & {\color{goodgreen}$d_1$}\\
  $a_2$ & $b_1$ & $c_1$ & $d_2$\\
  {\color{red}$a_1$} & {\color{red}$b_2$} & {\color{red}$c_2$} & {\color{red}$d_2$}\\
  {\color{blue}$a_1$} & {\color{blue}$b_2$} & {\color{blue}$c_2$} & {\color{blue}$d_1$}\\
  {\color{goodgreen}$a_2$} & {\color{goodgreen}$b_2$} & {\color{goodgreen}$c_2$} & {\color{goodgreen}$d_2$}\\
  $a_2$ & $b_2$ & $c_2$ & $d_1$\\  
  \bottomrule
\end{tabular}

&

\begin{tabular}{c@{\;}c@{\;}c@{\;}c@{\;}}
  \multicolumn{4}{c}{$\mathit{rel}(M)$} \\
  \toprule
  $A$ & $B$ & $C$ & $D$  \\\midrule 
  {\color{red}$a_1$} & {\color{red}$b_1$} & {\color{red}$c_1$} & {\color{red}$d_1$}\\
  $a_2$ & $b_1$ & $c_1$ & $d_2$\\
  {\color{blue}$a_1$} & {\color{blue}$b_2$} & {\color{blue}$c_2$} & {\color{blue}$d_1$}\\
  {\color{goodgreen}$a_2$} & {\color{goodgreen}$b_2$} & {\color{goodgreen}$c_2$} & {\color{goodgreen}$d_2$}\\
  \bottomrule
  \\
  \\
  \\
  \\
\end{tabular}
\end{tabular}\\[1em]

\scalebox{0.8}{
\begin{tikzpicture}

\node at(0,-3.2) (joinGraph){Query hypergraph};
\node at(0,-3.6) {\& decomposition};

\node (1)  at(-1, -4.7) {$A$};

\node (2) [below of =1, node distance=1.1cm] {$B$};

\node (3) [below of =2, node distance=1.1cm] {$C$};

\node (4) [below of =3, node distance=1.1cm] {$D$};

    \begin{scope}[fill opacity=0.8]
    \draw
       ($(1)+(0,0.4)$) 
         to[out=0,in=0] ($(2) + (0.4,-0.3)$)
        to[out=0,in=180] ($(2) + (0,-0.3)$)
        to[out=180,in=0] ($(2) + (-0.4,-0.3)$)
        to[out=180,in=180] ($(1) + (0,0.4)$);

         \draw
       ($(2)+(0,0.3)$) 
        to[out=0,in=180] ($(2) + (0.4,0.3)$)
        to[out=0,in=0] ($(3) + (0.4,-0.3)$)
        to[out=0,in=180] ($(3) + (0,-0.3)$)
        to[out=180,in=0] ($(3) + (-0.4,-0.3)$)
        to[out=180,in=180] ($(2) + (-0.4,0.3)$)
        to[out=0,in=180] ($(2) + (0,0.3)$);      
                
         \draw
       ($(3)+(0,0.3)$) 
        to[out=0,in=180] ($(3) + (0.4,0.3)$)
        to[out=0,in=0] ($(4) + (0,-0.4)$)
        to[out=180,in=180] ($(3) + (-0.4,0.3)$)
        to[out=0,in=180] ($(3) + (0,0.3)$);                  
     \end{scope}

\node (A1)  at(1, -4.5) {$A$};
\node (A2)  [below of =A1, node distance=0.5cm] {$B$};

\node (A23) [below of =A1, node distance=1.7cm] {$B$};
\node (A231) [below of =A23, node distance=0.5cm] {$C$};

\node (A45) [below of =A23, node distance=1.7cm] {$C$};
\node (A451) [below of =A45, node distance=0.5cm] {$D$};

\draw [decorate,decoration={brace,amplitude=6pt,raise=0pt}] 
(0.5,-4.4) -- (1.5,-4.4);
\draw [decorate,decoration={brace,amplitude=6pt,raise=0pt,mirror}] 
(0.5,-5.2) -- (1.5,-5.2);

\draw [decorate,decoration={brace,amplitude=6pt,raise=0pt}] 
(0.5,-6) -- (1.5,-6);
\draw [decorate,decoration={brace,amplitude=6pt,raise=0pt,mirror}] 
(0.5,-6.9) -- (1.5,-6.9);

\draw [decorate,decoration={brace,amplitude=6pt,raise=0pt}] 
(0.5,-7.7) -- (1.5,-7.7);
\draw [decorate,decoration={brace,amplitude=6pt,raise=0pt,mirror}] 
(0.5,-8.5) -- (1.5,-8.5);

\draw [decorate, segment length=20] (1,-5.4) -- (1,-5.8);
\draw [decorate, segment length=20] (1,-7.1) -- (1,-7.5);

\node at(5.3,-3.2) (joinGraph){Hypergraph of query result};   
\node at(3,-4.05) (11) {$a_1$};
\node (12) [right of =11,node distance=0.5cm] {$b_1$};

\node (21) [below of =11,node distance=1.4cm] {$a_2$};
\node (22) [right of =21,node distance=0.5cm] {$b_1$};

\node (31) [below of =21,node distance=1.4cm] {$a_1$};
\node (32) [right of =31,node distance=0.5cm] {$b_2$};

\node (31new) [below of =31,node distance=1.4cm] {$a_2$};
\node (32new) [right of =31new,node distance=0.5cm] {$b_2$};

\tikzstyle{background}=[rectangle,
                                                rounded corners=1mm]

\begin{pgfonlayer}{background}                   
   \node [background,
                    fit=(11)(12),
                    fill=gray!15,inner sep= -1] {};        
                    
   \node [background,
                    fit=(21)(22),
                    fill=gray!15,inner sep= -1] {};
                    
   \node [background,
                    fit=(31)(32),
                    fill=gray!15,inner sep= -1] {};                                                          

   \node [background,
                    fit=(31new)(32new),
                    fill=gray!15,inner sep= -1] {};                                                          
\end{pgfonlayer}

\node (invisible) [right of =12,node distance=1.5cm]  {};
\node (invisible2) [right of =invisible,node distance=0.7cm]  {};

\node (41) [below of =invisible,node distance=0.7cm]  {$b_1$};
\node (42) [right of =41,node distance=0.5cm] {$c_1$};

\node (51) [below of =41,node distance=2.8cm] {$b_2$};
\node (52) [right of =51,node distance=0.5cm] {$c_2$};

\tikzstyle{background}=[rectangle,
                                                rounded corners=1mm]

\begin{pgfonlayer}{background}                   
   \node [background,
                    fit=(41)(42),
                    fill=gray!15,inner sep= -1] {};        
                    
   \node [background,
                    fit=(51)(52),
                    fill=gray!15,inner sep= -1] {};
                    
\end{pgfonlayer}


\node (71) [right of =invisible2,node distance=1.7cm]  {$c_1$};
\node (72) [right of =71,node distance=0.5cm] {$d_1$};

\node (71new) [below of =71,node distance=1.4cm]  {$c_1$};
\node (72new) [right of =71new,node distance=0.5cm] {$d_2$};

\node (81) [below of =71new,node distance=1.4cm] {$c_2$};
\node (82) [right of =81,node distance=0.5cm] {$d_2$};

\node (91) [below of =81,node distance=1.4cm] {$c_2$};
\node (92) [right of =91,node distance=0.5cm] {$d_1$};

\tikzstyle{background}=[rectangle,
                                                rounded corners=1mm]

\begin{pgfonlayer}{background}                   
   \node [background,
                    fit=(71)(72),
                    fill=gray!15,inner sep= -1] {};        
                    
   \node [background,
                    fit=(71new)(72new),
                    fill=gray!15,inner sep= -1] {};                            
                    
   \node [background,
                    fit=(81)(82),
                    fill=gray!15,inner sep= -1] {};
                    
 \node [background,
                    fit=(91)(92),
                    fill=gray!15,inner sep= -1] {};  
                     
\end{pgfonlayer}

   
    \begin{scope}[fill opacity=0.8]
    \draw[color=red]
       ($(11)+(0,0.4)$) 
        to[out=0,in=180] ($(12) + (0,0.4)$)
        to[out=0,in=180] ($(41) + (0,0.4)$)
        to[out=0,in=180] ($(42)+(0,0.4)$)
        to[out=0,in=180] ($(71)+(0,0.4)$)      
        to[out=0,in=180] ($(72)+(0,0.4)$)            
        to[out=0,in=90] ($(72)+(0.4,0)$)                     
        to[out=270,in=0] ($(72)+(0,-0.4)$)                        
        to[out=180,in=0] ($(71)+(0,-0.4)$)                        
        to[out=180,in=0] ($(42)+(0,-0.4)$)         
        to[out=180,in=0] ($(41)+(0,-0.4)$)                
        to[out=180,in=0] ($(12)+(0,-0.4)$) 
        to[out=180,in=0] ($(11)+(0,-0.4)$)
        to[out=180,in=270] ($(11)+(-0.4,0)$)    
        to[out=90,in=180] ($(11)+(0,0.4)$);

    \draw[color=blue]
       ($(11)+(0,0.6)$) 
        to[out=0,in=180] ($(12) + (0,0.6)$)
        to[out=0,in=180] ($(41) + (0,0.65)$)
        to[out=0,in=180] ($(42)+(0,0.65)$)
        to[out=0,in=180] ($(71new)+(0,0.4)$)      
        to[out=0,in=180] ($(72new)+(0,0.4)$)            
        to[out=0,in=90] ($(72new)+(0.4,0)$)                     
        to[out=270,in=0] ($(72new)+(0,-0.4)$)                        
        to[out=180,in=0] ($(71new)+(0,-0.4)$)                        
        to[out=180,in=0] ($(42)+(0,-0.65)$)         
        to[out=180,in=0] ($(41)+(0,-0.65)$)                
        to[out=180,in=0] ($(12)+(0,-0.6)$) 
        to[out=180,in=0] ($(11)+(0,-0.6)$)
        to[out=180,in=270] ($(11)+(-0.6,0)$)    
        to[out=90,in=180] ($(11)+(0,0.60)$);

      \draw[color=goodgreen]
       ($(21)+(0,0.4)$) 
        to[out=0,in=180] ($(22) + (0,0.4)$)
        to[out=0,in=180] ($(41) + (0,0.85)$)
        to[out=0,in=180] ($(42)+(0,0.85)$)
        to[out=0,in=180] ($(71)+(0,0.6)$)      
        to[out=0,in=180] ($(72)+(0,0.6)$)            
        to[out=0,in=90] ($(72)+(0.6,0)$)                     
        to[out=270,in=0] ($(72)+(0,-0.6)$)                        
        to[out=180,in=0] ($(71)+(0,-0.6)$)                        
        to[out=180,in=0] ($(42)+(0,-0.85)$)         
        to[out=180,in=0] ($(41)+(0,-0.85)$)         
        to[out=180,in=0] ($(22)+(0,-0.4)$) 
        to[out=180,in=0] ($(21)+(0,-0.4)$)
        to[out=180,in=270] ($(21)+(-0.4,0)$)   
        to[out=90,in=180] ($(21)+(0,0.4)$);

      \draw
       ($(21)+(0,0.6)$) 
        to[out=0,in=180] ($(22) + (0,0.6)$)
        to[out=0,in=180] ($(41) + (0,1.1)$)
        to[out=0,in=180] ($(42)+(0,1.1)$)
        to[out=0,in=180] ($(71new)+(0,0.6)$)      
        to[out=0,in=180] ($(72new)+(0,0.6)$)            
        to[out=0,in=90] ($(72new)+(0.6,0)$)                     
        to[out=270,in=0] ($(72new)+(0,-0.6)$)                        
        to[out=180,in=0] ($(71new)+(0,-0.6)$)                        
        to[out=180,in=0] ($(42)+(0,-1.1)$)         
        to[out=180,in=0] ($(41)+(0,-1.1)$)         
        to[out=180,in=0] ($(22)+(0,-0.6)$) 
        to[out=180,in=0] ($(21)+(0,-0.6)$)
        to[out=180,in=270] ($(21)+(-0.6,0)$)   
        to[out=90,in=180] ($(21)+(0,0.6)$);

     
      \draw[color=red]
       ($(31)+(0,0.4)$) 
        to[out=0,in=180] ($(32) + (0,0.4)$)
        to[out=0,in=180] ($(51) + (0,0.4)$)
        to[out=0,in=180] ($(52)+(0,0.4)$)
        to[out=0,in=180] ($(81)+(0,0.4)$)      
        to[out=0,in=180] ($(82)+(0,0.4)$)            
        to[out=0,in=90] ($(82)+(0.4,0)$)                       
        to[out=270,in=0] ($(82)+(0,-0.4)$)                        
        to[out=180,in=0] ($(81)+(0,-0.4)$)                        
        to[out=180,in=0] ($(52)+(0,-0.4)$)         
        to[out=180,in=0] ($(51)+(0,-0.4)$)              
        to[out=180,in=0] ($(32)+(0,-0.4)$) 
        to[out=180,in=0] ($(31)+(0,-0.4)$)
        to[out=180,in=270] ($(31)+(-0.4,0)$)   
        to[out=90,in=180] ($(31)+(0,0.4)$);

      \draw[color=blue]
       ($(31)+(0,0.6)$) 
        to[out=0,in=180] ($(32) + (0,0.6)$)
        to[out=0,in=180] ($(51) + (0,0.65)$)
        to[out=0,in=180] ($(52)+(0,0.65)$)
        to[out=0,in=180] ($(91)+(0,0.4)$)      
        to[out=0,in=180] ($(92)+(0,0.4)$)            
        to[out=0,in=90] ($(92)+(0.4,0)$)                     
        to[out=270,in=0] ($(92)+(0,-0.4)$)                        
        to[out=180,in=0] ($(91)+(0,-0.4)$)                        
        to[out=180,in=0] ($(52)+(0,-0.65)$)         
        to[out=180,in=0] ($(51)+(0,-0.65)$)                
        to[out=180,in=0] ($(32)+(0,-0.6)$) 
        to[out=180,in=0] ($(31)+(0,-0.6)$)
        to[out=180,in=270] ($(31)+(-0.6,0)$)    
        to[out=90,in=180] ($(31)+(0,0.6)$);

     \draw[color=goodgreen]
       ($(31new)+(0,0.4)$) 
        to[out=0,in=180] ($(32new) + (0,0.4)$)
        to[out=0,in=180] ($(51) + (0,0.85)$)
        to[out=0,in=180] ($(52)+(0,0.85)$)
        to[out=0,in=180] ($(81)+(0,0.6)$)      
        to[out=0,in=180] ($(82)+(0,0.6)$)            
        to[out=0,in=90] ($(82)+(0.6,0)$)                     
        to[out=270,in=0] ($(82)+(0,-0.6)$)                        
        to[out=180,in=0] ($(81)+(0,-0.6)$)                        
        to[out=180,in=0] ($(52)+(0,-0.85)$)         
        to[out=180,in=0] ($(51)+(0,-0.85)$)                
        to[out=180,in=0] ($(32new)+(0,-0.4)$) 
        to[out=180,in=0] ($(31new)+(0,-0.4)$)
        to[out=180,in=270] ($(31new)+(-0.4,0)$)    
        to[out=90,in=180] ($(31new)+(0,0.4)$);

             \draw
       ($(31new)+(0,0.6)$) 
        to[out=0,in=180] ($(32new) + (0,0.6)$)
        to[out=0,in=180] ($(51) + (0,1.1)$)
        to[out=0,in=180] ($(52)+(0,1.1)$)
        to[out=0,in=180] ($(91)+(0,0.6)$)      
        to[out=0,in=180] ($(92)+(0,0.6)$)            
        to[out=0,in=90] ($(92)+(0.6,0)$)                     
        to[out=270,in=0] ($(92)+(0,-0.6)$)                        
        to[out=180,in=0] ($(91)+(0,-0.6)$)                        
        to[out=180,in=0] ($(52)+(0,-1.1)$)         
        to[out=180,in=0] ($(51)+(0,-1.1)$)                
        to[out=180,in=0] ($(32new)+(0,-0.6)$) 
        to[out=180,in=0] ($(31new)+(0,-0.6)$)
        to[out=180,in=270] ($(31new)+(-0.6,0)$)    
        to[out=90,in=180] ($(31new)+(0,0.6)$);

     \end{scope}

\node at(11.7,-3.2)(edgeSubset){Subset $M$ of the set of edges};
\node at(9.4,-4.05) (11) {$a_1$};
\node (12) [right of =11,node distance=0.5cm] {$b_1$};

\node (21) [below of =11,node distance=1.4cm] {$a_2$};
\node (22) [right of =21,node distance=0.5cm] {$b_1$};

\node (31) [below of =21,node distance=1.4cm] {$a_1$};
\node (32) [right of =31,node distance=0.5cm] {$b_2$};

\node (31new) [below of =31,node distance=1.4cm] {$a_2$};
\node (32new) [right of =31new,node distance=0.5cm] {$b_2$};

\tikzstyle{background}=[rectangle,
                                                rounded corners=1mm]

\begin{pgfonlayer}{background}                   
   \node [background,
                    fit=(11)(12),
                    fill=gray!15,inner sep= -1] {};        
                    
   \node [background,
                    fit=(21)(22),
                    fill=gray!15,inner sep= -1] {};
                    
   \node [background,
                    fit=(31)(32),
                    fill=gray!15,inner sep= -1] {};                                                          

   \node [background,
                    fit=(31new)(32new),
                    fill=gray!15,inner sep= -1] {};                                                          
\end{pgfonlayer}

\node (invisible) [right of =12,node distance=1.5cm]  {};
\node (invisible2) [right of =invisible,node distance=0.7cm]  {};

\node (41) [below of =invisible,node distance=0.7cm]  {$b_1$};
\node (42) [right of =41,node distance=0.5cm] {$c_1$};

\node (51) [below of =41,node distance=2.8cm] {$b_2$};
\node (52) [right of =51,node distance=0.5cm] {$c_2$};

\tikzstyle{background}=[rectangle,
                                                rounded corners=1mm]

\begin{pgfonlayer}{background}                   
   \node [background,
                    fit=(41)(42),
                    fill=gray!15,inner sep= -1] {};        
                    
   \node [background,
                    fit=(51)(52),
                    fill=gray!15,inner sep= -1] {};
                    
\end{pgfonlayer}


\node (71) [right of =invisible2,node distance=1.7cm]  {$c_1$};
\node (72) [right of =71,node distance=0.5cm] {$d_1$};

\node (71new) [below of =71,node distance=1.4cm]  {$c_1$};
\node (72new) [right of =71new,node distance=0.5cm] {$d_2$};

\node (81) [below of =71new,node distance=1.4cm] {$c_2$};
\node (82) [right of =81,node distance=0.5cm] {$d_2$};

\node (91) [below of =81,node distance=1.4cm] {$c_2$};
\node (92) [right of =91,node distance=0.5cm] {$d_1$};

\tikzstyle{background}=[rectangle,
                                                rounded corners=1mm]

\begin{pgfonlayer}{background}                   
   \node [background,
                    fit=(71)(72),
                    fill=gray!15,inner sep= -1] {};        
                    
   \node [background,
                    fit=(71new)(72new),
                    fill=gray!15,inner sep= -1] {};                            
                    
   \node [background,
                    fit=(81)(82),
                    fill=gray!15,inner sep= -1] {};
                    
 \node [background,
                    fit=(91)(92),
                    fill=gray!15,inner sep= -1] {};  
                     
\end{pgfonlayer}

   
    \begin{scope}[fill opacity=0.8]
    \draw[color=red]
       ($(11)+(0,0.4)$) 
        to[out=0,in=180] ($(12) + (0,0.4)$)
        to[out=0,in=180] ($(41) + (0,0.4)$)
        to[out=0,in=180] ($(42)+(0,0.4)$)
        to[out=0,in=180] ($(71)+(0,0.4)$)      
        to[out=0,in=180] ($(72)+(0,0.4)$)            
        to[out=0,in=90] ($(72)+(0.4,0)$)                     
        to[out=270,in=0] ($(72)+(0,-0.4)$)                        
        to[out=180,in=0] ($(71)+(0,-0.4)$)                        
        to[out=180,in=0] ($(42)+(0,-0.4)$)         
        to[out=180,in=0] ($(41)+(0,-0.4)$)                
        to[out=180,in=0] ($(12)+(0,-0.4)$) 
        to[out=180,in=0] ($(11)+(0,-0.4)$)
        to[out=180,in=270] ($(11)+(-0.4,0)$)    
        to[out=90,in=180] ($(11)+(0,0.4)$);

        \draw
       ($(21)+(0,0.4)$) 
        to[out=0,in=180] ($(22) + (0,0.4)$)
        to[out=0,in=180] ($(41) + (0,0.6)$)
        to[out=0,in=180] ($(42)+(0,0.6)$)
        to[out=0,in=180] ($(71new)+(0,0.4)$)      
        to[out=0,in=180] ($(72new)+(0,0.4)$)            
        to[out=0,in=90] ($(72new)+(0.4,0)$)                     
        to[out=270,in=0] ($(72new)+(0,-0.4)$)                        
        to[out=180,in=0] ($(71new)+(0,-0.4)$)                        
        to[out=180,in=0] ($(42)+(0,-0.6)$)         
        to[out=180,in=0] ($(41)+(0,-0.6)$)         
        to[out=180,in=0] ($(22)+(0,-0.4)$) 
        to[out=180,in=0] ($(21)+(0,-0.4)$)
        to[out=180,in=270] ($(21)+(-0.4,0)$)   
        to[out=90,in=180] ($(21)+(0,0.4)$);

     
         \draw[color=goodgreen]
       ($(31new)+(0,0.4)$) 
        to[out=0,in=180] ($(32new) + (0,0.4)$)
        to[out=0,in=180] ($(51) + (0,0.6)$)
        to[out=0,in=180] ($(52)+(0,0.6)$)
        to[out=0,in=180] ($(81)+(0,0.4)$)      
        to[out=0,in=180] ($(82)+(0,0.4)$)            
        to[out=0,in=90] ($(82)+(0.4,0)$)                     
        to[out=270,in=0] ($(82)+(0,-0.4)$)                        
        to[out=180,in=0] ($(81)+(0,-0.4)$)                        
        to[out=180,in=0] ($(52)+(0,-0.6)$)         
        to[out=180,in=0] ($(51)+(0,-0.6)$)                
        to[out=180,in=0] ($(32new)+(0,-0.4)$) 
        to[out=180,in=0] ($(31new)+(0,-0.4)$)
        to[out=180,in=270] ($(31new)+(-0.4,0)$)    
        to[out=90,in=180] ($(31new)+(0,0.4)$);

      \draw[color=blue]
       ($(31)+(0,0.4)$) 
        to[out=0,in=180] ($(32) + (0,0.4)$)
        to[out=0,in=180] ($(51) + (0,0.4)$)
        to[out=0,in=180] ($(52)+(0,0.4)$)
        to[out=0,in=180] ($(91)+(0,0.4)$)      
        to[out=0,in=180] ($(92)+(0,0.4)$)            
        to[out=0,in=90] ($(92)+(0.4,0)$)                     
        to[out=270,in=0] ($(92)+(0,-0.4)$)                        
        to[out=180,in=0] ($(91)+(0,-0.4)$)                        
        to[out=180,in=0] ($(52)+(0,-0.4)$)         
        to[out=180,in=0] ($(51)+(0,-0.4)$)                
        to[out=180,in=0] ($(32)+(0,-0.4)$) 
        to[out=180,in=0] ($(31)+(0,-0.4)$)
        to[out=180,in=270] ($(31)+(-0.4,0)$)    
        to[out=90,in=180] ($(31)+(0,0.4)$);

     \end{scope}

\end{tikzpicture}
}
\end{center}
\caption{
Top row: database $\inst{D}=\{R_1,R_2,R_3\}$, the result $Q(\inst{D})$ of the path query $Q$ in Example~\ref{ex:hypergraph}, and a subset of $Q(\inst{D})$; bottom row: the hypergraph of $Q$, the tree of a decomposition $\calT$ of $Q$, the hypergraph of $Q(\inst{D})$ over attribute sets $\sign{\calT}$, and a minimal edge cover $M$ of this hypergraph.}
\label{fig:hypergraph}
\end{figure}

{\noindent\bf Decompositions.} A {\em hypertree decomposition} $\calT$ of (the hypergraph $H$  of) a query $Q$ is a pair $(T,\chi)$, where $T$ is a tree and $\chi$ a function mapping each node in $T$ to a subset of the nodes of $H$. For a node $t\in T$, the set $\chi(t)$  is called a bag. A hypertree decomposition satisfies two properties. {\em Coverage}: For each edge $e$ in $H$, there must be a node $t$ in $T$ with $e \subseteq \chi(t)$.  {\em Connectivity}: For each node $v$ in $H$, the set $\{t \mid t \in T, v\in \chi(t)\}$ must be non-empty and form a connected subtree in $T$.  The schema of $\calT$ is the set of its bags: $\sign{\calT}= \{ \chi(t)\mid t\in T\}$. The attributes of $\calT$ are defined by $\att{\calT}=\bigcup_{B\in\sign{\calT}} B$.

A {\em fractional hypertree decomposition}~\cite{GroheM14} of (the hypergraph $H$  
 of) a query $Q$ is a triple $(T,\chi,\setindexel{\gamma_t}{t}{T})$ 
 where $(T,\chi)$ is a hypertree decomposition of $H$ and for each node $t\in T$, $\gamma_t$ is  a fractional edge cover of minimal weight
  for the subgraph of $H$ restricted to $\chi(t)$.
We define the {\em fractional hypertree width} of $\calT=(T,\chi, \setindexel{\gamma_t}{t}{T})$ as $\max_{t \in T} \{\weight{\gamma_t}\}$
and we denote it by $\fhtw{\calT}$.
The fractional hypertree width $\fhtw{H}$ of the hypergraph $H$
is the minimal possible such width of any fractional hypertree decomposition of $H$. 
The fractional hypertree width $\fhtw{Q}$ of a query $Q$ is the fractional hypertree width $\fhtw{H}$ of its hypergraph $H$.
For simplicity, we use the terms decomposition and width in place of fractional hypertree decomposition and fractional hypertree width, respectively.

A hypergraph $H$ is {\em $\alpha$-acyclic} (acyclic for short) if it has a decomposition in which  each bag is contained in an edge of $H$~\cite{BeeriFMY83}. A query whose hypergraph is acyclic is also called acyclic. The width of any acyclic hypergraph or query is one. A {\em join tree} of a query $Q$ is a 
labelled tree $(T,\ell)$
where $T = (\sign{Q},E)$ is a tree and $\ell$ is an edge labelling such that
 \begin{inparaenum}[(i)] 
\item  each edge $e = (R,R') \in E$ is labelled by 
$\ell(e) = \sign{R} \cap \sign{R'}$ and 
\item for every pair $R$, $R'$ of distinct nodes and for each attribute 
 $A \in \sign{R} \cap \sign{R'}$, the label of each edge along the unique path between 
 $R$ and $R'$ includes $A$ (Section 6.4 in \cite{AHV95}).   
\end{inparaenum}
A query is acyclic if and only if it admits a join tree (Theorem 6.4.5 in \cite{AHV95}). 
The decomposition $\calT$ {\em corresponding} to the join tree $\calJ$ of a query $Q$
is constructed as follows. Each node in $\calJ$, which corresponds to a relation symbol $R$, is mapped to a node in $\calT$, which has the bag $\sign{R}$. 
For each node $t$ in $\calT$ with bag $\sign{R}$, the function $\gamma_t$ maps the hyperedge for $R$ to 1.

\begin{example}\label{ex:decomposition}
Figure~\ref{fig:hypergraph} gives the hypergraph (left, bottom row) of the 
 path query in Example~\ref{ex:hypergraph} along with one of its decompositions. 
 This decomposition has width one, since each bag is included in one edge of the hypergraph; the path query is acyclic. The decomposition, where the top two bags are merged into one, has width two. For queries with cycles, e.g., Loomis-Whitney queries~\cite{Ngo:SIGREC:2013}, the width can be larger than one. For instance, the width of the triangle query (Loomis-Whitney query over three relations) is $3/2$~\cite{AtseriasGM13}.
\end{example}

{\noindent\bf Computational Model.}
We use the uniform-cost RAM model \cite{AhoHU74} 
where data values as well as pointers to databases are of constant size.
Our analysis is with respect to data complexity where the query is assumed fixed. We use $\widetilde{\mathcal O}$ to hide a $\log|\inst{D}|$ factor.

\smallskip
{\noindent\bf Result-preserving Transformation.} 
Let $(Q,\calT,\inst{D})$ denote a triple of a natural join query $Q$, a
 decomposition $\calT$ of $Q$, and a database $\inst{D}$.

\begin{proposition}\label{prop:rewriting2}
Given $(Q,\calT,\inst{D})$, we can compute $(Q',\calT,\inst{D}')$ 
with size ${\mathcal O}(|\inst{D}|^{\fhtw{\calT}})$  and in time $\widetilde{\mathcal O}(|\inst{D}|^{\fhtw{\calT}})$ 
such that $Q'$ is an acyclic natural join query, $\calT$ corresponds 
to  a join tree of $Q'$, 
 $\inst{D}'$ is globally consistent with respect to $Q'$ and $Q'(\inst{D}')=Q(\inst{D})$.
\end{proposition}

\begin{example}
Consider the path query $Q$, decomposition $\calT$, and database $\inst{D}$ in 
Example~\ref{ex:decomposition}. 
The application of Proposition~\ref{prop:rewriting2} 
 leaves $Q$ unchanged, since $Q$ is already acyclic and $\calT$ corresponds 
 to a join tree of $Q$.
 The database in Figure~\ref{fig:hypergraph} is not globally consistent with respect to $Q$, since it contains tuples (under the thin lines) that do not contribute to the result. We remove these dangling tuples to make it consistent.

Consider now the bowtie query 
$Q_\Join = R_1(A,B) \Join R_2(B,C) \Join R_3(A,C) \Join R_4(A,D)\Join 
R_5(D,E) \Join R_6(A,E)$. 
A decomposition $\calT_{\Join}$ with the lowest width of $3/2$ has two bags 
$S_1=\{A, B,C\}$ and  $S_2=\{A, D, E\}$, one for each clique (triangle) in the query. 
The application of Proposition~\ref{prop:rewriting2} 
constructs the acyclic query
 $Q' = B_1(A,B,C)\Join B_2(A,D,E)$.
The relations $B_1(A,B,C)$ and $B_2(A,D,E)$
are materializations of  
the two bags of $\calT_{\Join}$.
The database $\inst{D}' = 
\{B_1(A,B,C), B_2(A,D,E)\}$ is globally consistent 
with respect to $Q'$, i.e., each
tuple in $B_1'$ has at least one joinable tuple in $B_2'$ and 
vice versa.
The decomposition $\calT_{\Join}$ corresponds to a join tree 
of $Q'$. 
\end{example}


\section{Covers for Join Queries}
\label{sec:cores_join_results}

In this section we introduce the notion of covers of join query results along with a characterization of their size bounds, the connection to minimal edge covers for hypergraphs of join query results, and the complexity for enumerating the tuples in the query result from a cover.

Let $(Q,\calT,\inst{D})$ denote a triple of a natural join query $Q$, decomposition 
$\calT$ of $Q$, and database $\inst{D}$. 
For an instance $(Q,\calT,\inst{D})$, covers of the query result $Q(\inst{D})$ are relations that are minimal while preserving the information in the query result $Q(\inst{D})$ in the following sense.

\begin{definition}[Result Preservation]
A relation $K$ is {\em result-preserving with respect to $(Q,\calT,\inst{D})$} if its schema $\sign{K}$ is $\mathit{att}(Q)$ and $\pi_B K = \pi_B Q(\inst{D})$ for each $B \in \sign{\calT}$.   
\end{definition}

That is, for each bag $B$ in the decomposition $\calT$ of $Q$, both the relation 
$K$ and the query result $Q(\inst{D})$ have the same projection onto $B$. This also means that the natural join of these projections of $K$ is precisely $Q(\inst{D})$.

\begin{proposition}\label{Q(D)-preserv=join_preserv}
Given $(Q,\calT,\inst{D})$, a relation $K$ with schema $\mathit{att}(Q)$ is result-preserving with respect to $(Q,\calT,\inst{D})$ if and only if $\Join_{B\in\sign{\calT}} \pi_B K = Q(\inst{D})$.
 \end{proposition}

We further say that the relation $K$ is {\em minimal} result-preserving 
with respect to  $(Q,\calT,\inst{D})$
if it is result-preserving with respect to  $(Q,\calT,\inst{D})$, yet this is not the case for any strict subset of it.
We can now define the notion of covers of query results.

\begin{definition}[Covers]\label{def:cores}
Given $(Q,\calT,\inst{D})$, a {\em cover of the query result $Q(\inst{D})$ over the decomposition $\calT$} is a minimal result-preserving relation with respect to $(Q,\calT,\inst{D})$.
\end{definition}

\begin{example}\label{ex:cores}
Figure~\ref{fig:hypergraph} gives the decomposition $\calT$ of a path query and one cover $\mathit{rel}(M)$ of the query result over $\calT$. We give below four relations that are subsets of the query result. The relations $K_1$ and $K_2$ are covers, while the relations $N_1$ and $N_2$ are not covers:

\begin{center}
\begin{tabular}{llll}

\begin{tabular}{c@{\;}c@{\;}c@{\;}c@{\;}}
  \multicolumn{4}{c}{$K_1$} \\
  \toprule
  $A$ & $B$ & $C$ & $D$  \\\midrule 
  $a_1$ & $b_1$ &  $c_1$ & $d_2$  \\
  $a_2$ & $b_1$ & $c_1$  & $d_1$  \\
  $a_1$ & $b_2$ & $c_2$  & $d_2$  \\
  $a_2$ & $b_2$ & $c_2$  & $d_1$  \\
  \bottomrule
  \\
\end{tabular}

&

\begin{tabular}{c@{\;}c@{\;}c@{\;}c@{\;}}
  \multicolumn{4}{c}{$K_2$} \\
  \toprule
  $A$ & $B$ & $C$ & $D$  \\\midrule 
  $a_1$ & $b_1$ & $c_1$ & $d_2$  \\
  $a_2$ & $b_1$ & $c_1$ & $d_1$  \\
  $a_1$ & $b_2$ & $c_2$ & $d_1$  \\
  $a_2$ & $b_2$ & $c_2$ & $d_2$  \\
  \bottomrule
  \\
\end{tabular}

&

\begin{tabular}{c@{\;}c@{\;}c@{\;}c@{\;}}
  \multicolumn{4}{c}{$N_1$} \\
  \toprule
  $A$ & $B$ & $C$ & $D$  \\\midrule 
  $a_1$ & $b_1$ & $c_1$  & $d_1$  \\
  $a_1$ & $b_1$ & $c_1$  & $d_2$  \\
  $a_1$ & $b_2$ & $c_2$  & $d_1$  \\
  \bottomrule
  \\
  \\
\end{tabular}

&

\begin{tabular}{c@{\;}c@{\;}c@{\;}c@{\;}}
  \multicolumn{4}{c}{$N_2$} \\
  \toprule
  $A$ & $B$ & $C$ & $D$  \\\midrule 
  $a_1$ & $b_1$ & $c_1$ & $d_1$  \\
  $a_1$ & $b_1$ & $c_1$ & $d_2$  \\
  $a_2$ & $b_1$ & $c_1$ & $d_1$  \\
  $a_1$ & $b_2$ & $c_2$ & $d_2$  \\
  $a_1$ & $b_2$ & $c_2$ & $d_1$  \\
  \bottomrule
\end{tabular}

\end{tabular}
\end{center}

To check the minimal result-preservation property, we take projections onto the bags $B_1=\{A,B\}$, $B_2=\{B,C\}$, and $B_3=\{C,D\}$. The relation $N_1$ is not result-preserving, because 
$(a_2,b_2)\not\in\pi_{B_1}N_1$. The same argument also applies to relation $N_2$.

Consider now the coarser decomposition $\calT'$ with bags $B'_{1,2}=\{A,B,C\}$ and $B'_3=\{C,D\}$. The covers over $\calT$ discussed above are also covers over $\calT'$. The query result is the only cover over the coarsest decomposition $\calT''$ with only one bag.
\end{example}

\begin{example}\label{ex:core-sizes}
A query result may admit exponentially many covers over the same decomposition. Consider for instance the product query 
$R_1(A) \Join R_2(B)$ with relations $R_1$ and  $R_2$ of size two and respectively $n>1$. The query result has size $2\cdot n$. To compute a cover, we pair the first tuple in $R_1$ with any non-empty and strict subset of the $n$ tuples in $R_2$, while the second tuple in $R_1$ is paired with the remaining tuples in $R_2$. There are $2^n-2$ possible covers. The empty and the full sets are missing from the choice of a subset of $R_2$ as they would mean that one of the two tuples in $R_1$ would have to be paired with tuples in $R_2$ that are already paired with the other tuple in $R_1$ and that would violate the minimality criterion of the covers. All covers have size $n$ and none is contained in another.
\end{example}

We next give a characterization of covers via the hypergraph of the query result.

\begin{proposition}\label{characterize_cores} 
Given $(Q,\calT,\inst{D})$, a relation $K$ is a cover of the query result $Q(\inst{D})$ over $\calT$ if and only if the hypergraph of $Q(\inst{D})$ over 
$\sign{\calT}$ has a minimal edge cover $M$ such that $\relof{M}=K$.
\end{proposition}

\begin{example}\label{ex:cores-edgecover}
Figure~\ref{fig:hypergraph} gives a minimal edge cover $M$ and the cover $\mathit{rel}(M)$. By removing any edge from $M$, it is not anymore an edge cover. By removing the tuple corresponding to that edge from $\mathit{rel}(M)$, it is not anymore a cover since it is not result preserving. By adding an edge to $M$ or the corresponding tuple to $\mathit{rel}(M)$, they are not anymore minimal.
\end{example}

We now turn our investigation to sizes and first note the following immediate property.

\begin{proposition}\label{prop:core_subset}
Given $(Q,\calT,\inst{D})$, each cover of $Q(\inst{D})$ over $\calT$ is a subset of $Q(\inst{D})$.
\end{proposition}

An implication of Proposition~\ref{prop:core_subset} is that the covers cannot be larger than the query result. However, they can be much more succinct. 
We first give size bounds for covers using the sizes of projections of the query result onto the bags of the underlying decomposition. 
 
\begin{proposition}\label{prop:size_bounds_hypergraph}
Given $(Q,\calT,\inst{D})$, the size of each cover $K$ 
of $Q(\inst{D})$
over $\calT$ satisfies the inequalities
$\takemax{B\in\sign{\calT}}{\sizeof{\pi_BQ(\inst{D})}}$ $\leq$ 
$\sizeof{K}$ $\leq$ 
$\Sigma_{B\in \sign{\calT}}\sizeof{\pi_BQ(\inst{D})}$.
\end{proposition}

We can now characterize the size of a cover using
the width of the decomposition.

\begin{theorem}\label{theo:general_size_bounds}
Let $Q$ be a natural join query and $\calT$ a decomposition of $Q$.
\begin{enumerate}[(i)] 
\item For any database $\inst{D}$, each cover of the query result $Q(\inst{D})$ over 
$\calT$ has size $\calO(\sizeof{\inst{D}}^{\fhtw{\calT}})$.

\item There are arbitrarily large databases $\inst{D}$ such that 
each cover of the query result $Q(\inst{D})$ over $\calT$
has size $\Omega(\sizeof{\inst{D}}^{\fhtw{\calT}})$.
\end{enumerate}
\end{theorem}

The size gaps between query results and their covers can be arbitrarily large. For any join query $Q$ and database $\inst{D}$, it holds that $\sizeof{Q(\inst{D})}= \calO(\sizeof{\inst{D}}^{\rho^*(Q)})$ and there are arbitrarily large databases $\inst{D}$ for which $\sizeof{Q(\inst{D})}= \Omega(\sizeof{\inst{D}}^{\rho^*(Q)})$~\cite{AtseriasGM13}. For acyclic queries, the fractional edge cover number $\rho^*$ can be as large as $|Q|$, while the fractional hypertree width is one. Section~\ref{sec:core-plans} shows that the same gap also holds for time complexity.

\begin{example}
We continue Example~\ref{ex:cores}. The decomposition $\calT$ has width one, which is minimal. The covers over $\calT$, such as $K_1$ and $K_2$, have sizes upper bounded by the input database size. The minimum size of a cover over $\calT$ is the maximum size of a relation used in the query (assuming the relations are globally consistent). In contrast, there are arbitrarily large databases of size $N$ for which the query result has size $\Omega(N^2)$.
\end{example}

Proposition~\ref{characterize_cores} and Theorem~\ref{theo:general_size_bounds} give alternative equivalent characterizations of the size of a cover of a query result. The former gives it as the size of a {\em minimal edge cover} of the hypergraph of the query result over the attribute sets given by the bags of a decomposition $\calT$, while the latter states it using the fractional hypertree width of $\calT$ or equivalently the {\em maximum fractional edge cover} number over all the bags of $\calT$. Most notably, whereas the former is an {\em integral} number, the latter is a {\em fractional} number.

This size gap between query results and their covers is precisely the same as for query results and their factorized representations called d-representations~\cite{FDB:TODS:2015}. In this sense, covers can be seen as relational encodings of factorized representations of query results. We can easily translate covers into factorized representations. Appendix \ref{sec:intro_d-reps} gives a brief introduction to d-representations and a translation example.

\begin{proposition}\label{prop:core_to_d_rep}
Given $(Q,\calT,\inst{D})$, each cover $K$ of the query result $Q(\inst{D})$ over $\calT$ can be translated into a d-representation of $Q(\inst{D})$ of size $\calO(\sizeof{K})$ and in time $\widetilde{\calO}(\sizeof{K})$.
\end{proposition}

The above translation allows us to extend the applicability of covers to known  workloads over factorized representations, such as in-database optimization problems~\cite{NNOS2017} and in particular learning regression models~\cite{Olteanu:FactorizedDB:2016:SIGREC}. Nevertheless, it is practically desirable to process such workloads directly on covers, since this would avoid the indirection via factorized representations that comes with extra space cost and non-relational data representation. Aggregates, which are at the core of such workloads, can be computed directly on covers by joint scans of the projections of the cover onto the bags of the decomposition; alternatively, they can be computed by expressing any cover as the natural join of its bag projections and then pushing the aggregates past the join.

\begin{example}\label{aggregates_on_covers}
We consider the query $Q = R(A,B) \Join S(B,C)$ and its decomposition 
$\calT$ with bags $\{A,B\}$ and $\{B,C\}$. To compute aggregates over the 
join result $Q(\inst{D})$, we can use any cover $K$ of 
$Q(\inst{D})$ over $\calT$. The expression for counting the number of result tuples is $\sum_{b\in\text{dom}(B)}\sum_{a\in\text{dom}(A)}\sum_{c\in\text{dom}(C)} {\bf 1}_{R(a,b)}\cdot {\bf 1}_{S(b,c)}$, where $1_E$ is the Kronecker delta that is evaluated to {\bf 1} if the event $E$ is satisfied and {\bf 0} otherwise. We can compute it in one scan over $K$ if $K$ is sorted on $(B,A,C)$ or 
$(B,C,A)$. For each $B$-value $b$, we multiply the distinct numbers of $A$-values and of $C$-values paired with $b$ in $K$, and we sum up these products over all $B$-values. 
We can rewrite this expression as follows: 
$\sum_{b\in\text{dom}(B)}(\sum_{a\in\text{dom}(A)}{\bf 1}_{(a,b)\in\pi_{\{A,B\}} K})(\sum_{c\in\text{dom}(C)} {\bf 1}_{(b,c)\in\pi_{\{B,C\}} K})$. This expression only uses the pairs  $(a,b)$ and $(b,c)$ in $K$. The pairs $(a,c)$, which make the difference among covers and are the culprits for the explosion in the size of the query result, are not needed. 
\end{example}

Despite their succinctness over the explicit listing of tuples in a query result, any cover of the query result can be used to enumerate the result tuples with constant delay and extra space (data complexity) following linear-time pre-computation. In particular, the delay and the space are linear in the number of attributes of the query result which is as good as enumerating directly from the result.
This complexity follows from Proposition \ref{prop:core_to_d_rep} and the enumeration for factorized representations~\cite{FDB:TODS:2015} with constant delay and extra space. 

\begin{corollary}[Proposition \ref{prop:core_to_d_rep}, Theorem 4.11~\cite{FDB:TODS:2015}]\label{theo:const-delay-enum} 
Given $(Q,\calT,\inst{D})$, the tuples in the query result $Q(\inst{D})$ can be enumerated from any cover $K$ of $Q(\inst{D})$ over $\calT$ 
 with $\widetilde{\calO}(|K|)$ pre-computation time and
$\calO(1)$ delay and extra space.
\end{corollary}

An alternative way to achieve constant-delay enumeration with 
$\widetilde{\calO}(|K|)$ pre-computation is by noting that the acyclic join queries considered in this paper are free-connex and thus allow for enumeration with constant delay and $\softO(|\inst{D}|)$ pre-computation~\cite{Durand07}. An acyclic conjunctive query is called free-connex if its extension by a new relation symbol covering all 
attributes of the result remains acyclic~\cite{Segoufin:SIGREC:2015}. Moreover, given a cover $K$ over a decomposition $\calT$, the natural join of the projections of $K$ onto the bags of $\calT$ is an acyclic query that computes the original query result (Proposition \ref{Q(D)-preserv=join_preserv}).


 \section{Computing Covers for Join Queries using Cover-Join Plans}
 \label{sec:core-plans}
 
Given an arbitrary join query and database, we can compute covers using a {\em monolithic} algorithm akin to known algorithms for computing factorized representations of query results~\cite{Olteanu:FactorizedDB:2016:SIGREC}. However, is it possible to compute covers in a {\em compositional} way, by computing covers for one join at a time? In this section, we answer this question in the affirmative for acyclic natural join queries $Q$ and globally consistent databases $\inst{D}$ with respect to $Q$.

For a triple $(Q,\calJ,\inst{D})$, where $Q$ is an acyclic natural join query, $\calJ$ is a join tree of $Q$, and $\inst{D}$ is a database  globally consistent with respect to $Q$, we use so-called cover-join plans to compute covers of the query result $Q(\inst{D})$ over the decomposition corresponding to the join tree $\calJ$. 
Such plans follow the structure of the join tree $\calJ$ and use a new binary join operator called cover-join. The cover-join of two relations yields a cover of their natural join. This approach is in the spirit of standard relational query evaluation.
It is compositional in the sense that to compute a cover of the query result, it suffices to repeatedly compute a cover of the join of two relations. This is practical since it can be supported by existing query engines extended with the cover-join operator. We also show that, due to the binary nature of the cover-join operator, the cover-join plans cannot recover all possible covers of the query result. Furthermore, different plans may lead to different covers. Plans that do not follow the structure of a join tree may be unsound as they do not necessarily construct covers.

To compute covers for an arbitrary join query and database, we proceed in two stages.
We first materialize the bags of a decomposition of the query so as to reduce it to an acyclic query $Q$ over an extended database $\inst{D}$ that is now globally consistent with respect to $Q$ (Proposition~\ref{prop:rewriting2}). We then use a cover-join plan to compute covers of $Q(\inst{D})$. The first step has a non-trivial time complexity overhead, whereas the second step is linearithmic. Overall, this strategy is 
worst-case optimal for computing covers for arbitrary join queries and databases.

\subsection{The Cover-Join Operator}

The building block of our approach to computing covers is the binary cover-join operator.

\begin{definition}[Cover-Join]\label{def:core-join}
The cover-join of two relations $R_1$ and $R_2$, denoted by $R_1\corejoin R_2$, computes a cover of their join result over the decomposition with bags $\sign{R_1}$ and $\sign{R_2}$.
\end{definition}

Following the alternative characterization of covers of a query result by minimal edge covers in the hypergraph of the query result (Proposition~\ref{characterize_cores}), the cover-join defines 
the relation $\relof{M}$ of a minimal edge cover $M$ of the hypergraph $H$ of the result of the join $R_1 \Join R_2$ over the 
attribute sets $\sign{R_1}$ and $\sign{R_2}$. The hypergraph $H$ is bipartite and consists of disjoint complete bipartite subgraphs. Since a cover is a minimal edge cover, it corresponds to a bipartite subgraph with the same number of nodes but a subset of the edges, where all paths can only have one or two edges. A cover cannot have unconnected nodes, since it would not be an edge cover. A path of three (or more) edges violates the minimality of the edge cover: Such a path $a_1-b_1-a_2-b_2$ in a bipartite graph covers the four nodes, yet a minimal cover would only have the two edges $a_1-b_1$ and $a_2-b_2$.

We can compute a cover of a join of two relations 
$R_1$ and $R_2$ in time $\softO(|R_1| + |R_2|)$, since it amounts to computing a minimal edge cover in a collection of disjoint complete bipartite graphs that encode the join result. The smallest size of a cover is given by the edge cover number of the bipartite graph representing the join result, which is the maximum of the sizes of the two sets of nodes in the graph~\cite{lawler:matroid}. 
The largest size can be achieved in case one of the two node sets has size one, in which case this is paired with all nodes in the second set. In case both sets have more than one node, the largest size is achieved when we pair one node from one of the two node sets with all but one node in the second set and then the remaining node in the second set with all but the already used node in the first set.

For the analysis in this paper, we assume that our cover-join algorithm may return any cover of the natural join of two relations. In practice, however, it makes sense to compute a cover of minimum size. We choose this cover as follows: For each complete bipartite hypergraph in the join result with node sets $V_1$ and $V_2$  such that $|V_1|\leq|V_2|$, we choose a minimum edge cover by pairing each node in $V_1$ with one distinct node in $V_2$ and all remaining nodes in $V_2$ with one node in $V_1$.

\begin{proposition}\label{prop:comp_core_two_rel}
Given two consistent relations $R_1$ and $R_2$, the cover-join computes a cover $K$ of their join result over the decomposition with bags $\sign{R_1}$ and $\sign{R_2}$ in time $\widetilde{\calO}(|R_1|+|R_2|)$ and with size $\max\{|R_1|, |R_2|\}\leq |K|\leq |R_1|+|R_2|$.
\end{proposition}

\begin{example}
Consider again the product $R_1(A)\Join R_2(B)$ in Example~\ref{ex:core-sizes}, where $R_1=[2]$ and $R_2=[n]$ with $n > 1$. Examples of covers of size $n$ over the decomposition $\calT$ with bags $\{A\}$ and $\{B\}$ are: $\{(1,i)\mid i\in[n]-\{k\}\}\cup\{(2,k)\}$ for any $k\in[n]$; $\{(1,i)\mid i\in[k]\}\cup\{(2,j+k)\mid j\in[n-k]\}$ for any $k\in[n-1]$.
If $R_1=[m]$ with $m>n$, then examples of covers over $\calT$ 
of minimum size $m$
are: 
$\{(i,i)\mid i\in[k-1]\}\cup
\{(k-1+i , k + i) \mid i \in[n-k]\} \cup
\{(n-1+i , k) \mid i \in[m-n+1]\}$ 
for any $k\in[n]$. A cover over $\calT$ of maximal size $n+m-2$ is: $\{(1,i)\mid i\in[n-1]\}\cup\{(j+1,n)\mid j\in[m-1]\}$. Below are depictions of the complete bipartite graph corresponding to the query result for $n=4$ and $m=5$, where the edges in a minimal edge cover are solid lines and all other edges are dotted. The left minimal edge cover corresponds to a cover over $\calT$ of minimum size $m=5$, while the right minimal edge cover corresponds to a 
cover over $\calT$ of maximum size $n+m-2=7$.

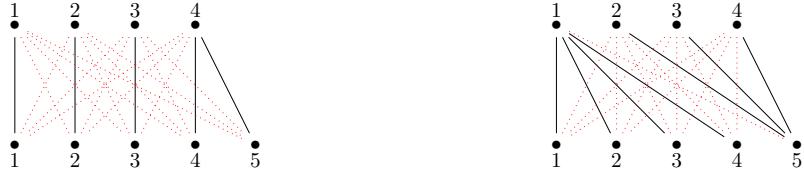
\begin{center}
\scalebox{0.8}{
\begin{tikzpicture}


\begin{scope}[rotate=-90]
\node (1)  at(-1, -4.7) {$\bullet$};
\node      [above  of =1, node distance=.25cm] {$1$};

\node (2) [right of =1, node distance=1cm] {$\bullet$};
\node     [above  of =2, node distance=.25cm] {$2$};

\node (3) [right of =2, node distance=1cm] {$\bullet$};
\node     [above  of =3, node distance=.25cm] {$3$};

\node (4) [right of =3, node distance=1cm] {$\bullet$};
\node     [above  of =4, node distance=.25cm] {$4$};


\node (11) [below of =1, node distance=2cm] {$\bullet$};
\node      [below of =11, node distance=.25cm] {$1$};

\node (12) [right of =11, node distance=1cm] {$\bullet$};
\node      [below of =12, node distance=.25cm] {$2$};

\node (13) [right of =12, node distance=1cm] {$\bullet$};
\node      [below of =13, node distance=.25cm] {$3$};

\node (14) [right of =13, node distance=1cm] {$\bullet$};
\node      [below of =14, node distance=.25cm] {$4$};

\node (15) [right of =14, node distance=1cm] {$\bullet$};
\node      [below of =15, node distance=.25cm] {$5$};

\draw[-] (1) -- (11);
\draw[-] (2) -- (12);
\draw[-] (3) -- (13);
\draw[-] (4) -- (14);
\draw[-] (4) -- (15);

\draw[dotted,red] (1) -- (12);
\draw[dotted,red] (1) -- (13);
\draw[dotted,red] (1) -- (14);
\draw[dotted,red] (1) -- (15);
\draw[dotted,red] (2) -- (11);
\draw[dotted,red] (2) -- (13);
\draw[dotted,red] (2) -- (14);
\draw[dotted,red] (2) -- (15);
\draw[dotted,red] (3) -- (11);
\draw[dotted,red] (3) -- (12);
\draw[dotted,red] (3) -- (14);
\draw[dotted,red] (3) -- (15);
\draw[dotted,red] (4) -- (11);
\draw[dotted,red] (4) -- (12);
\draw[dotted,red] (4) -- (13);
\end{scope}

\begin{scope}
\node (21) [right of =4, node distance=6cm] {$\bullet$};
\node      [above  of =21, node distance=.25cm] {$1$};

\node (22) [right of =21, node distance=1cm] {$\bullet$};
\node      [above  of =22, node distance=.25cm] {$2$};

\node (23) [right of =22, node distance=1cm] {$\bullet$};
\node      [above  of =23, node distance=.25cm] {$3$};

\node (24) [right of =23, node distance=1cm] {$\bullet$};
\node      [above  of =24, node distance=.25cm] {$4$};


\node (31) [below of =21, node distance=2cm] {$\bullet$};
\node      [below of =31, node distance=.25cm] {$1$};

\node (32) [right of =31, node distance=1cm] {$\bullet$};
\node      [below of =32, node distance=.25cm] {$2$};

\node (33) [right of =32, node distance=1cm] {$\bullet$};
\node      [below of =33, node distance=.25cm] {$3$};

\node (34) [right of =33, node distance=1cm] {$\bullet$};
\node      [below of =34, node distance=.25cm] {$4$};

\node (35) [right of =34, node distance=1cm] {$\bullet$};
\node      [below of =35, node distance=.25cm] {$5$};

\draw[-] (21) -- (31);
\draw[-] (21) -- (32);
\draw[-] (21) -- (33);
\draw[-] (21) -- (34);
\draw[-] (22) -- (35);
\draw[-] (23) -- (35);
\draw[-] (24) -- (35);

\draw[dotted,red] (21) -- (35);
\draw[dotted,red] (22) -- (31);
\draw[dotted,red] (22) -- (32);
\draw[dotted,red] (22) -- (33);
\draw[dotted,red] (22) -- (34);
\draw[dotted,red] (23) -- (31);
\draw[dotted,red] (23) -- (32);
\draw[dotted,red] (23) -- (33);
\draw[dotted,red] (23) -- (34);
\draw[dotted,red] (24) -- (31);
\draw[dotted,red] (24) -- (32);
\draw[dotted,red] (24) -- (33);
\draw[dotted,red] (24) -- (34);
\end{scope}

\end{tikzpicture}
}
\end{center}

\end{example}

\subsection{Cover-join Plans}

We now compose cover-join operators into so-called cover-join plans to compute covers for acyclic natural join queries. Before we define such plans, we need to introduce some notation.

For a join tree $\calJ$ of a query $Q$, we write $\calJ = \calJ_1\circ \calJ_2$ if $\calJ$ can be split into two non-empty subtrees $\calJ_1$ and $\calJ_2$ that are connected by a single edge in $\calJ$. Any subtree $\calJ'$ of $\calJ$ defines the subquery of $Q$ that is the natural join of all relation symbols that are nodes in $\calJ'$.

\begin{definition}[Cover-Join Plan]\label{def:core-join-plan}
Given $(Q,\calJ,\inst{D})$, a {\em cover-join plan $\varphi$ over the join tree $\calJ$} is defined recursively as follows:
\begin{itemize}
\item If $\calJ$ consists of one node $R$, then $\varphi = R$. The plan $\varphi$ returns $R$.

\item If $\calJ = \calJ_1\circ\calJ_2$ and $\varphi_{i}$ is a cover-join plan over $\calJ_i$, then $\varphi = 
\varphi_{1}\ \corejoin\ \varphi_{2}$. The plan $\varphi$ returns the result of $R_1\ \corejoin\ R_2$, where the relation $R_i$ is returned by
the plan $\varphi_{i}$  ($i\in[2]$).
\end{itemize}
\end{definition}

Lemma~\ref{lemma:core-plan-acyclic} states next that a cover-join plan computes a cover of the query result over the decomposition corresponding to a given join tree of the query.

\begin{lemma}\label{lemma:core-plan-acyclic}
Given $(Q,\calJ,\inst{D})$ where $\inst{D}=\{R_i\}_{i\in[n]}$
is globally consistent with respect to $Q$, 
any cover-join plan over the join tree $\calJ$ 
computes a cover $K$ of $Q(\inst{D})$ over the decomposition corresponding to $\calJ$ in 
time $\widetilde{\calO}(|K|)$ and with size $\takemax{{i\in[n]}}{\sizeof{R_i}} \leq|K|\leq \sum_{i\in[n]}|R_i|$.
\end{lemma}

Lemma~\ref{lemma:core-plan-acyclic} states three remarkable properties of cover-join plans. First, they compute covers compositionally: To obtain a cover of the entire query result it is sufficient to compute covers of the results for subqueries. More precisely, for a cover-join plan $\varphi_{1}\ \corejoin\ \varphi_{2}$, the sub-plans $\varphi_{1}$ and $\varphi_{2}$ compute covers for the subqueries defined by the joins of the relations in the join trees $\calJ_1$ and respectively $\calJ_2$. Then, the plan $\varphi_{1}\ \corejoin\ \varphi_{2}$ computes a cover for the join of the relations in the join tree $\calJ=\calJ_1\circ\calJ_2$. 
Second, the output of a cover-join plan is always a cover, regardless which cover is picked at each  cover-join operator in the plan.
Third, it does not matter which cover-join plan we choose for a given 
join tree, the resulting covers are computed with the same time guarantee. Nevertheless, different plans for the same 
join tree may lead to different covers (Example~\ref{ex:incomparable}).

These properties rely on the global consistency of the database and on the fact that the plans follow the structure of the join tree. For arbitrary databases, a cover-join operator may wrongly construct covers using dangling tuples at the expense of relevant tuples that are not anymore covered and therefore lost. 
Furthermore, plans that do not follow the structure of a join tree may be unsound (Example~\ref{ex:unsound}). 
Although each cover-join operator computes a cover of minimum size for the 
join of its input relations, the overall cover computed by a cover-join plan may not be a cover  of minimum size of the query result (Example~\ref{ex:cover_join_plan_non_minimal} in Appendix \ref{app:cover_join_plan_non-minimum}).

\begin{example}\label{ex:many_plans}
A join tree that admits several splits can define many plans. For instance, the join tree for the query $R_1(A,B)\Join R_2(B,C)\Join R_3(C,D)$ is the path $R_1-R_2-R_3$ and admits two possible splits that lead to the plans 
$\varphi_1 = (R_1(A,B) \corejoin R_2(B,C)) \corejoin R_3(C,D)$ and $\varphi_2 = R_1(A,B) \corejoin (R_2(B,C) \corejoin 
R_3(C,D))$. The relations are those in Figure~\ref{fig:hypergraph}, now calibrated. 
For this database, the covers computed by the sub-plans $R_1(A,B) \corejoin R_2(B,C)$  and $R_2(B,C) \corejoin R_3(C,D)$ correspond to full join results, since all join values only occur once in the relations. By taking any possible cover at each cover-join operator in the plans, both plans yield the same four possible covers of the query result: One of them is $\mathit{rel}(M)$ in Figure~\ref{fig:hypergraph} and two of them are $K_1$ and $K_2$ in Example~\ref{ex:cores}. The last cover is not depicted: It is the same as $K_1$ with the change that the values $d_1$ and $d_2$ are swapped between the first two rows. 
\end{example}

A corollary of Proposition~\ref{prop:rewriting2} and Lemma~\ref{lemma:core-plan-acyclic} is that covers over decompositions of {\em arbitrary} natural join queries  can be computed in time proportional to their sizes. 

\begin{theorem}[Proposition~\ref{prop:rewriting2}, Lemma~\ref{lemma:core-plan-acyclic}]\label{theo:time_core_join_plan}
Given a natural join query $Q$, decomposition $\calT$ of $Q$, and database $\inst{D}$, 
a cover of the query result $Q(\inst{D})$ over the decomposition $\calT$ and with size $\calO(|\inst{D}|^{\fhtw{\calT}})$ can be computed in time $\widetilde{\calO}(|\inst{D}|^{\fhtw{\calT}})$. 
\end{theorem}

Given $(Q,\calT,\inst{D})$ where $Q$ is an arbitrary natural join query
and $\inst{D}$ is an arbitrary database, 
we can compute a cover in four steps: 
construct $(Q',\calT,\inst{D}')$ such that
$Q'$ is an acyclic natural join query, $\calT$ corresponds 
to a join tree of $Q'$ and $\inst{D}'$
consists of materializations of the bags of $\calT$; 
 turn $\inst{D}'$ into a globally consistent database $\inst{D}''$ 
with respect to $Q'$; turn $\calT$ into a join tree $\calJ$ of $Q'$ by 
replacing each bag by the corresponding relation symbol in $Q'$;
and execute on $\inst{D}''$ a cover-join plan for $Q'$ over $\calJ$.      
Since there are arbitrarily large databases for which the size bounds on covers are tight (Theorem~\ref{theo:general_size_bounds}), the cover-join plans, together with a worst-case optimal algorithm for materializing bags~\cite{Ngo:SIGREC:2013}, represent a worst-case optimal  algorithm for computing covers.

We conclude this section with three insights into the ability of cover-join plans to compute covers. We give an example of an unsound cover-join plan that does not follow the structure of a join tree. We then note the incompleteness of our cover-join plans due to the binary nature of the cover-join operator. We give an example of a cover that cannot be computed with our cover-join plans, but can be computed using a multi-way cover-join operator. Finally, we give an example showing that distinct cover-join plans over the same (or also distinct) join trees can yield incomparable sets of covers. 

\begin{example}[Unsound plan]\label{ex:unsound}
Consider the query $Q = R_1(A,B)\Join R_2(B,C)\Join R_3(C,D)$,
the following database with relations  $R_1$, $R_2$, and  $R_3$, and four relations computed by cover-joining two of the three relations:

\begin{center}
\begin{tabular}{lllllll}

\begin{tabular}{c@{\;}c@{\;}}
  \multicolumn{2}{c}{$R_1$} \\
  \toprule
  $A$ & $B$  \\\midrule
  $a$ & $b_1$\\
  $a$ & $b_2$\\
  \bottomrule
\end{tabular}

&

\begin{tabular}{c@{\;}c@{\;}}
  \multicolumn{2}{c}{$R_2$} \\
  \toprule
  $B$ & $C$  \\\midrule
  $b_1$ & $c_1$\\
  $b_2$ & $c_2$\\
  \bottomrule
\end{tabular}

&

\begin{tabular}{c@{\;}c@{\;}}
  \multicolumn{2}{c}{$R_3$} \\
  \toprule
  $C$ & $D$  \\\midrule
  $c_1$ & $d$\\
  $c_2$ & $d$\\
  \bottomrule
\end{tabular}

& 

\begin{tabular}{c@{\;}c@{\;}@{\;}c@{\;}c@{\;}}
  \multicolumn{4}{c}{$K_{1,3}$} \\
  \toprule
  $A$ & $B$   & $C$   & $D$  \\\midrule
  $a$ & $b_1$ & $c_1$ & $d$\\
  $a$ & $b_2$ & $c_2$ & $d$\\
  \bottomrule
\end{tabular}

& 

\begin{tabular}{c@{\;}c@{\;}@{\;}c@{\;}c@{\;}}
  \multicolumn{4}{c}{$K'_{1,3}$} \\
  \toprule
  $A$ & $B$   & $C$   & $D$  \\\midrule
  $a$ & $b_1$ & $c_2$ & $d$\\
  $a$ & $b_2$ & $c_1$ & $d$\\
  \bottomrule
\end{tabular}

& 

\begin{tabular}{c@{\;}c@{\;}@{\;}c@{\;}}
  \multicolumn{3}{c}{$K_{1,2}$} \\
  \toprule
  $A$ & $B$   & $C$   \\\midrule
  $a$ & $b_1$ & $c_1$ \\
  $a$ & $b_2$ & $c_2$ \\
  \bottomrule
\end{tabular}

& 

\begin{tabular}{c@{\;}c@{\;}@{\;}c@{\;}}
  \multicolumn{3}{c}{$K_{2,3}$} \\
  \toprule
  $B$   & $C$   & $D$   \\\midrule
  $b_1$ & $c_1$ & $d$ \\
  $b_2$ & $c_2$ & $d$ \\
  \bottomrule
\end{tabular}
\end{tabular}
\end{center}

Following Definition~\ref{def:core-join-plan}, the plan 
$(R_1(A,B) \corejoin R_3(C,D)) \corejoin R_2(B,C)$ would require a split $\calJ_{1,3}\circ\calJ_2$ of a join tree, where  
the join tree $\calJ_{1,3}$ has two nodes $R_1$ and $R_3$ while the join tree
$\calJ_{2}$ has one node $R_2$. However, there is no join tree that allows such a split.

The cover-join $R_1(A,B) \corejoin R_3(C,D)$ computes one of the two covers $K_{1,3}$ and $K'_{1,3}$. The result of the join of $K'_{1,3}$ and $R_2$ is empty and so is the cover-join. This means that this plan does not always compute a cover, which makes it unsound.

This problem cannot occur with cover-join plans over join trees of $Q$. The only cover-join plans over join trees of $Q$  are (up to commutativity)  
$(R_1(A,B)$ $\corejoin$ $R_2(B,C))$ $\corejoin$ $R_3(C,D)$ and $R_1(A,B) \corejoin (R_2(B,C)\corejoin R_3(C,D))$. The only cover of 
$R_1(A,B) \corejoin R_2(B,C)$ is $K_{1,2}$ above, which can be cover-joined with $R_3$. The only cover of 
$R_2(B,C) \corejoin R_3(C,D)$ is 
$K_{2,3}$ above, which can be cover-joined with $R_1$.
\end{example}

\begin{example}[Cover-Join Incompleteness]
Consider the product query $Q =R_1(A)\Join R_2(B)\Join R_3(C)$,
the following database $\inst{D}$ with relations  $R_1$, $R_2$, and  
$R_3$ and one cover $K$ of the query result over the decomposition
with bags $\{A\}$, $\{B\}$, and $\{C\}$:
\begin{center}
\begin{tabular}{llll}

\begin{tabular}{c@{\;}c@{\;}}
  \multicolumn{1}{c}{$R_1$} \\
  \toprule
  $A$  \\\midrule
  $a_1$\\
  $a_2$\\
  \bottomrule
  \\
\end{tabular}

&

\begin{tabular}{c@{\;}c@{\;}}
  \multicolumn{1}{c}{$R_2$} \\
  \toprule
  $B$  \\\midrule
  $b_1$\\
  $b_2$\\
  \bottomrule
  \\
\end{tabular}

&

\begin{tabular}{c@{\;}c@{\;}}
  \multicolumn{1}{c}{$R_3$} \\
  \toprule
  $C$  \\\midrule
  $c_1$\\
  $c_2$\\
  \bottomrule
  \\
\end{tabular}

& 

\begin{tabular}{c@{\;}c@{\;}@{\;}c@{\;}}
  \multicolumn{3}{c}{$K$} \\
  \toprule
  $A$   & $B$   & $C$   \\\midrule
  $a_1$ & $b_1$ & $c_1$ \\
  $a_1$ & $b_2$ & $c_2$ \\
  $a_2$ & $b_1$ & $c_2$ \\
  \bottomrule
\end{tabular}

\end{tabular}
\end{center}
A decomposition of  $Q$ can have up to three bags
which are not included in other bags.

In case of decompositions with three bags, each bag consists of exactly one attribute. These decompositions correspond to the join trees that are permutations of the three relation symbols.
There are three possible cover-join plans (up to commutativity) over these 
join trees: 
$\varphi_1= R_1(A) \corejoin (R_2(B) \corejoin R_3(C))$, 
$\varphi_2= R_2(B) \corejoin (R_1(A) \corejoin R_3(C))$ and 
$\varphi_3= R_3(C) \corejoin (R_1(A) \corejoin R_2(B))$. None of these plans can yield the cover $K$ above. As discussed after Definition~\ref{def:core-join}, a minimal edge cover corresponding to a cover computed by a {\em binary} cover-join operator can only have paths of one or two edges. For instance, $\pi_{\{A,B\}}K$, which should correspond to a cover of $R_1(A) \corejoin R_2(B)$, has the path of three edges $b_2 - a_1 - b_1 - a_2$. The cover-join $R_1(A) \corejoin R_2(B)$ would not create this path since it corresponds to a non-minimal edge cover. Similarly, $\pi_{\{A,C\}}K$ and $\pi_{\{B,C\}}K$ have paths of three edges.

For decompositions with two bags, two of the three attributes are in the same bag. Without loss of generality, assume $A$ and $B$ are in the same bag. Following Proposition~\ref{prop:rewriting2}, this bag is covered by a new relation $R_{1,2}$ that is the product of $R_1$ and $R_2$. This means that $K$ has to be the cover of $R_{1,2}(A,B)\corejoin R_3(C)$, yet $\pi_{\{A,B\}}K$ is not 
$R_{1,2}$! 

The decomposition with one bag consisting of all three attributes has this bag covered by a new relation that is the product of the three relations. This relation is the Cartesian product of the three relations that is the full query result and different from $K=\pi_{\{A,B,C\}}K$.

We conclude that the cover  $K$ cannot be computed using cover-join plans with binary cover-join operators.
\end{example}

\begin{example}[Incomparable Sets of Covers]\label{ex:incomparable}
Consider the product query $Q=R_1(A)\Join R_2(B)\Join R_3(C)$ and 
the following database $\{R_1, R_2, R_3\}$:

\begin{center}
\begin{tabular}{llllll}

\begin{tabular}{c@{\;}c@{\;}}
  \multicolumn{1}{c}{$R_1$} \\
  \toprule
  $A$  \\\midrule
  $a_1$\\
  $a_2$\\
  \bottomrule
  \\
\end{tabular}

&

\begin{tabular}{c@{\;}c@{\;}}
  \multicolumn{1}{c}{$R_2$} \\
  \toprule
  $B$  \\\midrule
  $b_1$\\
  $b_2$\\
  \bottomrule
  \\
\end{tabular}

&

\begin{tabular}{c@{\;}c@{\;}}
  \multicolumn{1}{c}{$R_3$} \\
  \toprule
  $C$  \\\midrule
  $c_1$\\
  $c_2$\\
  $c_3$\\
  \bottomrule
\end{tabular}

& 

\begin{tabular}{c@{\;}c@{\;}@{\;}c@{\;}}
  \multicolumn{3}{c}{$K$} \\
  \toprule
  $A$   & $B$   & $C$   \\\midrule
  $a_1$ & $b_1$ & $c_1$ \\
  $a_2$ & $b_2$ & $c_2$ \\
  $a_1$ & $b_2$ & $c_3$ \\
  \bottomrule
\end{tabular}

& 

\begin{tabular}{c@{\;}c@{\;}@{\;}}
  \multicolumn{2}{c}{$K_{1,2}$} \\
  \toprule
  $A$   & $B$    \\\midrule
  $a_1$ & $b_1$  \\
  $a_2$ & $b_2$  \\
  \bottomrule
  \\
\end{tabular}

& 

\begin{tabular}{c@{\;}c@{\;}@{\;}}
  \multicolumn{2}{c}{$K'_{1,2}$} \\
  \toprule
  $A$   & $B$    \\\midrule
  $a_1$ & $b_2$  \\
  $a_2$ & $b_1$  \\
  \bottomrule
  \\
\end{tabular}

\end{tabular}
\end{center}

Let us consider the join tree $\calJ=R_1-R_2-R_3$ of $Q$.
There are (up to commutativity) two possible cover-join plans over 
$\calJ$: $\varphi_1= R_1(A) \corejoin (R_2(B) \corejoin R_3(C))$ and $\varphi_2=(R_1(A) \corejoin R_2(B)) \corejoin R_3(C)$. 
The above relation $K$ is a cover of the result of $Q$ and can be computed by $\varphi_1$, which cover-joins $R_1(A)$ and a cover of the join of $R_2(B)$ and $R_3(C)$. This cover cannot be computed by $\varphi_2$. Indeed, $\varphi_2$ first cover-joins 
$R_1(A)$ and $R_2(B)$, yielding $K_{1,2}$ or $K'_{1,2}$ as the only possible covers. Then, cover-joining any of them with $R_3(C)$ does not yield the cover $K$ since 
$\pi_{\{A,B\}}K$ is different from both $K_{1,2}$ and $K'_{1,2}$. Similarly, $\varphi_2$ computes covers that cannot be computed by $\varphi_1$.
\end{example}
\section{Covers for Functional Aggregate Queries}
\label{sec:applications}
We first give a brief introduction to functional aggregate queries (FAQ)
\cite{FAQ:PODS:2016}. A detailed description can be found in the 
appendix. 

Given an attribute set $S$, we use $\vala_S$ to 
indicate that tuple $\vala$ has schema $S$. For $S' \subseteq S$,
we denote by $\vala_{S'}$ the restriction of 
$\vala$ to $S'$.
A functional aggregate query has the following form 
 (slightly adapted to our notation):

\begin{equation}
\varphi(\vala_{\{A_1,\ldots,A_f\}}) = \underset{a_{f+1} \in \dom(A_{f+1})}{\bigoplus\ ^{(f+1)}} \cdots \underset{a_{n} \in \dom(A_{n})}{\bigoplus\ ^{(n)}}\  \ \underset{S\in\mathcal{E}}{\bigotimes}\ \psi_S(\vala_S),\text{ where:}\label{eq:faq}
\end{equation}

\begin{itemize}
\item $H=( \calV, \calE)$ is the multi-hypergraph of the query with ${\cal V}=\{A_i\}_{i\in[n]}$.

\item $\textsf{Dom}$ is a fixed (output) domain, such as $\{$\textsf{true},\textsf{false}$\}$, $\{0,1\}$, or $\mathbb{R}^+$.

\item $\calV_{\text{free}} = \{A_1,\ldots,A_f\}$ is the set of result or free attributes; all other attributes are bound.

\item For each attribute $A_i$ with $i>f$, $\oplus^{(i)}$ is a binary (aggregate) operator on the domain $\textsf{Dom}$. Different bound attributes may have different aggregate operators. 

\item For each attribute $A_i$ with $i>f$, either $\oplus^{(i)}$ is $\otimes$ or $(\textsf{Dom},\oplus^{(i)},\otimes)$ forms a commutative semiring with the same additive identity ${\bf 0}$ and multiplicative identity ${\bf 1}$ for all semirings.

\item For every hyperedge $S$ in $\cal E$, 
$\psi_S: \prod_{A\in S}\dom(A) \rightarrow \Dom$ 
is an (input)  function.
\end{itemize}

FAQs are a semiring generalization of aggregates over join queries, where the aggregates are the operators $\oplus^{(i)}$ and the natural join 
is expressed by $\bigotimes_{S\in\mathcal{E}} \psi_S({\sf a}_S)$.
The listing representation $R_{\psi_S}$ of a function $\psi_S$ is a relation over the schema $S \cup \{\psi_S(S)\}$ which consists of all input-output pairs for $\psi_S$ where the output is non-zero, i.e., $R_{\psi_S}$ contains a tuple $\vala_{S \cup \{\psi_S(S)\}}$ if and only if $\psi_S(\vala_S) = \vala_{\psi_S(S)} \neq \zero$. An input database for $\varphi$ contains for each $\psi_S$ its listing representation.
We say that $\calT$ is 
a decomposition of $\varphi$
if $\calT$ is a decomposition of the hypergraph $H$
of $\varphi$.
Given an FAQ $\varphi$ and database $\inst{D}$, the FAQ-problem is to compute the query result $\varphi(\inst{D})$.

Each FAQ $\varphi$ has an FAQ-width $\faqw{\varphi}$ 
which is 
defined similarly to the fractional hypertree width 
of the hypergraph of $\varphi$. 
For instance, in case where  all attributes of $\varphi$ are free,
$\faqw{\varphi}$ is equal to the fractional hypertree width of the 
hypergraph of $\varphi$.

Given an FAQ $\varphi$ and a database $\inst{D}$,
the \textsf{InsideOut} algorithm \cite{FAQ:PODS:2016} 
solves the FAQ-problem as follows. 
First, it 
eliminates all bound attributes
along with their corresponding aggregate operators
by performing equivalence-preserving  transformations
on $\varphi$. Then, it computes the listing representation 
of the remaining query. 
The algorithm runs in time 
$\widetilde{\calO}(|\inst{D}|^{\faqw{\varphi}} + Z)$ 
where 
$Z$ is the size 
of the output, i.e., the listing representation of $\varphi$.

We can compute  a cover of the result of a given FAQ $\varphi$
in time $\widetilde{\calO}(|\inst{D}|^{\faqw{\varphi}})$, which 
does not depend on the size of the listing representation  
of $\varphi$. 
Our strategy is as follows. 
We first eliminate all bound attributes 
in $\varphi$ by using \textsf{InsideOut} resulting in an
FAQ $\varphi'$. We then take a 
decomposition $\calT$ of $\varphi'$ and compute 
bag functions $\beta_B$, $B \in \sign{\calT}$,
with $\varphi'(\vala_{\calV_{\text{free}}}) = 
  \bigotimes_{B \in \sign{\calT}} \beta_B(\vala_{B})$.
Finally, we compute a cover of the join result of the listing representations   
of the bag functions over the extension of $\calT$ that
 contains, for each bag $B$, the attribute 
$\beta_B(B)$ for the values of the function 
$\beta_B$. Keeping the $\beta_B(B)$-values 
of the bag functions in the cover is necessary 
for recovering the output values 
of $\varphi$ when enumerating the result 
of $\varphi$ from the cover.

\begin{example}\label{ex:faq}
We consider the following FAQ $\varphi$ over the sum-product semiring $(\mathbb{N},+,\cdot)$ (for simplicity we skip the explicit iteration over the domains of the attributes in $\varphi$):
\begin{align*}
\varphi(a,b,d) = \sum_{c,e,f,g,h} \psi_1(a,b,c)\cdot \psi_2(b,d,e)\cdot \psi_3(d,e,f)\cdot \psi_4(f,h)\cdot \psi_5(e,g), \mbox{ where }
\end{align*}
$\varphi$, $\psi_1$, $\psi_2$, 
$\psi_3$, $\psi_4$ and $\psi_5$ are over
$\{A,B,D\}$, $\{A,B,C\}$, $\{B,D,E\}$,
$\{D,E,F\}$, $\{F,H\}$
and $\{E,G\}$, respectively.   
We first run \textsf{InsideOut} on $\varphi$ to eliminate the bound attributes
and obtain the following FAQ:

\begin{align*}
\varphi'(a,b,d) &= \underbrace{\big(\sum_{c} \psi_1(a,b,c)\big)}_{\psi_6(a,b)}\cdot \underbrace{\sum_e \big(\psi_2(b,d,e)\cdot \underbrace{\sum_f \big(\psi_3(d,e,f)\cdot \underbrace{\sum_h \psi_4(f,h)}_{\psi_7(f)}\big)}_{\psi_{9}(d,e)}\cdot \underbrace{\sum_g \psi_5(e,g)}_{\psi_8(e)}\big)}_{\psi_{10}(b,d)}.
\end{align*}
We consider the decomposition $\calT$
of $\varphi'$ with two bags 
$B_1 = \{A,B \}$ and $B_2 = \{B,D\}$ and bag functions
$\psi_6$ and respectively $\psi_{10}$. 
Then, we 
execute the cover-join plan 
$R_{\psi_6}\ \corejoin\ R_{\psi_{10}}$
over the extended decomposition $\calT'$ with bags
$\{A,B, \psi_{6}(A,B)\}$ and $\{B,D, \psi_{10}(B,D)\}$.
While the computation of the result of $\varphi'$
can take quadratic time, the above cover-join plan takes linear time. 
We exemplify the computation of the cover-join plan. Assume the following tuples in $\psi_6$ and $\psi_{10}$, 
where $\gamma_1,\ldots,\gamma_4,\delta_1,\ldots,\delta_3\in\mathbb{N}$:

\vspace*{-1em}
\begin{center}
\begin{tabular}{l@{\hspace*{2em}}l@{\hspace*{2em}}l@{\hspace*{2em}}l}

\begin{tabular}{cc|c}
  \multicolumn{3}{c}{$\psi_6$} \\
  \toprule
  $A$    &  $B$   & $\psi_{6}(A,B)$\\\midrule
  $a_1$  &  $b_1$ & $\gamma_1$\\
  $a_2$  &  $b_1$ & $\gamma_2$\\
  $a_3$  &  $b_2$ & $\gamma_3$\\
  $a_4$  &  $b_2$ & $\gamma_4$\\
  \bottomrule
\end{tabular}

&

\begin{tabular}{cc|c}
  \multicolumn{3}{c}{$\psi_{10}$} \\
  \toprule
  $B$     & $D$   &  $\psi_{10}(B,D)$\\\midrule
  $b_1$   & $d_1$ &  $\delta_1$\\
  $b_1$   & $d_2$ &  $\delta_2$\\
  $b_2$   & $d_3$ &  $\delta_3$\\
  \bottomrule
  \\
\end{tabular}

&

\begin{tabular}{ccc|cc}
  \multicolumn{5}{c}{$K$} \\
  \toprule
  $A$     &  $B$   &  $D$   &  $\psi_6(A,B)$  &  $\psi_{10}(B,D)$\\\midrule
  $a_1$   &  $b_1$ &  $d_1$ &  $\gamma_1$  &  $\delta_1$ \\
  $a_2$   &  $b_1$ &  $d_2$ &  $\gamma_2$  &  $\delta_2$ \\
  $a_3$   &  $b_2$ &  $d_3$ &  $\gamma_3$  &  $\delta_3$ \\
  $a_4$   &  $b_2$ &  $d_3$ &  $\gamma_4$  &  $\delta_3$ \\
  \bottomrule
\end{tabular}

\end{tabular}
\end{center}
\vspace*{1em}

The relation $K$ is a possible 
cover computed by the cover-join plan. 
The cover carries over the aggregates in columns 
$\psi_{6}(A,B)$  and  $\psi_{10}(B,D)$, one per bag
of $\calT'$. The aggregate of the first tuple in 
$K$ is $\gamma_1\cdot\delta_1$ (or $\gamma_1\otimes\delta_1$ under a semiring with multiplication $\otimes$).
\end{example}

The following theorem relies
on Lemma 
\ref{lemma:core-plan-acyclic} 
and Theorem \ref{theo:time_core_join_plan}
that give an 
upper bound on the time complexity for constructing 
covers of join results. 

\begin{theorem}\label{theo:core_FAQ_upper_bound}
For any FAQ $\varphi$ and database $\inst{D}$, a cover of the query result $\varphi(\inst{D})$
can be computed in time 
 $\widetilde{\calO}(\sizeof{\inst{D}}^{\faqw{\varphi}})$.
\end{theorem}

Any enumeration algorithm for covers of join results can 
be used to enumerate the tuples of an FAQ result from one of its covers.  We thus have the following corollary: 

\begin{corollary}[Corollary \ref{theo:const-delay-enum}]
\label{cor:enum_faq} 
Given a cover $K$
of the result $\varphi(\inst{D})$ of an FAQ $\varphi$ over a database $\inst{D}$,
the tuples in the query result $\varphi(\inst{D})$ 
can be enumerated
with $\widetilde{\calO}(|K|)$ pre-computation time and $\calO(1)$ delay and extra space.
\end{corollary}
\section{Conclusion}

Results of join and functional aggregate queries 
entail redundancy in both their computation and representation. In this paper we propose the notion of covers of query results to reduce such redundancy. While covers can be more succinct than the query results, they nevertheless enjoy desirable properties such as listing representation and constant-delay enumeration of result tuples. 
For a given database and a join or functional aggregate query, the query result can be normalized as a globally consistent database over an acyclic schema. Covers represent one-relational, lossless, linear-size encodings of such normalized databases.

\begin{definition}
\begin{quote}
{\bf borged} \texttt{/b\^orjd/} : Buy One Relation, Get Entire Database!
\end{quote}
\end{definition}

\subparagraph*{Acknowledgements.} This work has received funding from the European Union's Horizon 2020 research and innovation programme under grant agreement 682588. The authors would like to thank Milos Nikolic, Max Schleich, and the anonymous reviewers for their feedback on drafts of this paper, and Yu Tang for inspiring discussions that led to the concept of cover as a relational alternative to factorized representations.

\bibliographystyle{abbrv}
\bibliography{bibliography}
\appendix

\section{Further Preliminaries}
\label{sec:prel_proofs}

We introduce necessary notation for the 
proofs in the following sections.

\smallskip
{\noindent\bf Restrictions of Queries and Databases.}
Given a set $X$ of attributes and a natural join query 
$Q= R_1 \Join \ldots \Join R_n$, 
the $X$-restriction  of $Q$ is defined 
as $Q_X = R_1^X \Join \ldots \Join R_n^X$
where each $R_i^X$ results from $R$ by restricting 
its schema to $X$.
Likewise, we obtain the  
$X$-restriction $\inst{D}_X$ of
a database $\inst{D}$ by projecting each relation
in $\inst{D}$ onto the attributes in $X$.

\section{From Covers to D-Representations}\label{sec:intro_d-reps}

We next give a brief introduction to 
d-representations; for a detailed description, 
we refer the reader to the literature~\cite{FDB:TODS:2015}. We then discuss a translation from covers to d-representations. 

\begin{figure}[h]
\begin{center}

\begin{tabular}{lllll}

\begin{tabular}{c@{\;}c@{\;}}
  \multicolumn{2}{c}{$R_1$} \\
  \toprule
  $A$ & $B$  \\\midrule
  $a_1$ & $b_1$\\
  $a_2$ & $b_1$\\
  $a_3$ & $b_2$\\
  $a_4$ & $b_2$\\  
  \bottomrule
  \\
  \\
  \\
  \\
\end{tabular}

&

\begin{tabular}{c@{\;}c@{\;}}
  \multicolumn{2}{c}{$R_2$} \\
  \toprule
  $B$ & $C$  \\\midrule
  $b_1$ & $c_1$\\
  $b_2$ & $c_1$\\   
  \bottomrule
  \\
  \\
  \\
  \\
  \\
  \\
  \end{tabular}

&

\begin{tabular}{c@{\;}c@{\;}}
  \multicolumn{2}{c}{$R_3$} \\
  \toprule
  $C$ & $D$  \\\midrule
  $c_1$ & $d_1$\\
  $c_1$ & $d_2$\\
  \bottomrule
  \\
  \\
  \\
  \\
  \\
  \\
\end{tabular}

&

\begin{tabular}{c@{\;}c@{\;}c@{\;}c@{\;}c@{\;}c@{\;}}
  \multicolumn{4}{c}{$Q(\inst{D})$} \\
  \toprule
  $A$ & $B$ & $C$ & $D$  \\\midrule 
  {$a_1$} & {$b_1$} & {$c_1$} &  {$d_1$}\\
  {$a_1$} & {$b_1$} & {$c_1$} &  {$d_2$}\\
  {$a_2$} & {$b_1$} & {$c_1$} &  {$d_1$}\\
  {$a_2$} & {$b_1$} & {$c_1$} &  {$d_2$}\\
  {$a_3$} & {$b_2$} & {$c_1$} & {$d_1$}\\
  {$a_3$} & {$b_2$} & {$c_1$} & {$d_2$}\\
  {$a_4$} & {$b_2$} & {$c_1$} & {$d_1$}\\
  {$a_4$} & {$b_2$} & {$c_1$} & {$d_2$}\\  
  \bottomrule
\end{tabular}

&

\begin{tabular}{c@{\;}c@{\;}c@{\;}c@{\;}c@{\;}c@{\;}}
  \multicolumn{4}{c}{$K\subseteq Q(\inst{D})$} \\
  \toprule
  $A$ & $B$ & $C$ & $D$  \\\midrule 
  {$a_1$}  & {$b_1$} & {$c_1$} & {$d_1$}\\
  {$a_2$}  & {$b_1$} & {$c_1$} & {$d_1$}\\
  {$a_3$}  & {$b_2$} & {$c_1$} & {$d_2$}\\
  {$a_4$}  & {$b_2$} & {$c_1$} & {$d_2$}\\
  \bottomrule
  \\
  \\
  \\
  \\
\end{tabular}
\end{tabular}\\[1em]
\end{center}

\begin{center}
\scalebox{0.8}{
\begin{tikzpicture}
     \node (T)  at(-3, -4.5) {$\calT$:};
     \node[draw, ellipse] (1)  at(-1, -4.5) {$B$};

        \node[draw, ellipse] (2) [below left of =1, node distance=2.5cm] {$A,B$};
        \node[draw, ellipse] (3) [below right of =1, node distance=2.5cm] {
        $B,C$};
        
        \node[draw, ellipse] (4) [below of =3, node distance=2cm] {
        $C,D$};

        \draw (1)--(2);
        \draw (1)--(3);
        \draw (3)--(4);
\end{tikzpicture}
}
\hspace{2cm}
\scalebox{0.8}{
\begin{tikzpicture}
     \node (T)  at(-3, -4.5) {$\calT'$:};
      \node (1)  at(-1, -4.5) {$B$};

        \node (2) [below left of =1, node distance=2cm] {$A$};

        \node (3) [below right of =1, node distance=2cm] {$C$};

        \node (4) [below of =3, node distance=1.5cm] {$D$};
        
        \draw (1)--(2);
        \draw (1)--(3);
        \draw (3)--(4);        

        \node (k1) at(0.5, -4.5) {$key(B)=\emptyset\hspace*{1.5em}$};
        \node (k2) [left of =2, node distance=1.4cm] {$key(A) =\{B\}$};
        \node (k3) [right of =3, node distance=1.5cm] {$key(C)= \{B\}$};
        \node (k4) [right of =4, node distance=1.5cm] {$key(D)=\{C\}$};
\end{tikzpicture}
}
\end{center}
\caption{Top row: database $\inst{D}=\{R_1,R_2, R_3\}$, the result $Q(\inst{D})$ of the path query $Q = R_1\Join R_2\Join R_3$, and a cover $K\subseteq Q(\inst{D})$ over the decomposition $\calT$; bottom row: decomposition $\calT$ of $Q$ and an equivalent d-tree $\calT'$.}
\label{fig:tree_dec_d_tree}
\end{figure}

\subsection{D-Representations in a Nutshell}
D-representations are a succinct and lossless representation 
for relational data. A d-representation is a set of named relational algebra expressions 
$\{N_1: =E_1, \ldots , N_n:= E_n\}$, where each
$N_i$ is a unique name (or a {\em pointer}) and each $E_i$  is a relational algebra expression  with unions, Cartesian products, singleton relations, 
i.e., unary relations with one tuple, and name references in place of singleton relations. The size $\sizeof{E}$ of a d-representation $E$
is the number of its singletons. 

We consider a special class of d-representations that encode results of join queries and whose nesting structure is given by  so-called d-trees. In the literature, d-trees are defined as orderings on query variables. We give here an alternative, equivalent definition that is in line with our notion of fractional hypertree decomposition.
Given a query $Q$, a d-tree of $Q$ is a decomposition
of $Q$ where each bag is partitioned into one
attribute $A$, called the {\em bag attribute}, and a 
set of attributes, called the {\em key} of $A$ and
denoted by $\mathit{key}(A)$. There is one bag per distinct attribute $A$ in $Q$.
Each decomposition $\calT$ of a query $Q$ 
can be translated into a d-tree $\calT'$ of $Q$
with $\fhtw{\calT'} \leq \fhtw{\calT}$ 
(Proposition 9.3 
in \cite{FDB:TODS:2015}).
Given a query $Q$, a d-tree $\calT$ of $Q$, and a database 
$\inst{D}$, a d-representation $E$ of $Q(\inst{D})$ over $\calT$ with 
size $\calO(\sizeof{\inst{D}}^{\fhtw{\calT}})$ can be computed 
in time $\softO(\sizeof{\inst{D}}^{\fhtw{\calT}})$ 
(Theorem 7.13 and Proposition 8.2 in \cite{FDB:TODS:2015}).

\begin{example}
We consider the path query 
$Q = R_1(A,B) \Join R_2(B,C) \Join R_3(C,D)$.
Figure~\ref{fig:tree_dec_d_tree} depicts a database with 
relations $R_1$, $R_2$ and $R_3$ and the result of $Q$ 
over the input database $\{R_1, R_2, R_3\}$.  It also shows  
a decomposition $\calT$ of $Q$  and a cover $K$
of the query result over $\calT$. Finally, it depicts
a d-tree $\calT'$ (right below) derived from $\calT$ by using the translation in 
the proof of 
Proposition 9.3 
in \cite{FDB:TODS:2015}.

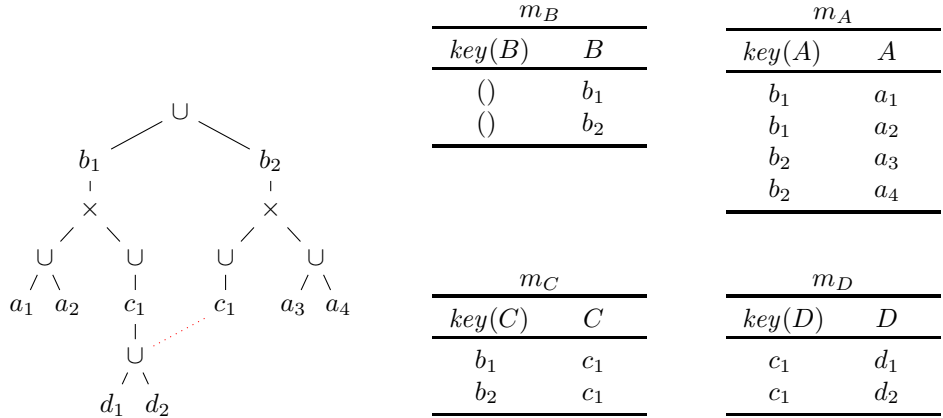
\begin{figure}[t]
\begin{center}
    \begin{tikzpicture}[xscale=0.3, yscale=0.65]
    
   \node at (0, 0) (u1) {$\cup$};

      \node at (-4, -1) (b1) {$b_1$} edge[-] (u1);
      \node at (4, -1)  (b2) {$b_2$} edge[-] (u1);

      \node at (-4, -2) (p1) {$\times$} edge[-] (b1);
      \node at (4, -2)  (p2) {$\times$} edge[-] (b2);

      \node at (-6, -3) (u2) {$\cup$} edge[-] (p1);
      \node at (-2, -3) (u3) {$\cup$} edge[-] (p1);
      \node at (2, -3)  (u4) {$\cup$} edge[-] (p2);
      \node at (6, -3)  (u5) {$\cup$} edge[-] (p2);

      \node at (-7, -4) (a1) {$a_1$} edge[-] (u2);
      \node at (-5, -4) (a2) {$a_2$} edge[-] (u2);
      \node at (-2, -4) (c1) {$c_1$} edge[-] (u3);
      \node at (5, -4) (a11) {$a_3$} edge[-] (u5);
      \node at (7, -4) (a21) {$a_4$} edge[-] (u5);

      \node at (-2, -5) (u6) {$\cup$} edge[-] (c1);
     \node at (2, -4) (c2) {$c_1$} edge[-] (u4);

      \node at (-3, -6) (d1) {$d_1$} edge[-] (u6);
      \node at (-1, -6) (d2) {$d_2$} edge[-] (u6);

      \draw[dotted,red] (c2) -- (u6);

    \end{tikzpicture}
    \hspace{0.5cm}
\newcolumntype{K}[1]{>{\centering\arraybackslash}p{#1}}
    \begin{tikzpicture}[xscale=0.3, yscale=0.65]
        \node at (-6, 4.5) {
        \begin{tabular}{ K{1.0cm}  K{1.0cm} }
        \multicolumn{2}{c}{$m_B$} \\
     	\toprule   
  	$\mathit{key}(B)$ & $B$ \\
	\toprule
	$()$ & $b_1$\\
	$()$ & $b_2$\\
	\bottomrule
  	\end{tabular}  
	};
	
        \node at (7, 3.8) {
        \begin{tabular}{ K{1.0cm}  K{1.0cm} }
        \multicolumn{2}{c}{$m_A$} \\
     	\toprule 
  	$\mathit{key}(A)$ & $A$ \\
    \toprule 
	$ b_1 $ & $a_1$\\
	$ b_1 $ & $a_2$\\
	$ b_2 $ & $a_3$\\
	$ b_2 $ & $a_4$\\
	\bottomrule
  	\end{tabular}  
	};
	
        \node at (-6, -1) {
      \begin{tabular}{ K{1.0cm}  K{1.0cm} }
        \multicolumn{2}{c}{$m_C$} \\
        \toprule 
  	$\mathit{key}(C)$ & $C$ \\
    \toprule 
	$ b_1 $ & $c_1$\\
	$ b_2 $ & $c_1$\\
	\bottomrule
  	\end{tabular}  
	};
	
        \node at (7, -1) {
      \begin{tabular}{ K{1.0cm}  K{1.0cm} }
        \multicolumn{2}{c}{$m_D$} \\
        \toprule 
  	$\mathit{key}(D)$ & $D$ \\
        \toprule 
	$ c_1$ & $d_1$\\
	$ c_1$ & $d_2$\\
	\bottomrule
  	\end{tabular}  
	};

  \end{tikzpicture}  
\end{center}
\caption{A d-representation encoded as a parse graph (left) and as 
a set of multimaps (right).}
\label{fig:d_rep_d_hashmap_rep}
\end{figure}

D-representations admit encoding as parse graphs and sets of multi-maps.  
Figure \ref{fig:d_rep_d_hashmap_rep} visualizes  
the d-representation  
of the query result from Figure \ref{fig:tree_dec_d_tree} over the 
d-tree $\calT'$ in the forms of a parse graph and of 
multi-maps.
The parse graph follows the structure of the d-tree. 
At the top level we have a union of $B$-values.
Then, given any $B$-value, the $A$-values are independent 
of the values for $C$ and $D$. Therefore, under each $B$-value, 
the $A$-values are represented in a different branch than the values 
for $C$ and $D$.  
Within the branches for $C$ and $D$, 
the values are first grouped by $C$ and then by
$D$. The information on keys is used to share subtrees across branches.
Since the key of attribute $D$ is $C$ only, 
all  $C$-nodes with the same value point to the same union 
of $D$-values.
In our example,
both $c_1$-nodes point to the same set $\{d_1,d_2\}$ of $D$-values.   

The cover $K$ from Figure \ref{fig:tree_dec_d_tree}  
can be mapped immediately to the parse graph: Under each product node, we take a minimum number of combinations of its children to ensure that every value under the product node occurs in one of these combinations.
To enumerate the tuples in the query result, it suffices to choose in turn one branch of each union node and all branches of each product node. For instance, the left product node represents the combinations of $\{a_1,a_2\}$ with $\{d_1,d_2\}$, together with the values $b_1$ and $c_1$. There are four combinations, so four tuples in the result. The first two tuples in the cover represent two of them, yet they are sufficient to recover all these tuples.

The multi-map encoding of a d-representation 
consists of one multi-map for each bag attribute: $m_A$ maps tuples over the attributes in $key(A)$ to (possibly several) values of $A$. Figure~\ref{fig:d_rep_d_hashmap_rep} shows these maps as relations whose columns are distinctly separated into those for the key attributes (the map keys) and the column for the attribute $A$ itself (the map payload). We have,
 for instance, $m_A(b_1) = a_1$ and $m_A(b_1) = a_2$, whereas
 $m_C(b_1) = c_1$. Since  $\mathit{key}(A)=\{B\}$
and there are two $B$-values in the d-representation leading to the 
sets $\{a_1,a_2\}$ and $\{a_3,a_4\}$, respectively, 
 $m_A$ maps the $B$-value $ b_1$ to both $A$-values $a_1$ and $a_2$ and the $B$-value $b_2$ to both $A$-values $a_3$ and $a_4$. 
\end{example}

\subsection{Translating Covers into D-Representations}
\begin{figure}[t]
\begin{center}
\begin{tabular}{|l|}\hline
{\bf cover2factorization} (cover $K$, decomposition $\calT$)\\\hline 
convert $\calT$ into an equivalent d-tree $\calT'$
following Proposition 9.3 
in \cite{FDB:TODS:2015}; \\  
$\LET\STAB {\bf V}$ be the set of attributes in $\calT'$;\\ 
$\FOREACH$ attribute $A\in{\bf V}\STAB\DO\STAB$ \\
$\TAB\TAB$ create multi-map $m_A: \prod_{X\in\mathit{key}(A)} \mbox{dom}(X) \mapsto \mbox{dom}(A)$;\\

  $\FOREACH$ tuple $t \in K \STAB\DO\STAB$ \\
 
 $\TAB\TAB\FOREACH$ attribute  $A\in{\bf V} \STAB\DO$ \\
 
 $\TAB\TAB\TAB\TAB$ insert assignment $\pi_{key(A)} t \mapsto \pi_A t$ into $m_A$;\\
  
 $\RETURN$ $\{m_A\}_{A\in{\bf V}}$;\\\hline
\end{tabular}
\end{center}
\caption{Translating a cover $K$ over a decomposition $\calT$ into an equivalent d-representation.}
\label{fig:factorize_from_cover}
\end{figure}

Figure~\ref{fig:factorize_from_cover} gives an algorithm that constructs an equivalent d-representation from a cover over a decomposition. Both the cover $K$ and the output d-representation are for the same query result $Q(\inst{D})$ of a query $Q$. The decomposition $\calT$ is for the query $Q$.

The algorithm creates a multi-map for each attribute $A$ and populates 
it with assignments of tuples over the keys of $A$ to the values of $A$ as encountered in the tuples of the cover.

\begin{example} 
We consider the cover $K$ over the decomposition
$\calT$ in 
 Figure~\ref{fig:tree_dec_d_tree} and the d-tree $\calT'$  equivalent to $\calT$. Following the algorithm in 
 Figure~\ref{fig:factorize_from_cover}, the cover $K$
 is translated into a d-representation over $\calT'$ as follows. 
After reading the first tuple $(a_1,b_1,c_1,d_1)$,  
 we add $() \mapsto b_1$ to $m_B$, 
 $b_1 \mapsto  a_1$ to $m_A$,
 $b_1 \mapsto  c_1$ to $m_C$, and
 $c_1 \mapsto  d_1$ to $m_D$,
 where $()$ means the empty tuple. 
 After processing the second tuple
 $(a_2,b_1,c_1,d_1)$, we only change 
 $m_A$ by adding
$b_1 \mapsto  a_2$ to $m_A$.
After the third tuple $(a_3,b_2,c_1,d_2)$, we add the following new assignments:
$() \mapsto b_2$ to $m_B$, 
$b_2 \mapsto  a_3$ to $m_A$,
$b_2 \mapsto  c_1$ to $m_C$, and
$c_1 \mapsto  d_2$ to $m_D$.
After reading the last tuple $(a_4,b_2,c_1,d_2)$, we add the new 
assignment
$b_2 \mapsto  a_4$ to $m_A$.
\end{example}

\section{Cover-Join Plans Computing Covers of Non-Minimum Size}
\label{app:cover_join_plan_non-minimum}

\begin{example}\label{ex:cover_join_plan_non_minimal}
We consider the acyclic natural join query 
$Q=R_1(A,B)\Join R_2(B,C)\Join R_3(C,D)$,
the database $\inst{D}=\{R_1, R_2, R_3\}$ globally consistent with respect to $Q$,
 and the join tree $\calJ=R_1-R_2-R_3$. The relations $R_i$ are depicted below. 

\begin{center}
\begin{tabular}{lllllll}

\begin{tabular}{c@{\;}c@{\;}}
  \multicolumn{2}{c}{$R_1$} \\
  \toprule
  $A$ & $B$  \\\midrule
  $a_1$ & $b_1$\\
  $a_2$ & $b_1$\\
  $a_3$ & $b_1$\\
  \bottomrule
\\
\end{tabular}

&

\begin{tabular}{c@{\;}c@{\;}}
  \multicolumn{2}{c}{$R_2$} \\
  \toprule
  $B$ & $C$  \\\midrule
  $b_1$ & $c_1$\\
  $b_1$ & $c_2$\\
  \bottomrule
\\
\\
\end{tabular}

&

\begin{tabular}{c@{\;}c@{\;}}
  \multicolumn{2}{c}{$R_3$} \\
  \toprule
  $C$ & $D$  \\\midrule
  $c_1$ & $d_1$\\
  $c_2$ & $d_1$\\
  $c_2$ & $d_2$\\
  \bottomrule
  \\
\end{tabular}

& 

\begin{tabular}{c@{\;}c@{\;}c@{\;}c@{\;}}
  \multicolumn{4}{c}{$K$} \\
  \toprule
  $A$ & $B$   & $C$   & $D$  \\\midrule
  $a_1$ & $b_1$ & $c_1$ & $d_1$\\
  $a_2$ & $b_1$ & $c_2$ & $d_1$\\
  $a_3$ & $b_1$ & $c_2$ & $d_2$\\
  \bottomrule
  \\
\end{tabular}

& 

\begin{tabular}{c@{\;}c@{\;}c@{\;}}
  \multicolumn{3}{c}{$K_{1,2}$} \\
  \toprule
  $A$ & $B$   & $C$   \\\midrule
  $a_1$ & $b_1$ & $c_1$ \\
  $a_2$ & $b_1$ & $c_1$ \\
  $a_3$ & $b_1$ & $c_2$ \\
  \bottomrule
  \\
\end{tabular}

& 

\begin{tabular}{c@{\;}c@{\;}c@{\;}c@{\;}}
  \multicolumn{4}{c}{$K'$} \\
  \toprule
$A$ & $B$   & $C$   & $D$   \\\midrule
$a_1$ &  $b_1$ & $c_1$ & $d_1$ \\
$a_2$ &  $b_1$ & $c_1$ & $d_1$ \\
$a_3$ &  $b_1$ & $c_2$ & $d_1$ \\
$a_3$ &  $b_1$ & $c_2$ & $d_2$ \\
  \bottomrule
\end{tabular}
\end{tabular}
\end{center}

\noindent 
The relation $K$ is a cover of the query result $Q(\inst{D})$
over the decomposition $\calT$ corresponding to $\calJ$. 
It follows from Proposition \ref{prop:size_bounds_hypergraph} 
that every cover of $Q(\inst{D})$ over  $\calT$ 
must have size 
at least three. Hence, $K$ is a minimum-sized 
cover of $Q(\inst{D})$ over  $\calT$.

We take the cover-join plan 
$(R_1(A,B) \corejoin R_2(B,C)) \corejoin R_3(C,D)$
over  $\calJ$ and assume that the cover-join operator 
computes for each two input relations $R$ and $R'$,
a minimum-sized cover of $R \Join R'$ over the 
decomposition with 
bags $\sign{R}$ and $\sign{R'}$.
Then, a possible output of the sub-plan 
$R_1(A,B) \corejoin R_2(B,C)$
is the relation $K_{1,2}$. A possible result 
of the cover-join of the latter relation with $R_3$ is 
the relation $K'$. Although $K'$ is a valid cover 
of $Q(\inst{D})$ over $\calT$, it is not a minimum-sized
cover of $Q(\inst{D})$ over $\calT$. 
\end{example}

\section{Covers for Equi-Join Queries}
\label{sec:covers_equi_join}

In this section, we extend the class of queries from natural join queries to arbitrary equi-join queries, whose relation  symbols may map to the same database relation.

\smallskip 

{\noindent\bf Equi-join Queries.} 
An equi-join query, aka full conjunctive query, has the form 
$Q = \sigma_\psi(R_1(S_1) \times \ldots \times R_n(S_n))$, 
where each $R_i$ is a relation symbol with schema $S_i$ and $\psi$ is a conjunction of equalities of the form $A_1=A_2$ with attributes $A_1$ and $A_2$. We require that all relation symbols in the query as well as all attributes occurring in the schemas of the relation symbols are distinct. We assume that each query comes with mappings 
$(\lambda, \{\mu_{R_i}\}_{i \in [n]})$, called the signature mappings
of $Q$,  where $\lambda$
maps the relation symbols in $Q$ to relation symbols in the schema of the database 
and each $\mu_{R_i}$ is a bijective mapping from the attributes of $R_i$ 
to the attributes of $\lambda(R_i)$. 
Since we do not require $\lambda$ to be injective, distinct relation symbols
in $Q$ might refer to the same relation in the database
(cf. Example \ref{ex:hypergraph_equi_join}).
The joins in equi-join queries are expressed by the equalities in $\psi$. 
The transitive closure $\psi^+$ of $\psi$ under the equality on attributes defines the attribute equivalence classes: The equivalence class ${\cal A}$ of an attribute $A$ is the set consisting of $A$ and of all attributes equal to $A$ in $\psi^+$. 
For a set $S$ of attributes, $S^+$ denotes the set of attributes transitively equivalent to those in $S$.

Hypergraphs and hypertree decompositions of equi-join queries are defined
just like for natural join queries with the additional requirement 
that each hyperedge or bag includes all
equivalent attributes for each contained attribute. More formally,
the hypergraph of an equi-jon query $Q$ consists of one node $A$ for each attribute 
$A$ in $Q$ and one edge $\sign{R}^+$ for each 
relation symbol $R\in\sign{Q}$. Similarly, 
a hypertree decomposition $\calT$ (of the hypergraph $H$) of $Q$ is a pair
$(T,\chi)$, where $T$ is a tree and $\chi$ is a function mapping each node in $T$ to a set $V^+$ where $V$ is a subset of the nodes of $H$. All other notions and notations
introduced in Section~\ref{sec:preliminaries}
as well as the definitions of result preservation and covers in Section~\ref{sec:cores_join_results} carry over to equi-join queries without any change.

\begin{example}\label{ex:hypergraph_equi_join} 
We consider the equi-join query 
$Q=\sigma_\psi(R_1(A_1,A_2)\times R_2(A_3,A_4))$, 
where $\psi$ consists of the equality $A_2=A_3$. 
Let $(\lambda, \{\mu_{R_1}, \mu_{R_2}\})$ be the signature mappings 
of the query. 
Assume that $\lambda(R_1(A_1,A_2)) = \lambda(R_2(A_3,A_4)) = R(A,B)$,
$\mu_{R_1}(A_1) = \mu_{R_2}(A_4) = A$ and
 $\mu_{R_1}(A_2) = \mu_{R_2}(A_3) = B$,
 i.e., both relation symbols   
are mapped to the same relation symbol $R(A,B)$, 
attributes $A_1$ and $A_4$ are mapped
to attribute $A$ and attributes $A_2$ and $A_3$ are mapped to attribute
$B$. Let $\inst{D} = \{R\}$ where $R$ is defined as in  
Figure~\ref{fig:hypergraph_equi_join}.
The figure  depicts in the top row (besides $R$)
the query result $Q(\inst{D})$, a cover 
$K$ of the query result over the decomposition $\calT$ 
depicted in the bottom row and two relations 
$R_1', R_2'$ obtained from $R$ by the application  
of Proposition \ref{prop:equi_to_natural} (given below). 
The bottom row shows 
the hypergraph
of $Q$, the hypergraph of $Q(\inst{D})$ over the attribute sets
 $\{\{A_1,A_2,A_3\},\{A_2,A_3,A_4\}\}$, and a minimal edge cover $M$
 of the latter hypergraph with $\relof{M} = K$. 
 \end{example}

\begin{figure}[h]

\begin{center}
\begin{tabular}{lllll}

\begin{tabular}{c@{\;}c@{\;}}
  \multicolumn{2}{c}{$R$} \\
  \toprule
  $A$ & $B$  \\\midrule
  $a_1$ & $b_1$\\
  $a_2$ & $b_1$\\
  $a_1$ & $b_2$\\
  $a_2$ & $b_2$\\
  \bottomrule
  \\
  \\
  \\
  \\
\end{tabular}

&

\begin{tabular}{c@{\;}c@{\;}c@{\;}c@{\;}}
  \multicolumn{4}{c}{$Q(\inst{D})$} \\
  \toprule
  $A_1$ & $A_2$ & $A_3$ & $A_4$   \\\midrule 
  {\color{red}$a_1$} & {\color{red}$b_1$} & {\color{red}$b_1$} & {\color{red}$a_1$} 
  \\
  {\color{blue}$a_1$} & {\color{blue}$b_1$} & {\color{blue}$b_1$} & {\color{blue}$a_2$} \\
  {\color{goodgreen}$a_2$} & {\color{goodgreen}$b_1$} & {\color{goodgreen}$b_1$} & {\color{goodgreen}$a_1$}\\
  $a_2$ & $b_1$ & $b_1$ & $a_2$ \\
  {\color{red}$a_1$} & {\color{red}$b_2$} & {\color{red}$b_2$} & {\color{red}$a_1$}
   \\
  {\color{blue}$a_1$} & {\color{blue}$b_2$} & {\color{blue}$b_2$} & {\color{blue}$a_2$} \\
  {\color{goodgreen}$a_2$} & {\color{goodgreen}$b_2$} & {\color{goodgreen}$b_2$} & {\color{goodgreen}$a_1$} \\
  $a_2$ & $b_2$ & $b_2$ & $a_2$\\  
  \bottomrule
\end{tabular}

&

\begin{tabular}{c@{\;}c@{\;}c@{\;}c@{\;}}
  \multicolumn{4}{c}{$K = \relof{M}$} \\
  \toprule
  $A_1$ & $A_2$ & $A_3$ & $A_4$  \\\midrule 
  {\color{red}$a_1$} & {\color{red}$b_1$} & {\color{red}$b_1$} & {\color{red}$a_1$} \\
   $a_2$ & $b_1$ & $b_1$ & $a_2$ \\
  {\color{blue}$a_1$} & {\color{blue}$b_2$} & {\color{blue}$b_2$} & {\color{blue}$a_2$} \\
  {\color{goodgreen}$a_2$} & {\color{goodgreen}$b_2$} & {\color{goodgreen}$b_2$} & {\color{goodgreen}$a_1$}\\  
  \bottomrule
  \\
  \\
  \\
  \\
\end{tabular}

&

\begin{tabular}{c@{\;}c@{\;}c@{\;}}
  \multicolumn{3}{c}{$R_1'$} \\
  \toprule
  $A_1$ & $A_2$ & $A_3$  \\\midrule
  $a_1$ & $b_1$ & $b_1$\\
  $a_2$ & $b_1$ & $b_1$\\
  $a_1$ & $b_2$ & $b_2$ \\
  $a_2$ & $b_2$ & $b_2$ \\
  \bottomrule
  \\
  \\
  \\
  \\
\end{tabular}

&

\begin{tabular}{c@{\;}c@{\;}c@{\;}}
  \multicolumn{3}{c}{$R_2'$} \\
  \toprule
  $A_2$ & $A_3$ & $A_4$  \\\midrule
  $b_1$ & $b_1$ & $a_1$\\
  $b_1$ & $b_1$ & $a_2$\\
  $b_2$ & $b_2$ & $a_1$ \\
  $b_2$ & $b_2$ & $a_2$ \\
  \bottomrule
  \\
  \\
  \\
  \\
\end{tabular}

\end{tabular}\\[1em]

\scalebox{0.8}{
\begin{tikzpicture}

\node at(0,-3.2) (joinGraph){Hypergraph of query};
\node at(0,-3.6) {\& decomposition  $\calT$};

\node (1)  at(-1, -4.7) {$A_1$};

\node (2) [below of =1, node distance=1.1cm] {$A_2,A_3$};

\node (3) [below of =2, node distance=1.1cm] {$A_4$};

    \begin{scope}[fill opacity=0.8]
    \draw
       ($(1)+(0,0.4)$) 
         to[out=0,in=0] ($(2) + (0.4,-0.3)$)
        to[out=0,in=180] ($(2) + (0,-0.3)$)
        to[out=180,in=0] ($(2) + (-0.4,-0.3)$)
        to[out=180,in=180] ($(1) + (0,0.4)$);

         \draw
       ($(3)+(0,-0.4)$) 
         to[out=0,in=0] ($(2) + (0.4,0.3)$)
        to[out=0,in=180] ($(2) + (0,0.3)$)
        to[out=180,in=0] ($(2) + (-0.4,0.3)$)
        to[out=180,in=180] ($(3) + (0,-0.4)$);       
        
     \end{scope}

\node (A1)  at(1, -4.5) {$A_1$};
\node (A2)  [below of =A1, node distance=0.5cm] {$A_2,A_3$};

\node (A23) [below of =A1, node distance=1.7cm] {$A_2,A_3$};
\node (A231) [below of =A23, node distance=0.5cm] {$A_4$};

\draw [decorate,decoration={brace,amplitude=6pt,raise=0pt}] 
(0.5,-4.4) -- (1.5,-4.4);
\draw [decorate,decoration={brace,amplitude=6pt,raise=0pt,mirror}] 
(0.5,-5.2) -- (1.5,-5.2);

\draw [decorate,decoration={brace,amplitude=6pt,raise=0pt}] 
(0.5,-6) -- (1.5,-6);
\draw [decorate,decoration={brace,amplitude=6pt,raise=0pt,mirror}] 
(0.5,-6.9) -- (1.5,-6.9);

\draw [decorate, segment length=20] (1,-5.4) -- (1,-5.8);

\node at(5,-3.2) (joinGraph){Hypergraph $H$ of $Q(\inst{D})$ over $\sign{\calT}$};   
\node at(2.7,-4.05) (11) {$a_1$};
\node (12') [right of =11,node distance=0.5cm] {$b_1$};
\node (12) [right of =12',node distance=0.5cm] {$b_1$};

\node (21) [below of =11,node distance=1.4cm] {$a_2$};
\node (22') [right of =21,node distance=0.5cm] {$b_1$};
\node (22) [right of =22',node distance=0.5cm] {$b_1$};

\node (31) [below of =21,node distance=1.4cm] {$a_1$};
\node (32') [right of =31,node distance=0.5cm] {$b_2$};
\node (32) [right of =32',node distance=0.5cm] {$b_2$};

\node (31new) [below of =31,node distance=1.4cm] {$a_2$};
\node (32new') [right of =31new,node distance=0.5cm] {$b_2$};
\node (32new) [right of =32new',node distance=0.5cm] {$b_2$};

\tikzstyle{background}=[rectangle,
                                                rounded corners=1mm]

\begin{pgfonlayer}{background}                   
   \node [background,
                    fit=(11)(12)(12'),
                    fill=gray!15,inner sep= -1] {};        
                    
   \node [background,
                    fit=(21)(22)(22'),
                    fill=gray!15,inner sep= -1] {};
                    
   \node [background,
                    fit=(31)(32)(32'),
                    fill=gray!15,inner sep= -1] {};                                                          

   \node [background,
                    fit=(31new)(32new)(32new'),
                    fill=gray!15,inner sep= -1] {};                                                          
\end{pgfonlayer}


\node (71) [right of =11,node distance=4cm]  {$b_1$};
\node (71') [right of =71,node distance=0.5cm] {$b_1$};
\node (72) [right of =71',node distance=0.5cm] {$a_1$};

\node (71new) [below of =71,node distance=1.4cm]  {$b_1$};
\node (71new') [right of =71new,node distance=0.5cm]  {$b_1$};
\node (72new) [right of =71new',node distance=0.5cm] {$a_2$};

\node (81) [below of =71new,node distance=1.4cm] {$b_2$};
\node (81') [right of =81,node distance=0.5cm] {$b_2$};
\node (82) [right of =81',node distance=0.5cm] {$a_1$};

\node (91) [below of =81,node distance=1.4cm] {$b_2$};
\node (91') [right of =91,node distance=0.5cm] {$b_2$};
\node (92) [right of =91',node distance=0.5cm] {$a_2$};

\tikzstyle{background}=[rectangle,
                                                rounded corners=1mm]

\begin{pgfonlayer}{background}                   
   \node [background,
                    fit=(71)(71')(72),
                    fill=gray!15,inner sep= -1] {};        
                    
   \node [background,
                    fit=(71new)(71new')(72new),
                    fill=gray!15,inner sep= -1] {};                            
                    
   \node [background,
                    fit=(81)(81')(82),
                    fill=gray!15,inner sep= -1] {};
                    
 \node [background,
                    fit=(91)(91')(92),
                    fill=gray!15,inner sep= -1] {};  
                     
\end{pgfonlayer}

   
    \begin{scope}[fill opacity=0.8]
    \draw[color=red]
       ($(11)+(0,0.4)$) 
        to[out=0,in=180] ($(12) + (0,0.4)$)
        to[out=0,in=180] ($(71)+(0,0.4)$)      
        to[out=0,in=180] ($(72)+(0,0.4)$)            
        to[out=0,in=90] ($(72)+(0.4,0)$)                     
        to[out=270,in=0] ($(72)+(0,-0.4)$)                        
        to[out=180,in=0] ($(71)+(0,-0.4)$)                                       
        to[out=180,in=0] ($(12)+(0,-0.4)$) 
        to[out=180,in=0] ($(11)+(0,-0.4)$)
        to[out=180,in=270] ($(11)+(-0.4,0)$)    
        to[out=90,in=180] ($(11)+(0,0.4)$);

    \draw[color=blue]
       ($(11)+(0,0.6)$) 
        to[out=0,in=180] ($(12) + (0,0.6)$)
        to[out=0,in=180] ($(71new)+(0,0.4)$)      
        to[out=0,in=180] ($(72new)+(0,0.4)$)            
        to[out=0,in=90] ($(72new)+(0.4,0)$)                     
        to[out=270,in=0] ($(72new)+(0,-0.4)$)                        
        to[out=180,in=0] ($(71new)+(0,-0.4)$)                        
        to[out=180,in=0] ($(12)+(0,-0.6)$) 
        to[out=180,in=0] ($(11)+(0,-0.6)$)
        to[out=180,in=270] ($(11)+(-0.6,0)$)    
        to[out=90,in=180] ($(11)+(0,0.60)$);

      \draw[color=goodgreen]
       ($(21)+(0,0.4)$) 
        to[out=0,in=180] ($(22) + (0,0.4)$)
        to[out=0,in=180] ($(71)+(0,0.6)$)      
        to[out=0,in=180] ($(72)+(0,0.6)$)            
        to[out=0,in=90] ($(72)+(0.6,0)$)                     
        to[out=270,in=0] ($(72)+(0,-0.6)$)                        
        to[out=180,in=0] ($(71)+(0,-0.6)$)                        
        to[out=180,in=0] ($(22)+(0,-0.4)$) 
        to[out=180,in=0] ($(21)+(0,-0.4)$)
        to[out=180,in=270] ($(21)+(-0.4,0)$)   
        to[out=90,in=180] ($(21)+(0,0.4)$);

      \draw
       ($(21)+(0,0.6)$) 
        to[out=0,in=180] ($(22) + (0,0.6)$)
        to[out=0,in=180] ($(71new)+(0,0.6)$)      
        to[out=0,in=180] ($(72new)+(0,0.6)$)            
        to[out=0,in=90] ($(72new)+(0.6,0)$)                     
        to[out=270,in=0] ($(72new)+(0,-0.6)$)                        
        to[out=180,in=0] ($(71new)+(0,-0.6)$)                        
        to[out=180,in=0] ($(22)+(0,-0.6)$) 
        to[out=180,in=0] ($(21)+(0,-0.6)$)
        to[out=180,in=270] ($(21)+(-0.6,0)$)   
        to[out=90,in=180] ($(21)+(0,0.6)$);

     
      \draw[color=red]
       ($(31)+(0,0.4)$) 
        to[out=0,in=180] ($(32) + (0,0.4)$)
        to[out=0,in=180] ($(81)+(0,0.4)$)      
        to[out=0,in=180] ($(82)+(0,0.4)$)            
        to[out=0,in=90] ($(82)+(0.4,0)$)                       
        to[out=270,in=0] ($(82)+(0,-0.4)$)                        
        to[out=180,in=0] ($(81)+(0,-0.4)$)                        
        to[out=180,in=0] ($(32)+(0,-0.4)$) 
        to[out=180,in=0] ($(31)+(0,-0.4)$)
        to[out=180,in=270] ($(31)+(-0.4,0)$)   
        to[out=90,in=180] ($(31)+(0,0.4)$);

      \draw[color=blue]
       ($(31)+(0,0.6)$) 
        to[out=0,in=180] ($(32) + (0,0.6)$)
        to[out=0,in=180] ($(91)+(0,0.4)$)      
        to[out=0,in=180] ($(92)+(0,0.4)$)            
        to[out=0,in=90] ($(92)+(0.4,0)$)                     
        to[out=270,in=0] ($(92)+(0,-0.4)$)                        
        to[out=180,in=0] ($(91)+(0,-0.4)$)                        
        to[out=180,in=0] ($(32)+(0,-0.6)$) 
        to[out=180,in=0] ($(31)+(0,-0.6)$)
        to[out=180,in=270] ($(31)+(-0.6,0)$)    
        to[out=90,in=180] ($(31)+(0,0.6)$);

     \draw[color=goodgreen]
       ($(31new)+(0,0.4)$) 
        to[out=0,in=180] ($(32new) + (0,0.4)$)
        to[out=0,in=180] ($(81)+(0,0.6)$)      
        to[out=0,in=180] ($(82)+(0,0.6)$)            
        to[out=0,in=90] ($(82)+(0.6,0)$)                     
        to[out=270,in=0] ($(82)+(0,-0.6)$)                        
        to[out=180,in=0] ($(81)+(0,-0.6)$)                        
        to[out=180,in=0] ($(32new)+(0,-0.4)$) 
        to[out=180,in=0] ($(31new)+(0,-0.4)$)
        to[out=180,in=270] ($(31new)+(-0.4,0)$)    
        to[out=90,in=180] ($(31new)+(0,0.4)$);

             \draw
       ($(31new)+(0,0.6)$) 
        to[out=0,in=180] ($(32new) + (0,0.6)$)
        to[out=0,in=180] ($(91)+(0,0.6)$)      
        to[out=0,in=180] ($(92)+(0,0.6)$)            
        to[out=0,in=90] ($(92)+(0.6,0)$)                     
        to[out=270,in=0] ($(92)+(0,-0.6)$)                        
        to[out=180,in=0] ($(91)+(0,-0.6)$)                        
        to[out=180,in=0] ($(32new)+(0,-0.6)$) 
        to[out=180,in=0] ($(31new)+(0,-0.6)$)
        to[out=180,in=270] ($(31new)+(-0.6,0)$)    
        to[out=90,in=180] ($(31new)+(0,0.6)$);

     \end{scope}

\node at(11.7,-3.2)(edgeSubset){Minimal edge cover $M$ of $H$};
\node at(9.4,-4.05) (11) {$a_1$};
\node (12') [right of =11,node distance=0.5cm] {$b_1$};
\node (12) [right of =12',node distance=0.5cm] {$b_1$};

\node (21) [below of =11,node distance=1.4cm] {$a_2$};
\node (22') [right of =21,node distance=0.5cm] {$b_1$};
\node (22) [right of =22',node distance=0.5cm] {$b_1$};

\node (31) [below of =21,node distance=1.4cm] {$a_1$};
\node (32') [right of =31,node distance=0.5cm] {$b_2$};
\node (32) [right of =32',node distance=0.5cm] {$b_2$};

\node (31new) [below of =31,node distance=1.4cm] {$a_2$};
\node (32new') [right of =31new,node distance=0.5cm] {$b_2$};
\node (32new) [right of =32new',node distance=0.5cm] {$b_2$};

\tikzstyle{background}=[rectangle,
                                                rounded corners=1mm]

\begin{pgfonlayer}{background}                   
   \node [background,
                    fit=(11)(12)(12'),
                    fill=gray!15,inner sep= -1] {};        
                    
   \node [background,
                    fit=(21)(22)(22'),
                    fill=gray!15,inner sep= -1] {};
                    
   \node [background,
                    fit=(31)(32)(32'),
                    fill=gray!15,inner sep= -1] {};                                                          

   \node [background,
                    fit=(31new)(32new)(32new'),
                    fill=gray!15,inner sep= -1] {};                                                          
\end{pgfonlayer}


\node (71) [right of =11,node distance=4cm]  {$b_1$};
\node (71') [right of =71,node distance=0.5cm] {$b_1$};
\node (72) [right of =71',node distance=0.5cm] {$a_1$};

\node (71new) [below of =71,node distance=1.4cm]  {$b_1$};
\node (71new') [right of =71new,node distance=0.5cm]  {$b_1$};
\node (72new) [right of =71new',node distance=0.5cm] {$a_2$};

\node (81) [below of =71new,node distance=1.4cm] {$b_2$};
\node (81') [right of =81,node distance=0.5cm] {$b_2$};
\node (82) [right of =81',node distance=0.5cm] {$a_1$};

\node (91) [below of =81,node distance=1.4cm] {$b_2$};
\node (91') [right of =91,node distance=0.5cm] {$b_2$};
\node (92) [right of =91',node distance=0.5cm] {$a_2$};

\tikzstyle{background}=[rectangle,
                                                rounded corners=1mm]

\begin{pgfonlayer}{background}                   
   \node [background,
                    fit=(71)(71')(72),
                    fill=gray!15,inner sep= -1] {};        
                    
   \node [background,
                    fit=(71new)(71new')(72new),
                    fill=gray!15,inner sep= -1] {};                            
                    
   \node [background,
                    fit=(81)(81')(82),
                    fill=gray!15,inner sep= -1] {};
                    
 \node [background,
                    fit=(91)(91')(92),
                    fill=gray!15,inner sep= -1] {};  
                     
\end{pgfonlayer}

   
    \begin{scope}[fill opacity=0.8]
    \draw[color=red]
       ($(11)+(0,0.4)$) 
        to[out=0,in=180] ($(12) + (0,0.4)$)
        to[out=0,in=180] ($(71)+(0,0.4)$)      
        to[out=0,in=180] ($(72)+(0,0.4)$)            
        to[out=0,in=90] ($(72)+(0.4,0)$)                     
        to[out=270,in=0] ($(72)+(0,-0.4)$)                        
        to[out=180,in=0] ($(71)+(0,-0.4)$)                                       
        to[out=180,in=0] ($(12)+(0,-0.4)$) 
        to[out=180,in=0] ($(11)+(0,-0.4)$)
        to[out=180,in=270] ($(11)+(-0.4,0)$)    
        to[out=90,in=180] ($(11)+(0,0.4)$);

      \draw
       ($(21)+(0,0.4)$) 
        to[out=0,in=180] ($(22) + (0,0.4)$)
        to[out=0,in=180] ($(71new)+(0,0.4)$)      
        to[out=0,in=180] ($(72new)+(0,0.4)$)            
        to[out=0,in=90] ($(72new)+(0.4,0)$)                     
        to[out=270,in=0] ($(72new)+(0,-0.4)$)                        
        to[out=180,in=0] ($(71new)+(0,-0.4)$)                        
        to[out=180,in=0] ($(22)+(0,-0.4)$) 
        to[out=180,in=0] ($(21)+(0,-0.4)$)
        to[out=180,in=270] ($(21)+(-0.4,0)$)   
        to[out=90,in=180] ($(21)+(0,0.4)$);


      \draw[color=blue]
       ($(31)+(0,0.4)$) 
        to[out=0,in=180] ($(32) + (0,0.4)$)
        to[out=0,in=180] ($(91)+(0,0.4)$)      
        to[out=0,in=180] ($(92)+(0,0.4)$)            
        to[out=0,in=90] ($(92)+(0.4,0)$)                     
        to[out=270,in=0] ($(92)+(0,-0.4)$)                        
        to[out=180,in=0] ($(91)+(0,-0.4)$)                        
        to[out=180,in=0] ($(32)+(0,-0.4)$) 
        to[out=180,in=0] ($(31)+(0,-0.4)$)
        to[out=180,in=270] ($(31)+(-0.4,0)$)    
        to[out=90,in=180] ($(31)+(0,0.4)$);

     \draw[color=goodgreen]
       ($(31new)+(0,0.4)$) 
        to[out=0,in=180] ($(32new) + (0,0.4)$)
        to[out=0,in=180] ($(81)+(0,0.4)$)      
        to[out=0,in=180] ($(82)+(0,0.4)$)            
        to[out=0,in=90] ($(82)+(0.4,0)$)                     
        to[out=270,in=0] ($(82)+(0,-0.4)$)                        
        to[out=180,in=0] ($(81)+(0,-0.4)$)                        
        to[out=180,in=0] ($(32new)+(0,-0.4)$) 
        to[out=180,in=0] ($(31new)+(0,-0.4)$)
        to[out=180,in=270] ($(31new)+(-0.4,0)$)    
        to[out=90,in=180] ($(31new)+(0,0.4)$);

     \end{scope}

\end{tikzpicture}
}
\end{center}
\caption{
Top row: database $\inst{D}=\{R\}$, the result $Q(\inst{D})$ of the query 
$Q$ in Example~\ref{ex:hypergraph_equi_join}, a cover $K$ of 
$Q(\inst{D})$ over $\calT$, and relations $R_1',R_2'$ obtained 
from $R$ by the application of Proposition \ref{prop:equi_to_natural}; 
bottom row: the hypergraph of $Q$,  a decomposition $\calT$ of $Q$, the hypergraph of $Q(\inst{D})$ over the attribute sets $\sign{\calT}$, and a minimal edge cover $M$ of this hypergraph.}
\label{fig:hypergraph_equi_join}
\end{figure}
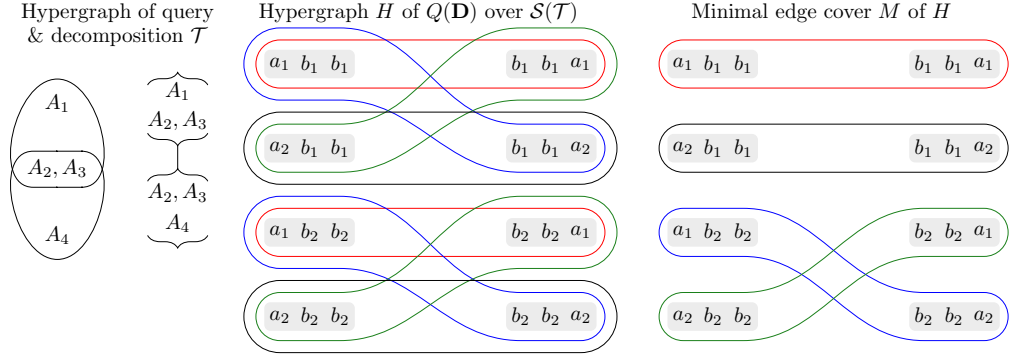

\smallskip
{\noindent\bf Adaption of the results on covers to equi-join queries.} 
Due to the following two propositions, all results on covers in Sections 
\ref{sec:cores_join_results} and  \ref{sec:core-plans}
carry over to equi-join queries.

\begin{proposition}\label{prop:equi_to_natural}
Given an equi-join query $Q$, a decomposition $\calT$ of $Q$, and 
a database $\inst{D}$, there exist a natural join query 
$Q'$ and a database $\inst{D}'$ such that: $Q'(\inst{D}')=Q(\inst{D})$, 
$Q'$ has the decomposition $\calT$ and
can be constructed in time  $\calO(|Q|)$, and $\inst{D}'$ can be constructed in time 
$\calO(|\inst{D}|)$.
\end{proposition}

We briefly explain the construction. The query 
$Q'$ is obtained  from $Q$ by replacing each relation symbol 
$R(S)$ in $Q$ by a relation symbol $R'(S^+)$. The database $\inst{D}'$
contains, for each relation symbol $R'(S^+)$ in $Q'$, 
a relation over the same schema that is obtained 
from relation $\lambda(R(S))$ 
as follows: for each attribute $A$ contained in $S^+$ but not in
$S$, $\lambda(R(S))$ is extended by a new $A$-column that is 
a copy of any $B$-column in $\lambda(R(S))$ such that $A$ is equivalent to $B$.
Figure~\ref{fig:hypergraph_equi_join} gives in the top row 
two relations $R_1'$ and $R_2'$ 
that result from relation $R$ by the application of 
Proposition \ref{prop:equi_to_natural} in case $Q$ is defined
as in Example \ref{ex:hypergraph_equi_join}.  

It follows from Proposition \ref{prop:equi_to_natural}
 that, since  $Q'(\inst{D}')=Q(\inst{D})$, any relation 
$K$ is a cover of $Q(\inst{D})$ over $\calT$
if and only if $K$ is a cover of $Q'(\inst{D}')$ over $\calT$.
Given the construction times for $Q'$ and $D'$, all our results on natural join queries in Sections~\ref{sec:cores_join_results} and~\ref{sec:core-plans}, except 
the lower size bound on covers in Theorem \ref{theo:general_size_bounds}(ii),  hold for equi-join queries, too.

The following proposition is the counterpart of 
Theorem \ref{theo:general_size_bounds}(ii) for equi-join queries. 

\begin{proposition}\label{prop:lower_bound_equi_join}
For any equi-join query $Q$ and any decomposition $\calT$ of $Q$,
there are arbitrarily large databases $\inst{D}$ such that 
each cover of $Q(\inst{D})$ over $\calT$
has size $\Omega(|\inst{D}|^{\fhtw{\calT}})$.
\end{proposition}

In Proposition 
\ref{prop:lower_bound_equi_join}, we first construct 
a natural join query  $Q'$ from $Q$ as in Proposition 
\ref{prop:equi_to_natural}. By Theorem 
\ref{theo:general_size_bounds}(ii),  
there are arbitrarily large databases $\inst{D}'$ such that 
each cover of $Q'(\inst{D}')$ over $\calT$
has size $\Omega(|\inst{D}'|^{\fhtw{\calT}})$.
Given such a database $\inst{D}'$, it follows from Proposition \ref{prop:size_bounds_hypergraph}, that 
$\Sigma_{B \in \sign{\calT}}|\pi_BQ'(\inst{D}')| = 
\Omega(|\inst{D}'|^{\fhtw{\calT}})$, hence, 
$\max_{B \in \sign{\calT}}\{|\pi_BQ'(\inst{D}')|\}
=  \Omega(|\inst{D}'|^{\fhtw{\calT}})$. 
The database $\inst{D}'$ can be converted into 
a database $\inst{D}$ of size $\calO(|\inst{D}'|)$
such that $|\pi_BQ(\inst{D})| \geq |\pi_BQ'(\inst{D}')|$
for each $B \in \sign{\calT}$. 
By Proposition \ref{prop:size_bounds_hypergraph} 
(adapted to equi-join queries), each cover of $Q(\inst{D})$
over $\calT$
must have size at least $\max_{B \in \sign{\calT}}\{|\pi_BQ(\inst{D})|\}$.
Since $\max_{B \in \sign{\calT}}\{|\pi_BQ'(\inst{D}')|\}$
$=$  $\Omega(|\inst{D}'|^{\fhtw{\calT}})$ and 
$\max_{B \in \sign{\calT}}\{|\pi_BQ(\inst{D})|\}$ $\geq$ 
$\max_{B \in \sign{\calT}}\{|\pi_BQ'(\inst{D}')|\}$, we conclude that
each cover of $Q(\inst{D})$
over $\calT$ is of size $\Omega(|\inst{D}'|^{\fhtw{\calT}})$ $=$ $\Omega(|\inst{D}|^{\fhtw{\calT}})$.

\section{Missing Proofs of Section \ref{sec:preliminaries}}
\subsection{Proof of Proposition \ref{prop:rewriting2}}
{\noindent\bf Proposition \ref{prop:rewriting2}.}
\textit{
Given $(Q,\calT,\inst{D})$, we can compute $(Q',\calT,\inst{D}')$ 
with size ${\mathcal O}(|\inst{D}|^{\fhtw{\calT}})$  and in time $\widetilde{\mathcal O}(|\inst{D}|^{\fhtw{\calT}})$ 
such that $Q'$ is an acyclic natural join query, $\calT$ corresponds 
to  a join tree of $Q'$, 
 $\inst{D}'$ is globally consistent with respect to $Q'$ and $Q'(\inst{D}')=Q(\inst{D})$.
 }

\smallskip
\smallskip
\smallskip

The construction is standard in the literature
\cite{AHV95,SOC:SIGMOD:2016}. For convenience,
 we describe the main ideas.
 
\smallskip 
\underline{\em Construction.}
The construction comprises two transformation steps.
We first 
 compute $R_B = Q_B(\inst{D}_B)$ 
 for each $B \in \sign{\calT}$
 (recall that 
 $Q_B$ and $\inst{D}_B$ are $B$-restrictions of 
 $Q$ and $\inst{D}$, respectively).
 Let  
$\widehat{\inst{D}}=\setindexel{R_B}{B}{\sign{\calT}}$
and $\widehat{Q}=\setindexjoin{R_B}{B}{\sign{\calT}}$.
In the second transformation step, we execute a semi-join programme 
on
$\widehat{\inst{D}}$ to turn it 
into a database $\inst{D}' = \setindexel{R_B'}{B}{\sign{\calT}}$ that is pairwise consistent with respect to $\widehat{Q}$,
i.e., $\inst{D}'$ does not contain any pair of relations such that 
one of the two relations contains a tuple which cannot be joined 
with any tuple from the other relation. To achieve pairwise consistency, 
it is not necessary to consider {\em all} pairs of relations in 
$\widehat{\inst{D}}$. It suffices to execute a bottom-up and a subsequent 
top-down traversal in $\calT$ \cite{Yannakakis81}. During each traversal,
we delete for each father-child pair 
$B_1, B_2$ of bags, all tuples in 
each of the two relations 
$R_{B_1}$ and $R_{B_2}$ which
do not have any join partner in the other relation.  
We define $Q'=\setindexjoin{R_B'}{B}{\sign{\calT}}$.

\smallskip 
\underline{\em $Q'$ is an acyclic natural join query and $\calT$ corresponds 
to a join tree  of $Q'$.}
By construction, we have a one-to-one correspondence 
between  
relation symbols $R_B' \in \sign{Q'}$
and bags $B \in \sign{\calT}$
with 
$\sign{R_B'}=B$.
Hence, $\calT$ corresponds to the join tree of $Q'$ that 
is obtained from $\calT$ by, basically, replacing each bag by the corresponding relation symbol in $Q$. Since $Q'$ has a join tree, it is acyclic.  

\smallskip 
\underline{\em $\inst{D}'$ is globally consistent with respect to $Q'$.} 
The relations in $\inst{D}'$ are pairwise consistent
with respect to $Q'$.
For acyclic queries,
pairwise consistency 
implies global consistency (Theorem  6.4.5 of \cite{AHV95}).
Hence, $\inst{D}'$ is globally consistent with respect to $Q'$.

\smallskip 
\underline{\em $Q(\inst{D}) = Q'(\inst{D}')$.} 
Since the second transformation step only deletes
tuples in $\widehat{\inst{D}}$ which do not contribute to the result 
of $\widehat{Q}(\widehat{\inst{D}})$, it suffices 
to show that
$Q(\inst{D}) = \widehat{Q}(\widehat{\inst{D}})$.
Let $Q = \indexjoin{R_i}{i}{n}$.

We first show $Q(\inst{D}) \subseteq  \widehat{Q}(\widehat{\inst{D}})$. 
Let $t \in Q(\inst{D})$.
Since $\pi_{B} Q(\inst{D}) \subseteq Q_B(\inst{D}_B)$,
it follows that $\pi_{B} t \in Q_B(\inst{D}_B)$
 for each $B \in \sign{\calT}$. 
 Hence, $\pi_{B} t \in R_B$
 for each $B \in \sign{\calT}$. 
 Since $t = \setindexjoin{\pi_B t}{B}{\sign{\calT}}$, 
 we derive that
 $t \in \setindexjoin{R_B}{B}{\sign{\calT}}$, thus,  
$t \in \widehat{Q}(\widehat{\inst{D}})$.

We now show $\widehat{Q}(\widehat{\inst{D}}) \subseteq  Q(\inst{D})$.
Let $t \in \widehat{Q}(\widehat{\inst{D}})$. By definition, 
$\pi_{B}t \in R_B$ for each $B \in \sign{\calT}$.
By the fact that the attributes of each relation symbol in $Q$
are covered by at least one bag of $\calT$ and by the 
construction of the relations $R_B$, it holds that
$\pi_{\sign{R_i}}t \in R_i$ for each $i \in \indexnat{n}$.
This implies $t \in  Q(\inst{D})$.  

\smallskip 
\underline{\em Construction size.} 
Each relation 
$R_B$ in $\widehat{\inst{D}}$  has size
$\calO(\sizeof{\inst{D}_B}^{\rho^*(Q_B)})$ \cite{AtseriasGM13}.
 Since $\fhtw{\calT}=
\takemax{B \in \sign{\calT}}{\rho^*(Q_B)}$,
it follows that the size of 
$\widehat{\inst{D}}$ is 
$\calO(\sizeof{\inst{D}}^{\fhtw{\calT}})$.
The semi-join program on $\widehat{\inst{D}}$
does not increase the size of the database.
The size of  $Q'$ is $\calO(|Q|)$.
Altogether, the size of 
$(Q',\calT,\inst{D}')$ is 
$\calO(\sizeof{\inst{D}}^{\fhtw{\calT}})$.
 
 \smallskip 
\underline{\em Construction time.} 
Each relation 
$R_B$ in $\widehat{\inst{D}}$  is computable 
in time 
$\softO(\sizeof{\inst{D}_B}^{\rho^*(Q_B)})$ \cite{Ngo:SIGREC:2013}. 
By $\fhtw{\calT}=
\takemax{B \in \sign{\calT}}{\rho^*(Q_B)}$,
we derive that the computation time 
for $\widehat{\inst{D}}$ is  
$\softO(\sizeof{\inst{D}}^{\fhtw{\calT}})$.
During the semi-join program on $\widehat{\inst{D}}$,
we can achieve consistency between each pair 
$R_{B_1}, R_{B_2}$ of father-child relations as follows.
We first sort 
both relations on the join attributes. In a subsequent scan we 
delete in each of the relations each tuple with no join partner in the 
other relation. Hence, the semi-join programme can be realised 
in time $\softO(\sizeof{\inst{D}}^{\fhtw{\calT}})$. 
It follows that the overall running time is 
$\softO(\sizeof{\inst{D}}^{\fhtw{\calT}})$.

\section{Missing Proofs of Section \ref{sec:cores_join_results}}
\subsection{Proof of Proposition \ref{Q(D)-preserv=join_preserv}}
{\noindent\bf Proposition \ref{Q(D)-preserv=join_preserv}.}
\textit{
Given $(Q,\calT,\inst{D})$, a relation $K$ with schema $\mathit{att}(Q)$ is result-preserving with respect to $(Q,\calT,\inst{D})$ if and only if $\Join_{B\in\sign{\calT}} \pi_B K = Q(\inst{D})$.
}

\smallskip
\smallskip
\smallskip

\underline{\em Proof of the ``$\Rightarrow$''-direction.} 
Assume that $K$ is 
result-preserving with respect to $(Q,\calT,\inst{D})$.
We show in two steps that $\setindexjoin{\pi_{B}K}{B}{\sign{\calT}} = Q(\inst{D})$.
 
\begin{itemize}
\item $\setindexjoin{\pi_{B}K}{B}{\sign{\calT}} \subseteq Q(\inst{D})$:
Let $t$ be an arbitrary tuple from 
$\setindexjoin{\pi_{B}K}{B}{\sign{\calT}}$. This means that 
$\pi_{B}t \in \pi_{B}K$ for every $B \in \sign{\calT}$.
Since $K$ is result-preserving with respect to $(Q,\calT,\inst{D})$, 
we derive that 
$\pi_{B}t \in \pi_{B}Q(\inst{D})$ for every $B \in \sign{\calT}$. 
By the definition of decompositions, for every relation symbol $R$ 
in $Q$, there is at least one bag of $\calT$
  containing all attributes of $R$. Hence,  
  $\pi_{\sign{R}}t \in \pi_{\sign{R}}Q(\inst{D})$ 
  for every $R \in \sign{Q}$. It follows that $t$
  is included in $Q(\inst{D})$. Thus, 
  $\setindexjoin{\pi_{B}K}{B}{\sign{\calT}} \subseteq Q(\inst{D})$.
  
  \item $Q(\inst{D}) \subseteq \setindexjoin{\pi_{B}K}{B}{\sign{\calT}}$:
    Let $t \in Q(\inst{D})$.
    It follows that $\pi_B t \in \pi_B Q(\inst{D})$ for every $B \subseteq \sign{Q(\inst{D})}$, hence, in particular for every $B\in \sign{\calT}$.
Due to result-preservation of $K$ with respect to $(Q,\calT, \inst{D})$,
this implies that $\pi_B t \in \pi_B K$ for every $B\in \sign{\calT}$
which means that $t \in \setindexjoin{\pi_{B}K}{B}{\sign{\calT}}$.
Hence, $Q(\inst{D}) \subseteq \setindexjoin{\pi_{B}K}{B}{\sign{\calT}}$. 
\end{itemize}

\underline{\em Proof of the ``$\Leftarrow$''-direction.}
Assume that $\setindexjoin{\pi_B K}{B}{\sign{\calT}} = Q(\inst{D})$.
 Given any $B \in \sign{\calT}$, we show in two steps that 
 $\pi_B K = \pi_B Q(\inst{D})$. 
  
  \begin{itemize}
  \item $\pi_B K \subseteq \pi_B Q(\inst{D})$: 
  Let $t$ be an arbitrary tuple 
  from $\pi_B K$. 
  This means that there is a tuple $t' \in K$ with 
$\pi_B t'=t$. 
Since $\pi_{B'} t' \in \pi_{B'}K$ for each $B' \in \sign{\calT}$, 
we derive that
$t'\in \setindexjoin{\pi_{B'}K}{B'}{\sign{\calT}}$. 
Using our assumption
$\setindexjoin{\pi_{B'} K}{B'}{\sign{\calT}} = Q(\inst{D})$, we get   
$t'\in Q(\inst{D})$. From the latter and the fact that 
$t = \pi_B t'$, it follows $t \in \pi_B Q(\inst{D})$.
Altogether, we conclude  $\pi_B K \subseteq \pi_B Q(\inst{D})$. 

  \item $\pi_B Q(\inst{D}) \subseteq \pi_B K$:
Let $t$ be an arbitrary tuple from $\pi_B Q(\inst{D})$.
This means that there is a tuple $t' \in Q(\inst{D})$
with $\pi_B t' = t$. 
By assumption,  
$t' \in \setindexjoin{\pi_{B'} K}{B'}{\sign{\calT}}$.
Since $B$ is an element of $\sign{\calT}$,
 the latter implies
$\pi_B t'= t \in \pi_B K$.
Altogether, we get $\pi_B Q(\inst{D}) \subseteq \pi_B K$.
\end{itemize}

\subsection{Proof of Proposition \ref{characterize_cores}}
{\noindent\bf Proposition \ref{characterize_cores}.}
\textit{
Given $(Q,\calT,\inst{D})$, a relation $K$ is a cover of the query result $Q(\inst{D})$ over $\calT$ if and only if the hypergraph of $Q(\inst{D})$ over 
$\sign{\calT}$ has a minimal edge cover $M$ such that $\relof{M}=K$.
}

\smallskip
\smallskip
\smallskip

We first recall that $Q(\inst{D})$ is obviously result-preserving with respect to 
$(Q,\calT,\inst{D})$ and therefore, by 
Proposition \ref{Q(D)-preserv=join_preserv}, it holds
$\Join_{B \in \sign{\calT}}\pi_B Q(\inst{D}) = Q(\inst{D})$.
Let $H=(V,E)$
be the hypergraph of $Q(\inst{D})$ over $\sign{\calT}$.
Let $\tupleofs_V$ be a function mapping each 
node $v \in V$ to its corresponding tuple in 
$\bigcup_{B \in \sign{\calT}}\pi_B Q(\inst{D})$.
Furthermore, let $\tupleofs_E$ 
be a function mapping each edge $e \in E$ to 
$\setindexjoin{\tupleofs_V(v)}{v}{e}$. 
The function $\tupleofs_V$ 
is a bijection from $V$ to 
$\bigcup_{B \in \sign{\calT}}\pi_B Q(\inst{D})$.
Since $\Join_{B \in \sign{\calT}}\pi_B Q(\inst{D}) = Q(\inst{D})$, 
$\tupleofs_E$ is a bijection from $E$ to 
$\setindexjoin{\pi_B Q(\inst{D})}{B}{\sign{\calT}}$. 
Likewise, the function
$\relofs$ (as defined in Section \ref{sec:preliminaries}) is a bijection from 
subsets of $E$ to subsets of 
$\setindexjoin{\pi_B Q(\inst{D})}{B}{\sign{\calT}}$.

\underline{\em Proof of the ``$\Rightarrow$''-direction.}
Let $K$ be a cover of $Q(\inst{D})$ over $\calT$.
We show that $\relofs^-(K)$ is defined and a minimal edge 
cover of $H$. 

Let $t \in K$. This means that for each $B \in \sign{\calT}$,
there is $t_B \in \pi_B K$ with $t = \Join_{B \in \sign{\calT}} t_B$.  
Since $K$ is result-preserving with respect to $(Q,\calT,\inst{D})$,
each $t_B$ is included in $\pi_B Q(\inst{D})$.
Hence, $\tupleofs_V^-(t_B)$ must be defined for each $B \in \sign{\calT}$. 
Since $Q(\inst{D}) = \Join_{B \in \sign{\calT}}Q(\inst{D})$, 
$t$ is included in $Q(\inst{D})$. 
It follows that $\tupleofs_E^-(t) = \{\tupleofs_V^-(t_B)\}_{B \in \sign{\calT}}$
is defined. Thus, $\relofs^-(K) = \{\tupleofs_E^-(t)\}_{t \in K}$
is defined.

It follows from $\pi_BQ(\inst{D}) = \pi_BK$, $B \in \sign{\calT}$,
that $\relofs^-(K)$ is an edge cover of $H$.
It remains to show that $\relofs^-(K)$ is a {\em minimal}
edge cover of $H$. For the sake of contradiction, assume 
that $\relofs^-(K)$ is not a minimal edge cover of $H$.
This implies that there is an edge $\overline{e} \in \relofs^-(K)$ 
such that $\relofs^-(K) \backslash \{\overline{e}\}$ is an edge 
cover of $H$. It follows that for each node $v \in V$, there 
is an edge $e \in \relofs^-(K) \backslash \{\overline{e}\}$ 
with $v \in e$. This means that for each tuple $t_B \in \pi_B Q(\inst{D})
= \pi_B K$, there is a tuple $t \in K \backslash \{\tupleofs_E(\overline{e})\}$   
with $\pi_B t = t_B$.
We conclude that $K \backslash \{\tupleofs_E(\overline{e})\}$
is result-preserving with respect to $(Q,\calT,\inst{D})$.
The latter is, however, a contradiction to our assumption that $K$
is a cover of $Q(\inst{D})$ over $\calT$ and, therefore, a 
minimal result-preserving relation with respect to $(Q,\calT,\inst{D})$.

\underline{\em Proof of the ``$\Leftarrow$''-direction.}
Let $M$ be a minimal edge cover of $H$.
We show that $\relofs(M)$ is a cover of $Q(\inst{D})$
over $\calT$. 

We first observe that $\sign{\relofs(M)} = \sign{Q(\inst{D})} = \att{Q}$.
Since $Q(\inst{D})$ is result-preserving with respect to 
$(Q,\calT, \inst{D})$ and $M$ is an edge cover of $H$,
$\relofs(M)$ must also be result-preserving with respect to $(Q,\calT,\inst{D})$.  
It remains to show that $\relofs(M)$ is a {\em minimal} result-preserving relation with 
respect to $(Q,\calT,\inst{D})$.
For the sake of contradiction, assume that $\relofs(M)$ is not minimal
in that respect. It follows that there is a tuple $\overline{t}\in \relofs(M)$
such that $\relofs(M) \backslash \{\overline{t}\}$ is result-preserving
with respect to $(Q,\calT, \inst{D})$. This means that 
for each $B \in \sign{\calT}$ and each tuple $t_B \in \pi_BQ(\inst{D})$,
there is a tuple $t \in \relofs(M) \backslash \{\overline{t}\}$
with $\pi_Bt = t_B$. This implies that for
each node $v \in V$, there is an edge 
$e \in M \backslash \{\tupleofs_E^-(\overline{t})\}$
with $v \in e$. We derive that      
$M \backslash \{\tupleofs_E^-(\overline{t})\}$
is a minimal edge cover of $H$, a contradiction 
to the minimality of $M$.

\subsection{Proof of Proposition \ref{prop:core_subset}}
{\noindent\bf Proposition \ref{prop:core_subset}.}
\textit{Given $(Q,\calT,\inst{D})$, each cover of $Q(\inst{D})$ over $\calT$ is a subset of $Q(\inst{D})$.}

\smallskip
\smallskip
\smallskip

Let 
$K$ be a cover of $Q(\inst{D})$ over $\calT$
and let 
$t\in K$ be an arbitrary tuple from $K$. 
We show that $t$ must be included in $Q(\inst{D})$.
For each $B \in \sign{\calT}$, let $t_B = \pi_{B}t$.
It holds that 
$t = \setindexjoin{t_B}{B}{\sign{\calT}}$.
As $K$
is result-preserving with respect to 
$(Q,\calT,\inst{D})$, $t_B$ must be included in 
$\pi_{B}Q(\inst{D})$ for each $B \in \sign{\calT}$.
Since by Proposition \ref{Q(D)-preserv=join_preserv},
$Q(\inst{D}) = \setindexjoin{\pi_BQ(\inst{D})}{B}{\sign{\calT}}$, 
it follows that $t$ is included in
$Q(\inst{D})$.

\subsection{Proof of Proposition \ref{prop:size_bounds_hypergraph}}
{\noindent\bf Proposition \ref{prop:size_bounds_hypergraph}.}
\textit{
Given $(Q,\calT,\inst{D})$, the size of each cover $K$ 
of $Q(\inst{D})$
over $\calT$ satisfies the inequalities
$\takemax{B\in\sign{\calT}}{\sizeof{\pi_BQ(\inst{D})}}$ $\leq$ 
$\sizeof{K}$ $\leq$ 
$\Sigma_{B\in \sign{\calT}}\sizeof{\pi_BQ(\inst{D})}$.
}

\smallskip
\smallskip
\smallskip

The first inequality holds due to $K$ being result-preserving 
with respect to $(Q,\calT,\inst{D})$. The second inequality is implied by 
Proposition \ref{characterize_cores}, since 
the hypergraph $H=(V,E)$ of $Q(\inst{D})$ over $\sign{\calT}$
must have a minimal edge cover $M$
with $\relof{M} = K$.
Each  hyperedge $e$ in $M$ must cover at least one 
node in $V$ which is not covered by any other hyperedge in $M$. 
Otherwise,  
$\minus{M}{e}$ would be an edge cover, 
which is a contradiction 
to the minimality of $M$. 
Hence, the total number of edges in $M$ is upper-bounded 
by $\sizeof{V}$.
As $\sizeof{V}=\Sigma_{B \in \sign{\calT}} \sizeof{\pi_B Q(\inst{D}) }$
and $\sizeof{M} = \sizeof{K}$, we derive
that the number of tuples in $K$ is upper-bounded 
by 
$\Sigma_{B \in \sign{\calT}} 
\sizeof{\pi_B Q(\inst{D})}$.

\subsection{Proof of Theorem \ref{theo:general_size_bounds}}
{\noindent\bf Theorem \ref{theo:general_size_bounds}.}
\textit{
Let $Q$ be a natural join query and $\calT$ a decomposition of $Q$.
}
\begin{enumerate}[(i)] 
\item \textit{For any database $\inst{D}$, each cover of the query result $Q(\inst{D})$ over 
$\calT$ has size $\calO(\sizeof{\inst{D}}^{\fhtw{\calT}})$.}

\item \textit{There are arbitrarily large databases $\inst{D}$ such that 
each cover of the query result $Q(\inst{D})$ over $\calT$
has size $\Omega(\sizeof{\inst{D}}^{\fhtw{\calT}})$.
}
\end{enumerate}

\smallskip
\smallskip
\smallskip

Our proof relies on the 
results that 
for any natural join query $Q$ and database $\inst{D}$, 
it holds $\sizeof{Q(\inst{D})} = \calO(\sizeof{\inst{D}}^{\rho^*(Q)})$
and there are arbitrarily large 
databases $\inst{D}$ with 
$\sizeof{Q(\inst{D})}$ $=$ $\Omega(\sizeof{\inst{D}}^{\rho^*(Q)})$
\cite{AtseriasGM13}.

Let $\calT =(T,\chi,\setindexel{\gamma_t}{t}{T})$.  
  Given a node $t$ in $T$ with $\chi(t) = B$ for some set $B$, 
  we recall that $\weight{\gamma_{t}} = \rho^*(Q_{B})$.
Moreover, if  $\weight{\gamma_{t}}$ 
  is maximal over all weight functions in $\calT$, 
  then $\weight{\gamma_{t}}$
$=\fhtw{\calT}$.

\underline{\em Proof of statement (i).}
Let $K$ be a cover of $Q(\inst{D})$ over $\calT$
and let $t$ be an arbitrary node of $T$ with $\chi(t)=B$
for some set $B$.
It holds 
$\sizeof{Q_B(\inst{D}_B)}$ 
$=$ 
$\calO(\sizeof{\inst{D}_B}^{\rho^*(Q_B)})$ 
\cite{AtseriasGM13},
thus,   
$\sizeof{Q_B(\inst{D}_B)}$
 $=$ 
 $\calO(\sizeof{\inst{D}_B}^{\fhtw{\calT}})$
 $=$ 
 $\calO(\sizeof{\inst{D}}^{\fhtw{\calT}})$.
Since
$\sizeof{\pi_B Q(\inst{D})}$ 
$\leq$ 
$\sizeof{Q_B (\inst{D}_B)}$  (Proposition 3.2 of \cite{FDB:TODS:2015}), 
it follows 
that 
$\sizeof{\pi_B Q(\inst{D})}$
 $=$ 
 $\calO(\sizeof{\inst{D}}^{\fhtw{\calT}})$.
 Using Proposition \ref{prop:size_bounds_hypergraph},   
 we conclude
$\sizeof{K}$ 
$\leq$ 
$\Sigma_{B\in \sign{\calT}} \sizeof{\pi_BQ(\inst{D})}$
$=$ 
 $\calO(\sizeof{\sign{\calT}} \cdot \sizeof{\inst{D}}^{\fhtw{\calT}})$
$=$ 
$\calO(\sizeof{\inst{D}}^{\fhtw{\calT}})$.

\underline{\em Proof of statement (ii).}
Let $t$ be a node in $T$ such that $\gamma_t$
has maximal weight and let $\chi(t)=B$.
There are arbitrarily large databases 
$\inst{D}'$ such that 
$\sizeof{Q_B(\inst{D}')}$ 
$=$ 
$\Omega(\sizeof{\inst{D}'}^{\rho^*(Q_B)})$
$=$
$\Omega(\sizeof{\inst{D}'}^{\fhtw{\calT}})$ 
\cite{AtseriasGM13}. 
For each such database $\inst{D}'$, 
there exists a database $\inst{D}$
with $\sizeof{\inst{D}} = \calO(\sizeof{\inst{D}'})$
and $\sizeof{\pi_B Q(\inst{D})}  = \Omega(\sizeof{Q_B(\inst{D}')})
= \Omega(\sizeof{\inst{D}'}^{\fhtw{\calT}})$
(Lemma 7.18 of \cite{FDB:TODS:2015}).
This means that there are 
arbitrarily large databases $\inst{D}$ such that 
$\sizeof{\pi_BQ(\inst{D})}$ 
$=$ 
$\Omega(\sizeof{\inst{D}}^{\fhtw{\calT}})$.
Due to Proposition \ref{prop:size_bounds_hypergraph}, 
each cover $K$ of $Q(\inst{D})$ over $\calT$
must be at least of size $\sizeof{\pi_BQ(\inst{D})}$, hence, 
$\sizeof{K}= \Omega(\inst{D}^{\fhtw{\calT}})$.

\subsection{Proof of Proposition \ref{prop:core_to_d_rep}}
{\noindent\bf Proposition \ref{prop:core_to_d_rep}.}
\textit{
Given $(Q,\calT,\inst{D})$, each cover $K$ of the query result $Q(\inst{D})$ over $\calT$ can be translated into a d-representation of $Q(\inst{D})$ of size $\calO(\sizeof{K})$ and in time $\widetilde{\calO}(\sizeof{K})$.
}

\smallskip
\smallskip
\smallskip

Using the algorithm in Figure \ref{fig:factorize_from_cover}, 
we construct from $K$ and $\calT$
a d-representation of $Q(\inst{D})$ encoded as 
a set M of maps. Recall that the constructed 
d-representation is over a d-tree $\calT'$ equivalent 
to $\calT$.  

\smallskip
\underline{\em Correctness of the construction.}
For each $m_A \in M$, we denote by $R_A$ the 
listing representation of $m_A$ as presented in Figure 
\ref{fig:d_rep_d_hashmap_rep}. 
For each bag attribute $A$ in $\calT'$, 
the set $\{A\} \cup \mathit{key}(A)$ constitutes 
a bag in the signature $\sign{\calT'}$ of $\calT'$. We write
 $B_A$ to express that the bag attribute of $B_A$
 is $A$.
 By the definition of d-representations, the query result  
 represented by the map set $M$ is
 $R = \setindexjoin{R_A}{B_A}{\sign{\calT'}}$
 \cite{FDB:TODS:2015}.
 It remains to show that $R = Q(\inst{D})$. 
 By construction
 of the maps in $M$, we have 
 $R_A = \pi_{B_A}K$ for each $B_A$. 
 For each $B_A \in \sign{\calT'}$,
there is a $B \in \sign{\calT}$ with $B_A \subseteq B$
(proof of Proposition 9.3 
in \cite{FDB:TODS:2015}).
Hence, by the definition of covers, we have 
$R_A = \pi_{B_A}K = \pi_{B_A} Q(\inst{D})$
for each $B_A$.
As $\calT'$ is a valid decomposition of $Q$,
it follows from Proposition \ref{Q(D)-preserv=join_preserv}
that $\setindexjoin{\pi_{B_A}K}{B_A}{\sign{\calT'}} = Q(\inst{D})$.
 Since for each $B_A$, we have $\pi_{B_A}K = R_A$
 and $R = \setindexjoin{R_{A}}{B_A}{\sign{\calT'}}$, 
 it follows $R = Q(\inst{D})$. 
 
\smallskip
\underline{\em Construction size and translation time.}
The number of the maps in $M$ is bounded by the number of attributes in 
$K$. We consider the cover $K$ sorted using a topological order of the 
decomposition $\calT'$, so that inserts into the multimaps become appends (alternatively, inserts in sorted order would take logarithmic time in the number of entries). For each tuple in $K$ we insert at most one tuple in the multimap of each attribute. Thus, the overall size of the set of multimaps, and thus of the d-representation, is $\calO(|K|)$ with respect to data complexity (the linear factor in the number of attributes is ignored).
The data complexity of the overall translation time is thus $\widetilde{\calO}(|K|)$.

\section{Missing Proofs of Section \ref{sec:core-plans}}
\subsection{Proof of Proposition \ref{prop:comp_core_two_rel}}
{\noindent\bf Proposition \ref{prop:comp_core_two_rel}.}
\textit{
Given two consistent relations $R_1$ and $R_2$, the cover-join computes a cover $K$ of their join result over the decomposition with bags $\sign{R_1}$ and $\sign{R_2}$ in time $\widetilde{\calO}(|R_1|+|R_2|)$ and with size $\max\{|R_1|, |R_2|\}\leq |K|\leq |R_1|+|R_2|$.
}

\smallskip
\smallskip
\smallskip

Let $Q = R_1 \Join R_2$, $\inst{D} = \{R_1, R_2\}$.
Moreover, let 
$\calT$ be the decomposition of $Q$ with bags
$\sign{R_1}$ and $\sign{R_2}$.
By Proposition 
\ref{characterize_cores}, 
a relation $K$ is a cover of $Q(\inst{D})$ 
over $\calT$ if and only if 
the hypergraph $H$ of  
$Q(\inst{D})$
 over the attribute sets
$\{\sign{R_1}, \sign{R_2}\}$
has a minimal edge cover $M$ with 
$\relof{M}=K$. 
The hypergraph 
$H$ is a collection of disjoint complete bipartite
subgraphs. 
The set of nodes of each such subgraph
corresponds to a maximal subset of tuples of the input relations 
agreeing on the join attributes. 
A minimal edge cover of $H$ is
a collection of minimal edge covers
for these subgraphs.
We construct a cover $K$ of minimum 
size such that each maximal subset of tuples in $K$ 
agreeing  on the join attributes corresponds 
to a minimal edge cover of one of the
complete bipartite
subgraphs of $H$.

\smallskip
\underline{\em Construction.}
Let $\calA$
be the set of common attributes of 
$R_1$ and $R_2$.
For $i \in \{1,2\}$ and $t \in \pi_{\calA}R_i$, we call
$\sigma_{\calA = t} R_i$ the {\em $t$-block} in $R_i$
and denote its size by $n^t_i$.
Since $R_1$ and $R_2$ are consistent,
for each $t$-block in $R_1$,
there must be a corresponding $t$-block in 
$R_2$, and vice-versa.
First, the algorithm sorts
$R_1$ and $R_2$
with respect to the values of the attributes in $\calA$.  
After sorting, the $t$-blocks 
occur in the same order
in both relations. 
The cover $K$
is constructed by performing 
the following procedure 
for each 
pair of corresponding $t$-blocks in 
$R_1$ and $R_2$.     
Without loss of generality, assume 
$n_1^t \geq n_2^t$.
For each $j < n_{2}^t$, the $j$-th
tuple $t'$ in the $t$-block of 
$R_1$ is combined 
with the $j$-th tuple $t''$ in the $t$-block 
of $R_2$
resulting in a new tuple $t' \Join t''$. 
Then, all remaining 
tuples in the $t$-block of  $R_1$ are combined with the 
$n_{2}^t$-th tuple in the $t$-block of 
$R_2$.
All new tuples are added to $K$.

\smallskip
\underline{\em Construction time.}
The sorting phase can be realised 
 in time 
$\softO(\sizeof{R_1}+\sizeof{R_2})$. 
 The phase for constructing the new tuples 
 can be done in one pass over the sorted relations. 
 Hence, the overall running time 
 of the described algorithm is $\softO(\sizeof{R_1}+\sizeof{R_2})$.

\smallskip
\underline{\em Size of the Cover.}
The size bounds 
$\max\{|R_1|, |R_2|\}\leq |K|\leq |R_1|+|R_2|$
follow from Proposition \ref{prop:size_bounds_hypergraph} and the assumption
 that $R_1$ and $R_2$ are consistent, 
 so we have $\pi_{\sign{R_i}} Q(\inst{D}) = R_i$
 for each $i \in \{1,2\}$.

\smallskip
Our algorithm above constructs a specific cover. 
Other covers can be constructed within the same
time bounds.  We exemplify the construction of some
further covers following different patterns. 
In our construction above, after combining the first 
$n_2^t-1$ tuples in the 
$t$-block of $R_1$
with the first $n_2^t-1$ tuples in the 
$t$-block of $R_2$, we combined the
last tuple 
in the $t$-block of $R_2$ with all remaining tuples 
in the $t$-block of $R_1$. Alternatively,
we can fix any tuple $t'$ in the $t$-block of 
 $R_2$, combine  the first $n_2^t-1$ tuples in the 
$t$-block of $R_1$
with all tuples besides $t'$  in the 
$t$-block of $R_2$
and then combine the remaining tuples
in the $t$-block of $R_1$ with $t'$.

\subsection{Proof of Lemma \ref{lemma:core-plan-acyclic}}
{\noindent\bf Lemma \ref{lemma:core-plan-acyclic}.}
\textit{
Given $(Q,\calJ,\inst{D})$ where $\inst{D}=\{R_i\}_{i\in[n]}$
is globally consistent with respect to $Q$, 
any cover-join plan over the join tree $\calJ$ 
computes a cover $K$ of $Q(\inst{D})$ over the decomposition corresponding to $\calJ$ in 
time $\widetilde{\calO}(|K|)$ and with size $\takemax{{i\in[n]}}{\sizeof{R_i}} \leq|K|\leq \sum_{i\in[n]}|R_i|$.
}

\smallskip
\smallskip
\smallskip

\noindent
\underline{\em Any cover-join plan 
over $\calJ$ computes 
a cover $K$ of $Q(\inst{D})$ over the decomposition corresponding to $\calJ$.}
\newline 
We show by induction on the 
structure of cover-join plans
that given $(Q,\calJ,\inst{D})$, 
where $\inst{D}$ is globally consistent 
with respect to $Q$, 
any cover-join plan over the join tree $\calJ$ 
computes a cover $K$ of $Q(\inst{D})$ over the decomposition 
corresponding to $\calJ$.

For the base case, assume that $\varphi$
consists of a single relation symbol $R$.
By Definition 
\ref{def:core-join-plan}, $\calJ$ consists 
of a single node $R$, hence, $Q=R$. 
The decomposition $\calT$ corresponding 
to $\calJ$ consists of a single bag $\sign{R}$. 
By Definition \ref{def:core-join-plan},
$\varphi$ returns the relation $R$.
By Definition \ref{def:cores}, $R$ is indeed 
the unique cover of $Q(\{R\})$ over
$\calT$.

Assume now that $\varphi$ is of the form
$\varphi_1 \corejoin \varphi_2$. By definition of
cover-join plans, there are subtrees $\calJ_1$
and $\calJ_2$ of $\calJ$ such that  
$\calJ = \calJ_1 \circ \calJ_2$ 
and each $\varphi_i$ is a cover-join plan over $\calJ_i$.
Let $\calT_1$ and $\calT_2$ be the decompositions
corresponding to $\calJ_1$ and $\calJ_2$, respectively.
The decomposition corresponding to $\calJ$
is obtained by connecting  
$\calT_1$ and $\calT_2$ by the same tree edge 
connecting  $\calJ_1$ and $\calJ_2$ in $\calJ$.   
We have $Q = Q_1 \Join Q_2$
where each $Q_i$ expresses the join of the relation 
symbols occurring in $\calJ_i$.
Moreover, $\inst{D} = \inst{D}_1 \cup \inst{D}_2$ where
 $\inst{D}_i= \setindexel{R}{R}{\sign{Q_i}}$, $i \in [2]$.
Note that for each $i \in [2]$, 
$Q_i$ is acyclic, $\calJ_i$ is a join tree 
of $Q_i$ and $\inst{D}_i$ is globally consistent 
with respect to $Q_i$. The latter follows simply
from the globally consistency of $\inst{D}$
with respect to $Q$.
Hence, by induction hypothesis,
each $\varphi_i$ returns a cover $K_i$
of $Q_i(\inst{D}_i)$ over $\calT_i$. 

Due to Proposition \ref{prop:comp_core_two_rel},
in case $K_1$ and $K_2$ are consistent,  
the cover-join operator computes a cover $K$
of $K_1 \Join K_2$ over the decomposition
with bags $\sign{K_1}$ and $\sign{K_2}$. 
Thus, by Definition \ref{def:core-join-plan},
the plan $\varphi$ returns $K$.
We proceed as follows. 
First, we show that 
$K_1$ and  $K_2$      
must be consistent. Then, we 
prove that $K$ is a cover of 
$Q(\inst{D})$
 over $\calT$, that is, 
 $K$ is result-preserving with respect to 
$(Q, \calT, \inst{D})$
and it is minimal in this respect.

\begin{itemize}
\item $K_1$ and $K_2$ are consistent:
Let $R_1$ and $R_2$ be the two relation symbols 
incident to the single edge connecting $\calJ_1$
and $\calJ_2$ in $\calJ$ and let $\calB$ be the set of common 
attributes of these relation symbols.
 Let $\calA$ be the set of common attributes 
of $K_1$ and $K_2$. 
We first show that $\calA \subseteq \calB$.
Let $A \in \calA$.
Since each $K_i$ is computed by the 
plan $\calJ_i$, there must be at least one 
relation symbol $R_1'$ in $\calJ_1$
and at least one relation symbol $R_2'$ in  
$\calJ_2$ containing $A$ in their schemas. 
Due to the construction of join trees,
$A$ must occur in the schemas of all relation symbols 
on the  single path between $R_1'$ and $R_2'$.
Since $R_1$ and $R_2$ are on this path, both must 
include  $A$. Hence, $\calA \subseteq \calB$.

Since $\inst{D}$ is globally consistent, 
the relations $R_1$ and $R_2$ must be consistent as well. 
As each $K_i$ is result-preserving
with respect to $(Q_i,\calT_i,\inst{D}_i)$,
$\pi_{\sign{R_i}} Q_i(\inst{D}_i) = R_i$ (due to global consistency)
and $\calB \subseteq \sign{R_i}$,
it follows $\pi_{\calB}K_1 = \pi_{\calB}Q_1(\inst{D}_1) = \pi_{\calB}R_1
= \pi_{\calB}R_2 = \pi_{\calB}Q_2(\inst{D}_2) = \pi_{\calB}K_2$.
As $\calA \subseteq \calB$, the relations $K_1$
and $K_2$ must be consistent.

\item $K$ is result-preserving with respect to 
$(Q, \calT, \inst{D})$:
Let $B$ be an arbitrary bag of $\calT$.
Since $\calT$ corresponds to $\calJ$,
the join tree $\calJ$ must have a node $R$
with $\sign{R} = B$. 
Without loss of generality, assume that 
 $R \in\inst{D}_1$ (the other case is handled 
 along the same lines).
Since, by induction hypothesis, 
 $K_1$ is result-preserving with respect to
 $(Q_1, \calT_1, \inst{D}_1)$ and $\inst{D}_1$ is globally consistent, 
 we have 
 $R= \pi_{\sign{R}}K_1$. 
 Since $\pi_{\sign{K_1}}K = K_1$
 and $\sign{R} \subseteq \sign{K_1}$,
 we get $R= \pi_{\sign{R}}K$.
 Using the global consistency of $\inst{D}$
 with respect to $Q$, we conclude 
 $\pi_{B}Q(\inst{D}) =  R=  \pi_{B}K$.

\item $K$ is a {\em minimal} result-preserving relation 
with respect to 
$(Q, \calT, \inst{D})$:
For the sake of contradiction, assume
that $K$ is not minimal in this respect. 
This means that there is 
a tuple $t^- \in K$ such that $\minus{K}{t^-}$ is still 
result-preserving with respect to $(Q, \calT, \inst{D})$.
It follows that 
$\pi_{\sign{K_i}} (\minus{K}{t^-})$ is result-preserving 
with respect to 
$(Q_i, \calT_i, \inst{D}_i)$
for each $i \in \indexnat{2}$.
Observe that the
minimal edge cover $M$ with $\relof{M} = K$ 
in the hypergraph of 
$K_1 \Join K_2$ 
over the attribute sets $\{\sign{K_1}, \sign{K_2}\}$ 
must contain an edge $e^-$ connecting 
$\pi_{\sign{K_1}}t^-$ and $\pi_{\sign{K_2}}t^-$.
This implies that $M$ cannot have two further edges $e_1$ and $e_2$ 
such that $e_1$ covers $\pi_{\sign{K_1}}t^-$ and $e_2$ covers 
$\pi_{\sign{K_2}}t^-$. Indeed, in this case, $\minus{M}{e^-}$ would be an edge cover, contradicting the minimality of 
$M$. 
Hence, there is no tuple $t \neq t^-$
in $K$ with $\pi_{\sign{K_1}}t^- = 
\pi_{\sign{K_1}}t$ or 
there is no tuple $t \neq t^-$
in $K$ with 
$\pi_{\sign{K_2}}t^-= \pi_{\sign{K_2}}t$.
It follows that  
$\pi_{\sign{K_1}} (\minus{K}{t^-}) \subset 
\pi_{\sign{K_1}} K$ or 
$\pi_{\sign{K_2}} (\minus{K}{t^-}) \subset 
\pi_{\sign{K_2}} K$. 
Using the consistency of $K_1$ and $K_2$, we obtain 
$\pi_{\sign{K_1}}$ $(\minus{K}{t^-})$ 
$\subset \pi_{\sign{K_1}} K$ $=$ 
$K_1$ or 
$\pi_{\sign{K_2}}$ $(\minus{K}{t^-})$ 
$\subset \pi_{\sign{K_2}} K$ $=$ $K_2$.
However, as we noticed that 
$\pi_{\sign{K_i}} (\minus{K}{t^-})$ is 
result-preserving with respect to
$(Q_i, \calT_i, \inst{D}_i)$
for each $i \in \indexnat{2}$,
the statement of the last sentence  contradicts the induction hypothesis that 
each $K_i$ is a {\em minimal} result-preserving relation with respect to 
$(Q_i, \calT_i, \inst{D}_i)$.

\end{itemize}

\underline{\em Size of $K$.}
From the global consistency of $\inst{D}$
with respect to $Q$
and Proposition 
 \ref{prop:size_bounds_hypergraph}, it follows 
 for any cover  
 $K$ of $Q(\inst{D})$ over the tree decomposition corresponding to 
 $\calJ$ that
$\takemax{{i\in[n]}}{\sizeof{R_i}} \leq |K|\leq \sum_{i\in[n]}|R_i|$.

\underline{\em Computation time for $K$.}
By Proposition  \ref{prop:comp_core_two_rel}, we can design
an algorithm for the cover-join operator which for every 
two input covers $K_1$ and $K_2$, computes 
a cover-join result of size 
$\calO(\sizeof{K_1}+ \sizeof{K_2})$
and
in time $\softO(\sizeof{K_1}+ \sizeof{K_2})$. Hence, given a triple 
$(Q,\calJ,\inst{D})$
and a cover-join plan $\varphi$ over $\calJ$,
starting from the innermost expressions of $\varphi$,
we can compute a cover $K$ of $Q(\inst{D})$
over the tree decomposition corresponding to $\calJ$ in time 
$\softO(\sizeof{K})$.

\section{Missing Details and Proofs in Section \ref{sec:applications}}
\label{app:sec:faq}
Given the hypergraph $H$ of an FAQ and 
an attribute set $U$, we denote by $H_U$ the hypergraph 
obtained from 
$H$ by restricting each hyperedge in $H$ to the attributes
in $U$.
For the rest of this section we fix an FAQ $\varphi$
as written in  (\ref{eq:faq}).

\subsection{Recap on FAQs}
Indicator projections are used in the \textsf{InsideOut} algorithm 
\cite{FAQ:PODS:2016} solving the FAQ-problem.
They will also occur in our construction of FAQ-covers. 
 
\begin{definition}[Indicator projections]
Given two attribute sets $S$ and $T$ with 
$S \cap T \neq \emptyset$ and a function $\psi_S$, the function
$\psi_{S/T}: \prod_{A\in (S \cap T)}\dom(A) \rightarrow \Dom$
defined by 
$$\psi_{S/T}({\textsf a}_{S \cap T}) =
  \begin{cases}
    \one
     & \quad \exists {\textsf b}_{S} \text{ s.t. } \psi_S(\valb_S) \neq 0 \text{ and }
       \vala_{S\cap T} = \valb_{S\cap T}, \\
    \zero  & \quad \text{otherwise } \\
  \end{cases}
$$
is called the indicator projection of $\psi_S$ onto $T$.
\end{definition}
In particular, if $S \subseteq T$, then 
$\psi_{S/T}(\vala_{S}) = \one$ if and
only if $\psi_S(\vala_{S}) \neq \zero$.

\smallskip
{\noindent\bf Equivalent attribute orderings.}
A $\varphi$-equivalent attribute ordering $\tau = \tau(1), \ldots , \tau(n)$ 
is a permutation of the indices of the attributes in $\calV$ 
 satisfying the following conditions:
\begin{enumerate}[(a)]
\item $\{A_{\tau(1)}, \ldots , A_{\tau(f)}\} = \{A_1, \ldots , A_f\}$ and
\item 
$$
\varphi'(\vala_{\{A_{\tau(1)},\ldots,A_{\tau(f)}\}}) = \underset{a_{\tau(f+1)}\in\textsf{dom}(A_{\tau(f+1)})}{\bigoplus\ ^{(\tau(f+1))}} \cdots \underset{a_{\tau(n)}\in\textsf{dom}(A_{\tau(n)})}{\bigoplus\ ^{(\tau(n))}}\  \underset{S\in\mathcal{E}}{\bigotimes}\ \psi_S({\sf a}_S)$$
is equivalent to $\varphi$ irrespective of the definition of the input functions 
$\psi_S$.
\end{enumerate}
We denote by $\evo(\varphi)$ the set of all $\varphi$-equivalent attribute orderings.

\paragraph*{The \textsf{InsideOut} algorithm}
Given an FAQ $\varphi$, a database $\inst{D}$ and a $\varphi$-equivalent 
attribute ordering, the \textsf{InsideOut} algorithm computes the 
listing representation of $\varphi(\inst{D})$.  
The algorithm first rewrites the query
according to the given attribute ordering and then processes
the resulting query in two phases:
{\em bound attribute elimination} and {\em output computation}.
We sketch the main steps of the algorithm on input 
$\varphi$, some database $\inst{D}$ and the attribute ordering that 
corresponds to the identity permutation. Thus, 
the initial rewriting step
does not change the structure of $\varphi$.

In the bound attribute elimination phase,
 the algorithm eliminates 
 attributes $A_{f+1},$ $\ldots,$ $A_{n}$   
along with their corresponding aggregate operators
in reverse order. 
 When eliminating an attribute $A_j$ 
 it distinguishes between the cases whether 
$\bigoplus^{(j)}$ is 
different from $\bigotimes$ or not. 
We demonstrate the two cases in the elimination step for $A_n$.
In case that $\bigoplus^{(n)}$ is 
different from $\bigotimes$, the algorithm first rewrites the query as follows
\begin{align*}
 &  
\underset{a_{f+1}\in\dom(A_{f+1})}{\bigoplus\ ^{(f+1)}} \cdots \underset{a_{n}\in\dom(A_{n})}{\bigoplus\ ^{(n)}}  \ \underset{S\in\mathcal{E}}{\bigotimes} \ \psi_S({\sf a}_S) \\
= & \underset{a_{f+1}\in\dom(A_{f+1})}{\bigoplus\ ^{(f+1)}} \cdots \underset{a_{n-1}\in\dom(A_{n-1})}{\bigoplus\ ^{(n-1)}} \ \underset{S\in\calE\backslash \partial(n)}{\bigotimes} \psi_S(\vala_S)
\otimes \Big(\underbrace{\underset{a_n \in \dom(A_n)}{\bigoplus\ ^{(n)}} \bigotimes_{S \in \partial(n)} \psi_S(\vala_S)}_{\delta}\Big),
\end{align*}
where $\partial(n)= \{S \in \calE \mid A_n \in S\}$
and $U_n = \bigcup_{S\in \partial(n)} S$.
The correctness of the rewriting 
follows from the distributivity of $\otimes$ over $\oplus^{(n)}$.
Then, the algorithm computes the listing representation 
of a function 
$\psi'_{U_n\backslash \{A_n\}}$ such that 
replacing $\delta$
by $\psi'_{U_n\backslash \{A_n\}}$
does not change the semantics of $\varphi$.
Observe that the cartesian product of 
the domains of the attributes in $U_n \backslash \{A_n\}$ can contain 
tuples $\vala_{U_n \backslash \{A_n\}}$ such that 
\begin{inparaenum}[(i)]
\item there is a $\psi_S$ with $S \in \calE \backslash \partial(n)$,  
$S \cap (U_n \backslash \{A_n\}) \neq \emptyset$ and
\item there is no $\valb_S$ that agrees with 
$\vala_{U_n \backslash \{A_n\}}$ on the common attributes 
and $\psi_S(\valb_S)\neq 0$.
\end{inparaenum}
Such tuples will not occur in the final result. To rule them out 
in advance, 
indicator projections are used inside $\psi'_{U_n\backslash \{A_n\}}$. 
The function $\psi'_{U_n\backslash \{A_n\}}$ is  defined as   
 $$\psi'_{U_n\backslash \{A_n\}}(\vala_{U_n\backslash \{A_n\}}) = \underset{a_n \in \dom(A_n)}{\bigoplus\ ^{(n)}} \bigg[ \Big( \bigotimes_{S \in \partial(n)}\psi_S(\vala_S) \Big) \otimes \Big(\bigotimes_{\substack{S \notin \partial(n) \\
 S \cap U_n \neq \emptyset}}  \psi_{S/U_{n}} (\vala_{S\cap U_n}) \Big)\bigg].$$
The computation of the listing representation of this function 
requires  the computation of the join of the 
listing representations of the functions $\psi_S$ with $S \in \partial(n)$
and the indicator projections. 
The computation time for this elimination step is  
$\softO(|\inst{D}|^{\rho^{\ast}(H_{U_n})})$.

In case that $\bigoplus^{(n)}$ is equal to $\bigotimes$, 
the formula is rewritten as follows

\begin{align*}
& 
\underset{a_{f+1}\in\dom(A_{f+1})}{\bigoplus\ ^{(f+1)}} \cdots \underset{a_{n}\in\dom(A_{n})}{\bigoplus\ ^{(n)}}  \ \underset{S\in\mathcal{E}}{\bigotimes} \ \psi_S(\vala_S) \\
= & \underset{a_{f+1}\in\dom(A_{f+1})}{\bigoplus\ ^{(f+1)}} \cdots 
\underset{a_{n-1}\in\dom(A_{n-1})}{\bigoplus\ ^{(n-1)}} \
\underset{a_{n}\in\dom(A_{n})}{\bigotimes}  \ 
\underset{S\in\mathcal{E}}{\bigotimes} \ \psi_S(\vala_S) \\
= &  \underset{a_{f+1}\in\dom(A_{f+1})}{\bigoplus\ ^{(f+1)}} \cdots \underset{a_{n-1}\in\dom(A_{n-1})}{\bigoplus\ ^{(n-1)}}  \ 
\underset{S\notin\partial(n)}{\bigotimes} \ \psi_S(\vala_S)^{|\dom(A_n)|}
\underset{S \in \partial(A_n)}{\bigotimes} \ \underbrace{\underset{a_n \in \dom(A_n)}{\bigotimes} \
\psi_S(\vala_S)}_{\delta^S},
\end{align*}
where $\partial(n)$ is defined as above.  
Then, the algorithm  computes for each 
$S \notin \partial(n)$, a function 
$\psi_S'$ equivalent to 
$\psi_S^{|\dom(A_n)|}$
and for each $S \in \partial(n)$,
a function $\psi_{S\backslash A_n}'$
equivalent to $\delta^S$. This elimination step can be realised in 
time 
$\softO(|\inst{D}|)$.

After the elimination of all bound attributes
we are 
left with a formula $\varphi'_{\vala_{\{A_1, \ldots ,A_f\}}}$ without any 
bound attributes.  
In the output computation phase 
the algorithm first computes 
(a factorized representation of) the set of tuples 
$\vala_{\{A_1, \ldots ,A_f\}}$
for which $\varphi'_{\vala_{\{A_1, \ldots ,A_f\}}}(\vala_{\{A_1, \ldots ,A_f\}})\neq \zero$ and then reports the output.

\smallskip
Before giving the overall running time of \textsf{InsideOut}, we 
introduce elimination hypergraph sequences corresponding 
to attribute orderings.

\paragraph*{Elimination hypergraph sequence}
Given a $\varphi$-equivalent 
attribute ordering $\tau = \tau(1),$ $\ldots ,$ $\tau(n)$, 
we recursively define 
the elimination hypergraph sequence 
$H_n^{\tau}, \ldots , 
H_1^{\tau}$ associated with 
$\tau$. For each $j$ with $n \geq j \geq 1$, we additionally define 
two sets $U_j^{\tau}$ and $\partial^{\tau}(j)$.
For the sake of readability, in the following we skip the 
superscript $\tau$ in our notation.

We set $H_n = (\calV_n, \calE_n) = H$ and 
define
$\partial(n) = \{S \in \calE_n \mid A_{\tau(n)} \in S\}$ and $U_n = \bigcup_{S \in \partial(n)}S.$

For each $j$ with $n-1 \geq j \geq 1$, we define:
\begin{itemize}
\item If $\bigoplus\ ^{(\tau(j+1))} = \bigotimes$, then, 
$\calV_j = \{A_{\tau(1)}, \ldots , A_{\tau(j)}\}$ and 
$\calE_j$ is 
obtained from $\calE_{j+1}$ by removing $A_{\tau(j+1)}$ 
from all edges in 
$\calE_{j+1}$.
\item Otherwise, $\calV_j = \{A_{\tau(1)}, \ldots , A_{\tau(j)}\}$ and 
   $\calE_j = (\calE_{j+1} \backslash \partial(j+1) ) \cup 
   (U_{j+1} \backslash \{A_{\tau(j+1)}\})$.
\end{itemize}
We further set 
$\partial(j) = \{S \in \calE_j \mid A_{\tau(j)} \in S\}$ and 
$U_j = \bigcup_{S \in \partial(j)}S.$

\paragraph*{Running time of \textsf{InsideOut}}
For a $\varphi$-equivalent
attribute ordering $\tau$,
let $K = [f] \cup \{ j \mid j > f, \oplus^{(\tau(j))} \neq \otimes \}$.
The FAQ-width of $\tau$ is defined as 
$\faqw{\tau} = 
\takemax{j \in K}{\rho^{\ast}(H_{U_j^{\tau}})}$.
For a given $\tau$, \textsf{InsideOut} runs in time
$\softO(|\inst{D}|^{\faqw{\tau}} + Z)$ 
where $Z$ is the size 
of the output.
The FAQ-width of $\varphi$ is defined as
$\faqw{\varphi} = \takemin{\tau \in \evo(\varphi)}{\faqw{\tau}}$.
 Hence, given the best attribute ordering (i.e., with smallest FAQ-width),
 the running time of 
\textsf{InsideOut} is 
$\softO(|\inst{D}|^{\faqw{\varphi}} + Z)$.

\paragraph*{From attribute orderings to decompositions}
We say that $\calT$ is 
a decomposition of $\varphi$
if $\calT$ is a decomposition of the hypergraph $H$
of $\varphi$.

\begin{proposition}[\cite{FAQ:KhamisNRR15:arxiv}, Proposition C.2]\label{prop:induced_hypertree}  
For any FAQ $\varphi$ without bound attributes 
and any $\varphi$-equivalent attribute ordering
$\tau$, one can construct  a  
decomposition $\calT$ of $\varphi$
with $\fhtw{\calT} \leq \faqw{\tau}$.
\end{proposition}

\subsection{Covers for FAQs}
Given two input functions $\psi_S$ and $\psi_T$
with $T \subseteq S$, we can always compute 
the function $\psi'_S = \psi_S \otimes \psi_T$
in time $\softO(|R_{\psi_S}| + |R_{\psi_T}|)$  
and replace $\psi_S \otimes \psi_T$     
by $\psi'_S$ without changing the semantics of 
the FAQ. 
To do this, we first sort the listing representations 
$R_{\psi_S}$ and $R_{\psi_T}$ of
 $\psi_S$ and $\psi_T$ on the attributes in $T$. 
 During a subsequent 
scan through both relations we add 
for each pair $\vala_{S \cup \{\psi_S(S)\}} \in R_{\psi_S}$ and
$\valb_{T\cup \{\psi_T(T)\}} \in R_{\psi_T}$
with $\vala_T = \valb_T$,
the tuple 
$\valc_{S\cup \{\psi'_S(S)\}}$
with 
$\valc_S = \vala_S$ and $\valc_{\{\psi'_S(S)\}}= 
\psi_S(\vala_S) \otimes \psi_T(\valb_T)$
to the listing representation of $\psi'_S$.
Hence, in the following 
we assume, without loss of generality, that $\varphi$
does not contain any function 
whose attributes are included in 
the attribute set of another function.

\paragraph*{Bag functions}
Given an FAQ $\varphi$ without bound attributes and 
a  decomposition of $\varphi$,
we define bag functions 
which are the counterparts of bag relations in case of join queries. 
 Our goal is to define for each bag $B$ of $\calT$,
  a function $\beta_B$ such that  
  $\varphi(\vala_{\calV}) = 
  \bigotimes_{B \in \sign{\calT}} \beta_B(\vala_{B})$.
  While in case of join queries 
   it is harmless to include all relations sharing attributes with 
   $B$ into the join computing  the bag relation of $B$, 
   in case of FAQs
   we have to be a bit
   careful. Including the same input function 
   into the computation of bag functions 
   of several bags can violate the above equality. 
   Therefore, in the definition below we use a mapping 
    from input functions to bags.
   To keep the sizes of the bag functions small 
   we also use indicator projections which achieve pairwise 
   consistency between 
   listing representations of  
   bag functions sharing attributes.

\begin{definition}[Bag functions]\label{def:bag_functions}
Given an FAQ $\varphi$ without bound attributes 
and a  decomposition $\calT$
of $\varphi$, 
a set $\{\beta_B\}_{B \in \sign{\calT}}$
is called a set of bag functions for $\varphi$
and $\calT$ 
if there is a mapping 
$m: \calE \rightarrow \sign{\calT}$ such that 
$S \subseteq m(S)$ for each $S \in \calE$ and 
$\beta_B$ is 
defined by
$$\beta_{B}(\vala_B) = 
 \underset{{S \in \calE:S \cap B \neq \emptyset}}{\bigotimes} \psi_{S/B}({\vala}_{B\cap S}) \ \otimes \underset{S \in \calE: m(S) = B}{\bigotimes} \psi_S(\vala_{S})$$
for each $B \in \sign{\calT}$.

We define 
$\calB(\varphi,\calT) = \{\setindexel{\beta_B}{B}{\sign{\calT}}
\mid$  $\setindexel{\beta_B}{B}{\sign{\calT}}$ 
is a set of bag functions for $\varphi$ and $\calT\}$.
\end{definition}
Note that since each hyperedge in the hypergraph of $\varphi$
must be included in at least one bag of the decomposition, 
one can always find a mapping 
$m$ meeting the condition given in the above definition. 
Observe also that for bags $B$ to which no input function is mapped, 
the function $\beta_B$ is just the product of indicator projections of all
$\psi_S$ sharing attributes with $B$ onto $B$.

\begin{observation}\label{obs:faq_equiv_bag_fac}
Given an FAQ $\varphi$ without bound attributes, 
a  decomposition $\calT$
of $\varphi$ and a set 
$\{\beta_B\}_{B \in \sign{\calT}} \in \calB(\varphi,\calT)$,
it holds 
$$\varphi(\vala_{\calV}) = \underset{B\in \sign{\calT}}{\bigotimes}\ 
\beta_B({\vala}_{B}).$$
\end{observation}

Given $\{\beta_B\}_{B \in \sign{\calT}} \in \calB(\varphi,\calT)$,
we denote by $\ext(\calT,\{\beta_B\}_{B \in \sign{\calT}})$
the decomposition obtained from $\calT$ by adding 
into each bag $B$ the attribute $\beta_B(B)$.  
Observe that if $\calT$ is a decomposition of  
$\varphi(\vala_{\calV}) = \bigotimes_{B\in \sign{\calT}} 
\beta_B({\vala}_{B})$, then, $\ext(\calT,\{\beta_B\}_{B \in \sign{\calT}})$
is a decomposition of the query joining 
the listing representations of the functions 
$\beta_B$. Moreover, $\calT$
and $\ext(\calT,\{\beta_B\}_{B \in \sign{\calT}})$
have the same fractional hypertree width.

\paragraph*{Covers of FAQ results}
We turn towards the general case where 
FAQs can contain bound attributes also. 
Let $\tau = \tau_1\tau_2$
be a $\varphi$-equivalent attribute ordering 
where $\tau_1$ consists of the free and $\tau_2$
consists of the bound attributes in $\varphi$.
By $\varphi^{\tau}_{\free}$ we denote the FAQ constructed 
by the \textsf{InsideOut} algorithm after eliminating all bound attributes in $\varphi$
according to the ordering $\tau_2$.
We write $(\varphi,\tau,\calT, \inst{D})$ to express
that $\varphi$ is an FAQ, $\tau$ is a $\varphi$-equivalent 
attribute ordering, $\calT$ is a  decomposition
of  $\varphi^{\tau}_{\free}$ with $\fhtw{\calT} \leq \faqw{\tau_1}$
and $\inst{D}$ is an input database for $\varphi$.
Note that due to Proposition \ref{prop:induced_hypertree},
for any $\tau$
such a decomposition $\calT$ is always constructible.

\begin{definition}[Covers of FAQ results]\label{def:faq_covers}
Given $(\varphi,\tau, \calT,\inst{D})$,  
a relation $K$ is a cover of the query result $\varphi(\inst{D})$ over $\calT$
induced by $\tau$ if  there is a set  
$\{\beta_B\}_{B \in \sign{\calT}}\in \calB(\varphi^{\tau}_{\free},\calT)$ 
such that $K$ is a cover 
of the join of the relations $\setindexel{R_{\beta_B}}{B}{\sign{\calT}}$ over 
$\ext(\calT,\{\beta_B\}_{B \in \sign{\calT}})$.

We call $\{\beta_B\}_{B \in \sign{\calT}}$ the set of bag functions underlying $K$.
\end{definition}
Observe that if $K$ is a cover of $\varphi(\inst{D})$ over $\calT$ with underlying bag functions 
$\{\beta_B\}_{B \in \sign{\calT}}$, then, 
$\pi_{\calV_{\text{free}}} K$ must be a cover
of $\setindexjoin{\pi_B}{B}{\sign{\calT}}R_{\beta_B}$
over $\calT$.

The following Proposition relies
on Lemma 
\ref{lemma:core-plan-acyclic} 
and Theorem \ref{theo:time_core_join_plan}
which give an 
upper bound on the time complexity for constructing 
covers of join results. 

\begin{proposition}\label{prop:core_FAQ_tau_upper_bound}
Given $(\varphi,\tau, \calT,\inst{D})$, 
a cover of the query result $\varphi(\inst{D})$ over $\calT$ induced by $\tau$ 
can be computed in time 
 $\softO(|\inst{D}|^{\faqw{\tau}})$.
\end{proposition}

\begin{proof}
\underline{\em Construction.}
Let $\tau = \tau_1\tau_2$
where $\tau_1$ consists of the
free and $\tau_2$ consists of the 
bound attributes in $\varphi$. 
We first run \textsf{InsideOut} on $\varphi$ according to 
the attribute ordering $\tau_2$ until
all bound attributes are eliminated and we obtain 
$\varphi^{\tau}_{\free}$. 
Then, we construct a set  
$\{\beta_B\}_{B \in \sign{\calT}} \in \calB(\varphi^{\tau}_{\free},\calT)$
of bag functions.
Finally, using a cover-join plan as introduced 
in Definition \ref{def:core-join-plan}, we construct  a 
cover $K$ of the join of the relations $\setindexel{R_{\beta_B}}{B}{\sign{\calT}}$
over $\ext(\calT,\{\beta_B\}_{B \in \sign{\calT}})$.

\smallskip
\underline{\em Construction time.}  
The FAQ $\varphi^{\tau}_{\free}$ can be computed in time 
$\softO(|\inst{D}|^{\faqw{\tau_2}})$
\cite{FAQ:PODS:2016}.
The construction of the bag functions $\{\beta_B\}_{B \in \sign{\calT}}$ 
can be realised via the 
computation of the bag relations of $\calT$.
By Proposition \ref{prop:rewriting2}, 
the size of the listing representations 
of these bag functions is 
$\calO(|\inst{D}|^{\fhtw{\calT}})$ 
and their  computation time is 
$\softO(|\inst{D}|^{\fhtw{\calT}})$. 
By Theorem \ref{lemma:core-plan-acyclic},
$K$ can be computed
in time 
$\softO(\Sigma_{B \in \sign{\calT}}
|R_{\beta_B}|)$. 
Hence, the time for computing  
$K$ from $\varphi^{\tau}_{\free}$
is $\softO(|\inst{D}|^{\fhtw{\calT}})$.
Since $\faqw{\tau} = \takemax{1 \leq i \leq 2}{\faqw{\tau_i}}$
and $\fhtw{\calT} \leq \faqw{\tau_1}$ 
(by construction),
the overall computation time is
$\softO(|\inst{D}|^{\faqw{\tau}})$.
\end{proof}

Theorem \ref{theo:core_FAQ_upper_bound} is an immediate corollary: 

\smallskip
\smallskip
{\noindent\bf Theorem \ref{theo:core_FAQ_upper_bound}.}
\textit{
For any FAQ $\varphi$ and database $\inst{D}$, a cover of the query result $\varphi(\inst{D})$
can be computed in time 
 $\widetilde{\calO}(\sizeof{\inst{D}}^{\faqw{\varphi}})$.
}

\subsection{Enumeration of Tuples in FAQ Results using Covers}
\label{sec:app_faq_enum}
Any enumeration algorithm on covers
of join results 
can easily be turned into an enumeration algorithm
on covers of FAQ-results. 
Assume that $K$ is a cover
of the result of the FAQ $\varphi$ over some decomposition 
$\calT$ (induced by some attribute ordering).
Let $\setindexel{\beta_B}{B}{\sign{\calT}}$
be the underlying set of bag functions.  
Recall that the set of attributes of $K$
is $\calV_{\text{free}} \cup \setindexel{\beta_B(B)}{B}{\sign{\calT}}$
and the set of attributes of the 
listing representation of $\varphi$ must be 
$\calV_{\text{free}} \cup \{\varphi(\calV_{\text{free}})\}$. 
To enumerate the listing representation of $\varphi$, 
we can run any enumeration  
algorithm on $K$ with respect to the decomposition 
$\ext(\calT, \setindexel{\beta_B}{B}{\sign{\calT}})$
and adapt its output as follows. 
For each output tuple 
$\vala_{\calV_{\text{free}} \cup \setindexel{\beta_B(B)}{B}{\sign{\calT}}}$,
we output the tuple
$\valb_{\calV_{\text{free}} \cup \{\varphi(\calV_{\text{free}})\}}$
that agrees with 
$\vala_{\calV_{\text{free}} \cup \setindexel{\beta_B(B)}{B}{\sign{\calT}}}$
on $\calV_{\text{free}}$ and where the 
$\varphi(\calV_{\text{free}})$-value
is defined by $\bigotimes_{B \in \sign{\calT}}\vala_{\{\beta_B(B)\} }$.

The following proposition shows that by 
this strategy we indeed enumerate the listing representation of 
$\varphi$.

\begin{proposition}\label{prop:equiv_phi_join}
Given $(\varphi,\tau, \calT,\inst{D})$, let $K$
be a cover of the query result of $\varphi(\inst{D})$ over $\calT$ induced by 
$\tau$
and
let $\setindexel{\beta_B}{B}{\calT}$ be the set of bag functions
underlying $K$.
It holds
$$\varphi(\vala_{\calV_{\text{free}}}) = v \neq 0 \text{ for some } v \in \Dom$$
\center{if and only if}
$$\exists \valb_{\calV_{\text{free}} \cup \setindexel{\beta_B(B)}{B}{\sign{\calT}}} \in \setindexjoin{\pi_{B \cup \{\beta_B(B) \} }K}{B}{\sign{\calT}},
\ \vala_{\calV_{\text{free}}} = \valb_{\calV_{\text{free}}} 
\text{ and }
\bigotimes_{B \in \sign{\calT}} \valb_{ \{\beta_B(B)\}} = v.$$
\end{proposition}

\begin{proof}
Let $\varphi_{\free}^{\tau} = \bigotimes_{S' \in \calE'} \psi_{S'}$. Then, 
\begin{align*}
&\  \ \  \varphi(\vala_{\calV_{\text{free}}}) = v \neq \zero  \\
\overset{(1)}{\Longleftrightarrow} & \ \ \bigotimes_{S' \in \calE'} \psi_{S'}(\vala_{S'}) = v \neq \zero \\
\overset{(2)}{\Longleftrightarrow}  & 
\bigotimes_{B \in \sign{\calT}} \beta_B(\vala_{B}) = v \neq \zero \\
\overset{(3)}{\Longleftrightarrow}  & 
\ \ \ 
\exists \valb_{\calV_{\text{free}} \cup \setindexel{\beta_B(B)}{B}{\sign{\calT}}} \in \setindexjoin{R_{\beta_B}}{B}{\sign{\calT}},
\ \vala_{\calV_{\text{free}}} = \valb_{\calV_{\text{free}}} 
\text{ and }
\bigotimes_{B \in \sign{\calT}} \valb_{\{\beta_B(B)\} } = v \\
\overset{(4)}{\Longleftrightarrow}  &
\ \ \ 
\exists \valb_{\calV_{\text{free}} \cup \setindexel{\beta_B(B)}{B}{\sign{\calT}}} \in \setindexjoin{\pi_{B \cup \{\beta_B(B) \} }K}{B}{\sign{\calT}},
\ \vala_{\calV_{\text{free}}} = \valb_{\calV_{\text{free}}} 
\text{ and } \\
& \ \ \ \bigotimes_{B \in \sign{\calT}} \valb_{\{\beta_B(B)\}} = v.
\end{align*}
Equivalence (1) holds by the correctness of the \textsf{InsideOut} algorithm. 
The second equivalence holds by Observation \ref{obs:faq_equiv_bag_fac}.
Equivalence (3) follows from 
the simple observation that the product of 
functions corresponds to the join of their 
listing representations.
The last equivalence follows  
from Proposition \ref{Q(D)-preserv=join_preserv}
which guarantees that 
$\setindexjoin{R_{\beta_B}}{B}{\sign{\calT}}$
is equal to
$\setindexjoin{\pi_{B \cup \{\beta_B(B) \} }K}{B}{\sign{\calT}}$.
\end{proof}

Thus, our enumeration result for covers of join results 
carries over to covers of FAQ-results.

\smallskip
\smallskip
{\noindent\bf Corollary \ref{cor:enum_faq}.}
(Corollary \ref{theo:const-delay-enum}, Proposition \ref{prop:equiv_phi_join}).
\textit{
Given a cover $K$
of the result $\varphi(\inst{D})$ of an FAQ $\varphi$ over a database $\inst{D}$,
the tuples in the query result $\varphi(\inst{D})$ 
can be enumerated
with $\widetilde{\calO}(|K|)$ pre-computation time and $\calO(1)$ delay and extra space.
}

\section{Missing Proofs of Appendix \ref{sec:covers_equi_join}}
\label{sec:covers_equi_join_proof}

In case the signature mappings of an equi-join query are not
clear from the context, we write the signature mappings as a superscript 
to the query. Moreover, for a relation symbol 
 $R$ in an equi-join query with signature mappings 
 $(\lambda, \{\mu_{R}\}_{R \in \sign{Q}})$ and a database $\inst{D}$,
we write $\lambda(R)_{\inst{D}}$ to denote the relation 
assigned to the relation symbol 
$\lambda(R)$ in $\inst{D}$.  

\subsection{Proof of Proposition \ref{prop:equi_to_natural}}
{\noindent\bf Proposition \ref{prop:equi_to_natural}.}
\textit{Given an equi-join query $Q$, a decomposition $\calT$ of $Q$, and 
a database $\inst{D}$, there exist a natural join query 
$Q'$ and a database $\inst{D}'$ such that: $Q'(\inst{D}')=Q(\inst{D})$, 
$Q'$ has the decomposition $\calT$ and
can be constructed in time  $\calO(|Q|)$, and $\inst{D}'$ can be constructed in time 
$\calO(|\inst{D}|)$.
}

\smallskip
\smallskip
\smallskip

The query $Q$ has the form 
$\sigma_\psi(R_1\times\cdots\times R_n)$, where $\psi$ is a conjunction of equality conditions.  The relation symbols as well as all 
attributes occurring in the schemas of the relation symbols are pairwise distinct.
Let $(\lambda, \{\mu_{R_i}\}_{i \in [n]})$ be the signature mappings of $Q$.
Given an equivalence class ${\cal A}$ of attributes in $Q$, we let $\phi_{\cal A}=\bigwedge_{A_i,A_j\in{\cal A}} A_i=A_j$. Then, given the set $\{{\cal A}_j\}_{j\in[l]}$ of all equivalence classes in $Q$, the conjunction $\bigwedge_{j\in[l]}\phi_{{\cal A}_j}$ is the transitive closure $\psi^+$ of $\psi$ in $Q$.

\smallskip

\underline{\em Construction of $Q'$.} 
The query $Q'$ has one relation symbol $R'_i$ for each relation symbol $R_i$ in $Q$ such that $\sign{R'_i}=\sign{R_i}^+$. We thus have $Q'=R_1'\Join \cdots \Join R_n'$, where the equality conditions in the transitive closure of $\psi$ are now expressed by natural joins in $Q'$.

\smallskip

\underline{\em Construction of $\inst{D}'$.} 
For the sake of simplicity, we describe 
the construction of $\inst{D}'$ in three steps.  
\begin{itemize}
\item Construction of database $\inst{D}_1$:
The database $\inst{D}_1$ contains for 
each $R_i \in \sign{Q}$, a relation
$R_i^1$ which results from 
$\lambda(R_i)_{\inst{D}}$
by replacing each attribute $A$
by the attribute $B$ with 
$\mu_{R_i}(B) = A$.

\item Construction of database $\inst{D}_2$:
The database $\inst{D}_2$ consists of 
the relations $R_1^2, \ldots , R_n^2$ where each
$R_i^2$ results from $R_i^1$ as follows.
For each equality $A=B$
in $\psi^+$ such that $A,B \in \sign{R_i^1}$,
we delete 
in $R_i^1$
all tuples $t$ with $t(A) \neq t(B)$.
Note that such tuples $t$ cannot occur
in the projection of $Q(\inst{D})$ onto the schema 
of $t$. 

\item Construction of database $\inst{D}'$:
We obtain the database $\inst{D}'$
from $\inst{D}_2$ by replacing each relation 
$R_i^2$ by a relation $R_i'$
defined as follows. The relation $R_i'$ is a 
copy of $R_i^2$ extended with one new column for each attribute $A$ in 
 $\sign{R_i'} \backslash \sign{R_i}$ such that $\pi_{A} R_i' = 
 \pi_{B} R_i^2$ 
 for any attribute $B\in\sign{R_i}$ transitively equal to $A$.
\end{itemize} 

\smallskip

\underline{\em $\calT$ is a decomposition of $Q'$.} By construction, $Q$ and $Q'$ have the same set of attributes and thus the same equivalence classes of attributes. Moreover, the transitive closures of the schemas of relation symbols are identical: For any pair of relation symbols $R_i \in \sign{Q}$ and $R_i' \in \sign{Q'}$, 
it holds that $\sign{R'_i}^+=\sign{R'_i}=\sign{R_i}^+$. 
The hypergraphs of $Q'$ and $Q$ are thus the same as they have the same nodes, which are the attributes in $Q$ and $Q'$ respectively, and the same hyperedges, which are the transitive closures $\sign{R_i}^+$ and $\sign{R'_i}^+$ respectively. This means that the decomposition $\calT$ of $Q$ is also a decomposition of $Q'$.

\smallskip

\underline{\em $Q'(\inst{D}')=Q(\inst{D})$.} 
We define two further signature mappings  
$(\lambda^1, \{\mu^1_{R_i}\}_{i \in [n]})$
and 
$(\lambda^2, \{\mu^2_{R_i}\}_{i \in [n]})$
for $Q$.
The function  $\lambda^1$ maps each 
relation symbol $R_i$ in $Q$
to $R_i^1$. Moreover, each $\mu^1_{R_i}$ 
is an identity mapping on the attributes of $R_i$.
The function  $\lambda^2$ maps each 
relation symbol $R_i$ in $Q$
to $R_i^2$. Finally, $\mu^1_{R_i} = \mu^2_{R_i}$
for each $R_i \in \sign{Q}$.

The Database $\inst{D}_1$ results from 
$\inst{D}$ by, basically, making for each relation $R$ as many copies as 
the number of relation symbols in $Q$ mapped to $R$.
We obtain $\inst{D}_2$ from $\inst{D}_1$ by ruling out 
tuples which cannot be contained in (the projections of) the final result. 
Hence, it easily follows 
$Q^{(\lambda, \{\mu_{R_i}\}_{i \in [n]})}(\inst{D})$ 
$=$ 
$Q^{(\lambda^1, \{\mu^1_{R_i}\}_{i \in [n]})}(\inst{D}_1)$
$=$
$Q^{(\lambda^2, \{\mu^2_{R_i}\}_{i \in [n]})}(\inst{D}_2)$.
Thus, it remains to show 
 $Q^{(\lambda^2, \{\mu^2_{R_i}\}_{i \in [n]})}(\inst{D}_2)$
 $=$
 $Q(\inst{D})$.

 We first treat the special case when $Q$ is a Cartesian product, i.e., it does not contain any equality conditions. 
Then, $Q'=Q$ and each relation in $\inst{D}'$ is an exact copy of a relation in 
$\inst{D}_1$. Hence, $Q'(\inst{D}')=
Q^{(\lambda^2, \{\mu^2_{R_i}\}_{i \in [n]})}(\inst{D}_2)$ holds trivially.
We next consider the case when $Q$ has equality conditions.

We first show $Q'(\inst{D}')\subseteq Q^{(\lambda^2, \{\mu^2_{R_i}\}_{i \in [n]})}(\inst{D}_2)$. 
Assume there is a tuple $t$ that is contained in $Q'(\inst{D}')$. 
Then, $t=\Join_{i\in[n]} t_i$ is the natural join of tuples $t_i\in R'_i$. 
Let ${\cal A}$ be any equivalence class of attributes in $Q'$. By construction, whenever one of these attributes occur in the schema of a relation 
$R'_i$, so are the others. Furthermore, their values are the same in any tuple of 
$R'_i$. Since $t$ is a join of tuples $t_i$, it follows that all attributes in ${\cal A}$ have the same value in $t$ and therefore $\sigma_{\phi_{\cal A}}(t)=t$. 
This holds for all equivalence classes of attributes, so 
$\sigma_{\psi^+}(t)=t$ and thus $\sigma_{\psi}(t)=t$.
This means that $t\in Q^{(\lambda^2, \{\mu^2_{R_i}\}_{i \in [n]})}(\inst{D}_2)$.

We now show $Q^{(\lambda^2, \{\mu^2_{R_i}\}_{i \in [n]})}(\inst{D}_2) \subseteq Q'(\inst{D}')$. Assume there is a tuple $t$ that is in $Q^{(\lambda^2, \{\mu^2_{R_i}\}_{i \in [n]})}(\inst{D}_2)$. 
This means that $t = \bigtimes_{i\in[n]} t_i$ is a product of tuples 
$t_i\in R_i^2$, $\sigma_{\psi^+}(t)=t$ and in particular $\sigma_{\phi_{\cal A}}(t)=t$ for each equivalence class ${\cal A}$ in $Q$. 
We extend each tuple $t_i$ with values for all attributes in the class ${\cal A}$ whenever $\sign{t_i}\cap{\cal A}\neq\emptyset$. Let $t'_i$ be the extension of $t_i$. Then, $t = \Join_{i\in[n]} t'_i$. All attributes in ${\cal A}$ thus have the same value in $t'_i$. Since, by construction, the relation $R'_i$ is an extension of $R_i^2$ 
with same-valued columns for all attributes in ${\cal A}$ whenever $\sign{R_i^2}\cap{\cal A}\neq\emptyset$, it follows that $t'_i\in R'_i$. Thus, $t\in Q'(\inst{D}')$.

\smallskip

\underline{\em Construction time.}
The natural join query $Q'$ evolves from $Q$
by replacing the schema $S$ of each relation symbol 
by $S^+$. This can be done in time $\calO(|Q|)$. 

The database $\inst{D}_1$ evolves from 
$\inst{D}$ by duplicating each relation in 
$\inst{D}$ at most $|Q|$ times. Hence,  $\inst{D}_1$ 
can be constructed in linear time. 
We obtain $\inst{D}_2$ from $\inst{D}_1$ by deleting 
in each relation $R_i^1$ in $\inst{D}_1$,
each tuple tuple $t$ with $t(A)\neq t(B)$
and $A= B \in \psi^+$.
This deletion procedure  
can be realised via a single pass through the 
relations in $\inst{D}_1$ and requires, therefore, only linear time.
Likewise, each relation $R_i'$ in $\inst{D}'$ can be constructed from 
$R_i^2$ in $\inst{D}_2$ by a single pass through the tuples 
in $R_i^2$. For each tuple, we choose for each new attribute $A$ 
in $R_i'$ but not in $R_i^2$, an equivalent attribute in 
$R_i^2$ and copy its value to the $A$-column. 
Thus, the transformation from $\inst{D}_2$ to $\inst{D}'$
can also be done in linear time.

\nop{
\smallskip

\underline{\em Size of the construction.} 
The number of relations in $\inst{D}_1$ is the same as the number of relation
symbols in $Q$. Moreover, each relation in  $\inst{D}_1$ is a copy of 
a relation in $\inst{D}$. Hence,  $|\inst{D}_1| = \calO(|\inst{D}|)$.
The relations $R_i'$ in $\inst{D}'$ are copies of relations 
$R_i$ in $\inst{D}_1$ with additional columns. 
Each $R_i'$ in $\inst{D}'$ has a new column for each attribute which is 
not in 
$\sign{R_i}$ but equivalent to an attribute in 
$\sign{R_i}$. This transformation does not increase the size of the relations.
Altogether, we have $|\inst{D}'| = \calO|(\inst{D}|)$.
}

\subsection{Proof of Proposition \ref{prop:lower_bound_equi_join}}
{\noindent\bf Proposition \ref{prop:lower_bound_equi_join}.}
\textit{For any equi-join query $Q$ and any decomposition $\calT$ of $Q$,
there are arbitrarily large databases $\inst{D}$ such that 
each cover of $Q(\inst{D})$ over $\calT$
has size $\Omega(|\inst{D}|^{\fhtw{\calT}})$.
}

\smallskip
\smallskip
\smallskip

We will prove the following claim:

\smallskip
\smallskip

\noindent
\underline{\textit{Claim}}:
\textit{Given an equi-join query $Q$ and a decomposition $\calT$ of $Q$, 
there exist a natural join query 
$Q'$ that has the decomposition $\calT$ such that:  
$Q'$ can be constructed in time  $\calO(|Q|)$ and for each 
database $\inst{D}'$ there is a database $\inst{D}$ of size $\calO(\inst{D}')$
such that $|\pi_{B}Q(\inst{D})| \geq |\pi_{B} Q'(\inst{D}')|$ for each $B \in \sign{\calT}$.
}

\smallskip
\smallskip

Using this claim, the result of the proposition can be derived straightforwardly.
Given an equi-join query $Q$, we first construct the natural join query 
as promised in the claim. 
By Theorem 
\ref{theo:general_size_bounds}(ii),  
there are arbitrarily large databases $\inst{D}'$ such that 
each cover of $Q'(\inst{D}')$ over $\calT$
has size $\Omega(|\inst{D}'|^{\fhtw{\calT}})$.
Given such a database $\inst{D}'$, it follows from Proposition \ref{prop:size_bounds_hypergraph}, that 
$\Sigma_{B \in \sign{\calT}}|\pi_BQ'(\inst{D}')| = 
\Omega(|\inst{D}'|^{\fhtw{\calT}})$, hence, 
$\max_{B \in \sign{\calT}}\{|\pi_BQ'(\inst{D}')|\}
=  \Omega(|\inst{D}'|^{\fhtw{\calT}})$. 
By our claim, the database $\inst{D}'$ can be converted into 
a database $\inst{D}$ of size $\calO(|\inst{D}'|)$
such that $|\pi_BQ(\inst{D})| \geq |\pi_BQ'(\inst{D}')|$
for each $B \in \sign{\calT}$. 
By Proposition \ref{prop:size_bounds_hypergraph} 
(adapted to equi-join queries), each cover of $Q(\inst{D})$
over $\calT$
must have size at least $\max_{B \in \sign{\calT}}\{|\pi_BQ(\inst{D})|\}$.
Since $\max_{B \in \sign{\calT}}\{|\pi_BQ'(\inst{D}')|\}$
$=$  $\Omega(|\inst{D}'|^{\fhtw{\calT}})$ and 
$\max_{B \in \sign{\calT}}\{|\pi_BQ(\inst{D})|\}$ $\geq$ 
$\max_{B \in \sign{\calT}}\{|\pi_BQ'(\inst{D}')|\}$, we conclude that
each cover of $Q(\inst{D})$
over $\calT$ is of size $\Omega(|\inst{D}'|^{\fhtw{\calT}})$ $=$ $\Omega(|\inst{D}|^{\fhtw{\calT}})$.

\smallskip
\smallskip

We turn towards the proof of our claim. Let $(\lambda, \{\mu_{R_i}\}_{i \in [n]})$
be the signature mappings of $Q$.

\underline{\em Construction of $Q'$.}
The natural join query $Q'$ is constructed exactly as in the proof of Proposition
 \ref{prop:equi_to_natural}.

\smallskip

\underline{\em Construction of $\inst{D}$.} 
Given a database $\inst{D}'$, we describe the construction
of $\inst{D}$
in three steps.  
\begin{itemize}
\item Construction of database $\inst{D}_1$: 
For each equivalence class $\calA \subseteq \bigcup_{i \in [n]}\sign{R_i'}$, 
let $f_{\calA}$ be an injective function mapping tuples over $\calA$
to fresh data values not occurring 
in $\inst{D}'$. Moreover, let $f$ be a function mapping tuples $t$ with 
$\sign{t} \subseteq \bigcup_{i \in [n]}\sign{R_i'}$ and 
$\sign{t} = \sign{t}^+$ to 
tuples $t'$ with $\sign{t'} = \sign{t}$ as follows. 
For each attribute $A \in \sign{t'}$
from some equivalence class $\calA$, it holds $t'(A) = f_{\calA}(\pi_{\calA}t)$.
From each relation $R_i' \in \inst{D}'$, we construct 
a relation $R_i^1$ where each tuple $t$ is replaced
by $f(t)$. 
We define $\inst{D}_1 = \{R_{i}^1\}_{i \in [n]}$.

\item Construction of database $\inst{D}_2$:
From each relation $R_{i}^1 \in \inst{D}_1$ we design 
a relation $R_{i}^2$ by performing the following procedure. 
We first project away 
all columns of attributes not included in $\sign{R_i}$. Then, 
we rename each attribute $A$ in the resulting relation 
by $\mu_{R_i}(A)$. Let 
$\inst{D}_2 = \{R_{i}^2 \}_{i \in [n]}$.
  
\item Construction of database $\inst{D}$:
We obtain database $\inst{D}$ from $\inst{D}_2$
as follows.
For each maximal set 
$\{R_{i_1}, \ldots ,R_{i_k}\} \subseteq \sign{Q}$ such that all
 $R_{i_j}$ are mapped to the same relation symbol 
 $\lambda(R_{i_j})=R$, we replace 
the relations $R_{i_1}^2, \ldots ,R_{i_k}^2$ by a single relation 
$\bigcup_{j \in [k]} R_{i_j}^2$ with relation symbol $R$.
\end{itemize}

 \smallskip

\underline{$\calT$ is a decomposition of $Q'$.}
This follows from the proof of Proposition \ref{prop:equi_to_natural}.
 
\smallskip

\underline{\em $|\pi_{B}Q(\inst{D})| \geq |\pi_{B}Q'(\inst{D}')|$
for each $B \in \sign{\calT}$.}
\nop{We define two further signature mappings  
$(\lambda^1, \{\mu^1_{R_i}\}_{i \in [n]})$
and 
$(\lambda^2, \{\mu^2_{R_i}\}_{i \in [n]})$
for $Q$.
The function  $\lambda^1$ maps each 
relation symbol $R_i$ in $Q$
to $R_i^1$. Moreover, each $\mu^1_{R_i}$ 
is an identity mapping on the attributes of $R_i$.
The function  $\lambda^2$ maps each 
relation symbol $R_i$ in $Q$
to $R_i^2$. Finally, $\mu^1_{R_i} = \mu^2_{R_i}$
for each $R_i \in \sign{Q}$.
}
Our proof contains three steps.
\begin{enumerate}[(1)]
\item We show that $|\pi_{B}(\Join_{i \in [n]}R_i^1)| \geq |\pi_{B}Q'(\inst{D}')|$
for each $B \in \sign{\calT}$. Let $B \in \sign{\calT}$.
Since $B = B^+$, the function $f$ is defined for all tuples
in $\pi_{B}Q'(\inst{D}')$. Moreover, for two distinct tuples 
$t_B, t_B' \in \pi_{B}Q'(\inst{D}')$, the tuples $f(t_B)$ and $f(t_B')$
are distinct, too. Hence, it suffices to show that for each
$t_B \in \pi_{B}Q'(\inst{D}')$, we have $f(t_B) \in \pi_{B}(\Join_{i \in [n]}R_i^1)$.
Let $t_B \in \pi_{B}Q'(\inst{D}')$. It follows that there is a tuple 
$t \in Q'(\inst{D}')$ with $t_B = \pi_B t$. By definition of $Q'(\inst{D}')$, 
it must hold $\pi_{\sign{R_i'}} t \in R_i'$ for each $i \in [n]$.
We have  $\sign{R_i'} = \sign{R_i'}^+$ for each $i \in [n]$. Thus, 
it holds  
$f(\pi_{\sign{R_i'}} t) \in R_i^1$ for each $i \in [n]$.
This is equivalent to saying $\pi_{\sign{R_i^1}} f(t) \in R_i^1$ 
for each $i \in [n]$. 
By definition of $\Join_{i \in [n]}R_i^1$, this implies that 
$f(t) \in \Join_{i \in [n]}R_i^1$. Thus, $f(t_B) = \pi_B f(t) \in 
\pi_B(\Join_{i \in [n]}R_i^1)$. 

\item Let $\lambda'$ be a function that maps each relation symbol $R_i$ 
in $\sign{Q}$ to 
$R_i^2$. We show that $\Join_{i \in [n]}R_i^1 \subseteq Q^{(\lambda',\{\mu_{R_i} \}_{i \in [n]})}(\inst{D}_2)$. To this end, let 
$t \in \Join_{i \in [n]}R_i^1$. It follows that $\pi_{\sign{R_i^1}}t \in R_i^1$
for each $i \in [n]$. By the construction of $\inst{D}_2$, 
it holds $\pi_{\sign{R_i^2}}t \in R_i^2$
for each $i \in [n]$. Furthermore, by the 
construction of $\inst{D}_1$, for each equivalence class
$\calA$ and all attributes $A,B \in \calA$, we have $t(A) = t(B)$.
Thus, $t \in Q^{(\lambda',\{\mu_{R_i} \}_{i \in [n]})}(\inst{D}_2)$.  

\item We show that $Q^{(\lambda',\{\mu_{R_i} \}_{i \in [n]})}(\inst{D}_2) \subseteq Q^{(\lambda,\{\mu_{R_i} \}_{i \in [n]})}(\inst{D})$.
We recall that database $\inst{D}$ results from $\inst{D}_2$ by 
replacing each maximal set $R_{i_1}^2, \ldots ,R_{i_k}^2$
of relations with $\lambda(R_{i_1}) = \ldots = \lambda(R_{i_k}) = R$,
by the relation 
$\bigcup_{j \in [k]} R_{i_j}^2$ with the relation symbol $\lambda(R_{i_1})$.
Observe that the result of 
$\sigma_{\psi}(R_{i_1} \times  \ldots \times R_{i_k})(\{R_{i_1}^2, \ldots ,R_{i_k}^2\})$ (under signature mappings $(\lambda',\{\mu_{R_i} \}_{i \in [n]})$)
must be included in the result of 
$\sigma_{\psi}(R_{i_1} \times  \ldots \times R_{i_k})(\bigcup_{j \in [k]} R_{i_j}^2)$ (under signature mappings $(\lambda,\{\mu_{R_i} \}_{i \in [n]})$).
By generalising this insight, we obtain that 
every tuple from  
$Q^{(\lambda',\{\mu_{R_i} \}_{i \in [n]})}(\inst{D}_2)$ must be included in 
$Q^{(\lambda,\{\mu_{R_i} \}_{i \in [n]})}(\inst{D})$.
\end{enumerate} 

\noindent
By (1), $|\pi_{B}(\Join_{i \in [n]}R_i^1)| \geq |\pi_{B}Q'(\inst{D}')|$
for each $B \in \sign{\calT}$.
By (2) and (3), $\Join_{i \in [n]}R_i^1 \subseteq
Q^{(\lambda,\{\mu_{R_i} \}_{i \in [n]})}(\inst{D})$. We conclude 
that $|\pi_{B}Q^{(\lambda,\{\mu_{R_i} \}_{i \in [n]})}(\inst{D})| \geq
 |\pi_{B}Q'(\inst{D}')|$
for each $B \in \sign{\calT}$.

\smallskip

\smallskip

\underline{\em Construction time for $Q$.}
It follows from the proof of Proposition \ref{prop:equi_to_natural}
that $Q'$ can be constructed in time $\calO(|Q|)$.

\smallskip

\underline{\em Size of $\inst{D}$.} 
Since $f$ is a bijective mapping
and $\inst{D}_1$ is obtained  from $\inst{D}'$
by replacing tuples $t$ by $f(t)$,
we have 
 $|\inst{D}_1| = |\inst{D}'|$. 
 As $\inst{D}_2$ results from 
$\inst{D}_1$ by taking projections of relations, 
the size of $\inst{D}_2$ cannot be larger than the size
of $\inst{D}_1$. Database $\inst{D}$ results from 
$\inst{D}_2$ by taking unions of relations. Thus,  
the number of tuples in $\inst{D}$ cannot be more 
than the number of tuples in $\inst{D}_2$.
Altogether, we have $|\inst{D}| = \calO(|\inst{D}'|)$.

\nop{
\smallskip

\underline{\em Construction time.}
We can turn database $\inst{D}'$ into database $\inst{D}_1$
by performing the following procedure for each equivalence 
class $\calA$. We first sort all relations in  $\inst{D}'$ on the 
attributes in $\calA$. Then, in a 
subsequent pass through all relations, we construct the function $f_{\calA}$
which maps each distinct combination of $\calA$-values occurring in at least 
one tuple to a distinct fresh data value. During a final pass, 
we replace in each tuple $t$ in the relations the value of each attribute from 
$\calA$ by $f_{\calA}(\pi_{\calA}t)$. Hence, the construction 
time for $\inst{D}_1$ is $\softO(\inst{D}')$.

We can take any projection of any relation in $\inst{D}_1$
by first sorting the relation, then, performing a subsequent pass
in which we take the projection of each single tuple, and finally, eliminating 
duplicates in the resulting relation. Since the relation is sorted,
duplicate elimination can be done in a  single pass. 
Renaming of attributes requires only constant time.  
 Thus,  $\inst{D}_2$ can be constructed in time $\softO(\inst{D}_1)$.

Recall that $\inst{D}$ results from $\inst{D}_2$ by taking unions of sets 
of relations. We can take the union of a set of relations 
by first merging the relations, then sorting the resulting multi-set
and finally eliminating duplicates. Hence, the construction time for 
$\inst{D}$ is $\softO(\inst{D}_2)$.

Since $|\inst{D}'| \geq |\inst{D}_1|  \geq |\inst{D}_2|$, the overall construction time
for $\inst{D}$ is $\softO(|\inst{D}'|)$.
}
\end{document}